\documentclass[11pt]{article}
\usepackage{amsmath,amsthm,amsfonts,amscd,eucal,latexsym,amssymb,amsbsy,dsfont,mathrsfs,yhmath}
\usepackage{geometry} 
\def\sp{\hskip -5pt} 
\def\spa{\hskip -3pt}

\def\bk{{\boldsymbol k}}

\def\bx{{\boldsymbol x}}

\def\cF{{\ca F}}

\def\cH{{\ca H}}
\def\cD{{\ca D}}

\def\cO{{\ca O}}

\def\cF{{\ca F}}
\def\cP{{\ca P}}

\def\sS{{\mathsf S}}
\def\sK{{\mathsf K}}

\def\sH{\mathsf{H}}
 
\def\cW{\mathscr{W}}

\def\mA{\mathscr{A}}
\def\mS{\mathscr{S}}

\def\bC{{\mathbb C}}           

\def\bN{{\mathbb N}}
\def\bM{{\mathbb M}}

\def\bR{{\mathbb R}}

\def\bS{{\mathbb S}}

\def\beq{\begin{eqnarray}}
\def\eeq{\end{eqnarray}}
\newcommand{\ca}[1]{{\cal #1}}         

\def\ph{\varphi}

\def\supp{\mbox{supp}}
\def\Dom{\mbox{Dom}}

\newcommand{\Int}{\mathrm{Int}}
\newcommand{\Imm}{\mbox{Im}}
\newcommand{\Rea}{\mbox{Re}}

\begin{document}
\title{\Huge Modular dynamics in diamonds}
\author{{\Large Romeo Brunetti$^{1,a}$, Valter Moretti$^{1,2,b}$} 
\\\null\\
   \noindent$^1$ Dipartimento di Matematica -- Universit\`a di Trento
\\
 \noindent$^2$ Istituto Nazionale di Fisica Nucleare -- Gruppo Collegato di Trento\\
 via Sommarive 14,  
 I-38050 Povo (TN), Italy. 
 \\
\small  $^a$brunetti@science.unitn.it, $^b$moretti@science.unitn.it}

 \theoremstyle{plain}

  \newtheorem{definition}{Definition}[subsection]

  \newtheorem{theorem}[definition]{Theorem}

  \newtheorem{proposition}[definition]{Proposition}

  \newtheorem{corollary}[definition]{Corollary}

  \newtheorem{lemma}[definition]{Lemma}

\theoremstyle{definition}

 \newtheorem{remark}[definition]{Remark}

  \newtheorem{example}[definition]{Example}

\maketitle

\begin{abstract}
We investigate the relation between the actions of Tomita-Takesaki modular operators for local 
von Neumann algebras in the vacuum for free massive and massless bosons in four dimensional 
Minkowskian spacetime. In particular, we prove a long-standing conjecture that says that the 
generators of the mentioned actions differ by a pseudo-differential operator of order zero. 
To get that, one needs a careful analysis of the interplay of the theories in the bulk  and at the 
boundary of double cones (a.k.a. diamonds). 
After introducing some technicalities, we prove the crucial result that the vacuum state 
for massive bosons in the bulk of a double cone restricts to a KMS state at its boundary, 
and that the restriction of the algebra at the boundary does not depend anymore on the mass. 
The origin of such result lies in a careful treatment of classical Cauchy and Goursat problems for 
the Klein-Gordon equation as well as the application of known general mathematical techniques, 
concerning the interplay of algebraic structures related with the bulk and algebraic structures related 
with the boundary of the double cone, arising from quantum field theories in curved spacetime. 
Our procedure gives explicit formulas for the modular group and its generator in terms of integral 
operators acting on symplectic space of solutions of massive Klein-Gordon Cauchy problem. 
\end{abstract}
\newpage
\tableofcontents

\section{Introduction}\label{sec1}
This paper deals with the task of finding explicit expressions of the (generators of the) Tomita-Takesaki modular groups of automorphisms of von Neumann algebras 
 localized in double cones in $4$-dimensional  Minkowski spacetime for a free massive scalar field theory in the (cyclic and separating) vacuum state. In the following 
discussion we shall deal exclusively with Minkowski spacetime and hence we will not mention 
 it explicitly anymore.
 
Algebraic quantum field theory \cite{Haag} is a particularly fruitful approach for unveiling structural features of specific (class of) theories and, as
 such, it has been largely applied to study spin-statistics, PCT theorems and superselection structures. In this respect, one of the milestones has
  been the connection of Tomita and Takesaki modular theory\footnote{For short lucid introductions we refer to the papers of Guido \cite{Guido},  
  Lled\'o \cite{Lledo1} and Summers \cite{Summers}. For more extensive considerations we cannot help better than suggesting the reading of Borchers' 
  paper \cite{Borchers}.} \cite{TT} to quantum statistical mechanics, via KMS conditions \cite{BR1, BR2}. Modular theory has been connected with geometry-dynamics
   and thermodynamics, as shown in the pioneering papers of Bisognano-Wichmann \cite{BW1,BW2} and Sewell \cite{Se}. The first gave the interpretation of
    the modular group of automorphisms of von Neumann algebras --generated by Wightman scalar fields
    in the vacuum Poincar\'e-invariant state-- localized in wedge regions  as boost transformations leaving the wedges fixed; whereas  the second offered a physical interpretation of the Bisognano-Wichmann analysis  
    in terms of the Unruh and Hawking effects. These ideas have been further developed in various directions, which especially emphasize the geometric interpretation 
    of the action of the modular groups. Let us recall some of them. 

In the case of massless quantum field theories Buchholz \cite{Bu} proved that the generator of the modular group, w.r.t. the vacuum state, 
associated with a von Neumann algebra localized in the (forward) light-cone is related to the scaling transformation.  Hislop and Longo \cite{HL} 
showed that, for free scalar field theories invariant under the (suitably defined) group  of conformal transformations, the modular groups associated 
with classes of regions as light-cones, wedges and double cones, are all related to the generators of the geometric transformations leaving all respective 
regions fixed. In particular, they are all related via suitable elements of the conformal group. A purely algebraic version of these results has been presented 
in \cite{BGL1, BGL2}, using the results of a landmark paper by Borchers \cite{Bor}. Reconstruction of some spacetime and their 
symmetries can also be done, mainly using modular conjugations, as shown by Buchholz, Summers and collaborators \cite{BS1,BS2}, 
and, similarly, one can define free theories solely in terms of the modular data  \cite{BGL3}, basing the approach on the abstraction of the Bisognano-Wichmann 
correspondence and the representations of the Poincar\'e group. In the case of conformal field theories on the circle, modular groups (always related to a 
cyclic and separating vacuum state) continues to have a geometrical flavor when related to single intervals, but, as shown recently in \cite{MLR} also in 
case the localization has been weaken to several disjoint intervals. The last paper was inspired by recent results of Casini and Huerta \cite{CH} who in turn 
were able to put on solid ground the computation of the (resolvent of the) generator of the modular group (for massless free Fermi fields in two dimensions) 
developed abstractly in the elegant paper of Figliolini and Guido \cite{FG} (see also \cite{BJL,Saffary1}). However, it appears clear that the geometric behavior 
of the modular groups is tightly related either to the massless condition or to a particular choice of the state. Indeed, in the massive case for localizations 
different from wedges, for instance on double cones, one should not expect such a geometric behavior. Moreover, even in wedges localizations, it has
 been proved by Borchers and Yngvason \cite{BY} that other choice of states, for instance KMS states,  would certainly break down the geometric connection. 
 This last can only be retrieved in particular situations like near the apex of a double cone or the edge of the wedge.  The
  application of the milestone contributions by Wiesbrock about half-sided modular conditions
\cite{Wie1,Wie2,Wie3} was crucial to achieve those results.

It is the main goal of this paper to start a direct attack to one of the main questions left open  in the last years (see \cite{SW} for more): What is the explicit 
form of the generator of the modular group for massive (scalar) theories in the vacuum localized in a double cone? With this paper, we are able to give a
 precise answer only in the case of free theories, while the ambitious task of dealing with more general theories is left to future investigations. 

A well-known conjecture, sometimes referred to as  ``Fredenhagen's conjecture," says that the generator of the modular group for the massive
 free scalar theory (in the vacuum) for a double cone should differ from the massless one by a pseudodifferential operator of order zero. A possible
  strategy to prove this conjecture appeared in several places, for instance in \cite{SW} and \cite{Saffary2}, and relies upon the idea that, if  
  we consider the vacuum for the massless theory as a positive functional (state) in the (Fock) Hilbert space of the massive one, then the modular groups of both 
  theories associated with this last state should coincide. Notice that, for the massive case, that vacuum would be the ``wrong'' one. However, by one of the 
  celebrated theorems of Connes \cite{Connes}, the modular automorphisms of the von Neumann algebra for the  massive case referred to the ``wrong'' and correct vacua, 
  would differ only by a cocycle, i.e. by the action of a (continuous) family of unitary operators in the vacuum Hilbert space. However, so far, no one was able to turn this ideas into a real proof, not even the authors of the present paper who, instead, took a radically different route.

\subsubsection*{Main ideas and structure of the work} 

After having defined the general geometric scenario within
Section \ref{sec2}, we will tackle the main problem to give an explicit
representation of the modular group
of the free massive Klein-Gordon theory of a double cone as follows.
Instead of relying upon choices of vacua and construction of cocycles,
our main idea is to compare the massive free scalar theory
in the double cone with its shadow at the boundary.
This simple idea has been a source of many interesting recent results of
quantum field theories in curved spacetime \cite{DMP1,Mor1,Mor2,DMP2,DMP3,DMP4,DPP,P}, and
produced powerful tools for the comparison of theories in the boundary-bulk
correspondence.

 Generally speaking, one starts by considering the  algebra of
observables $\cW_m(D)$
 (the Weyl $C^*$-algebra of the free Klein-Gordon field with mass $m\geq
0$) localized in the interior of a region $D$ of a globally hyperbolic
spacetime, whose boundary
 contains a $3$-dimensional light-like submanifold with the structure of
a conical surface $V$. In the case considered in this paper, $D$ is
nothing but a double
 cone  in Minkowski spacetime and $V$ is the lower light-cone forming
part of the boundary of $D$.
 Under some hypotheses (e.g. see \cite{Mor1,DMP2} for asymptotically
flat and cosmological backgrounds), exploiting the symplectic structure
of the space of solutions of Klein-Gordon equation (KG from now on),
 it is possible to define a $C^*$-algebra of observables $\cW(V)$
localized on $V$, together with an embedding map ($C^*$-algebra
isometric homomorphism) $\ell^{VD}_m : \cW_m(D) \to \cW(V)$. Remarkably, the boundary algebra
$\cW(V)$ does not depend on the mass $m$ anymore, since the information
about $m$ is completely  encoded in the embedding $\ell^{VD}_m$ as we shall prove in Section \ref{sec3}.

  After some preparations given in Section \ref{sec3} concerning the characteristic Cauchy problem
(also known as the {\em Goursat problem}) of the KG equation with data
assigned on $V$,  the further step in our construction is to establish, in Section \ref{sec4}, the existence of a
 state $\lambda : \cW(V) \to \bC$ that induces the (restriction to $D$ of the) vacuum $\omega_m : \cW_m(D)\to \bC$ 
 through $\ell^{VD}_m$. In other words $\omega_m = \lambda\circ\ell^{VD}_m$. 
Once again, as a remarkable feature, $\lambda$ turns out to be independent from $m$.  The idea is
then to pass to the von Neumann algebras, $\pi_{m}(\cW_m(D))''$,
  $\pi_\lambda(\cW(V))''$ respectively associated with the
   GNS representations $(\cH_m, \pi_m, \Psi_m)$
   of $\omega_m$ and  $(\cH_\lambda, \pi_\lambda, \Psi_\lambda)$ of
   $\lambda$ and to study the interplay of the corresponding modular
groups,  referred to the respective cyclic and separating vectors
$\Psi_m$ and $\Psi_\lambda$.
   As we shall prove in Section \ref{sec5}, though $\cW_m(D)$ is smaller
than $\cW(V)$ -- that is $\ell_m^{VD}$ is not surjective -- the algebras
coincide when
   promoted to von Neumann algebras.
This is because $\ell_m^{VD}$ turns out to be implemented by a unitary
$Z_m: \cH_m  \to \cH_\lambda$
 which preserves the cyclic vectors, $Z_m \Psi_m = \Psi_\lambda$,
identifying the von Neumann algebras: $Z_m \pi_{m}(\cW_m(D))'' Z_m^{-1}
= \pi_{\lambda}(\cW(V))'' $.
 This technically complicated result will be achieved by a careful
analysis of the characteristic Cauchy problem and by enlarging in a canonical manner the initially
defined symplectic space of solutions of KG equation
following the analysis performed in Section \ref{sec3}.

To go on towards the final result, the crucial observation is that the
modular group of the boundary has a
geometric meaning related to the conformal Killing field studied in the massless case by Hislop and Longo in \cite{HL}.
This remark is used  in Section \ref{sec6} and permits to construct  the group of symplectic
isomorphisms of the space of solutions of KG equation that corresponds, through quantization
procedure, to the group of modular automorphisms of $(\pi_{m}(\cW_m(D))'',\Psi_m)$ in the one-particle Hilbert space.
In the massless case the one-parameter symplectic group will match that presented in \cite{HL}. Furthermore, the same procedure will give rise to an explicit formula for the infinitesimal generator
of the modular group (represented in the symplectic space of KG solutions). The difference of the generator for $m>0$ and the analog for
$m=0$  turns out to be a pseudodifferential operator of class $L^0_{1,1}$.

\section{Geometric features}\label{sec2}
Here and in the following $(\bM,g)$ is the four dimensional Minkowski spacetime assumed to be time-oriented,
 $g$ the flat metric with signature $-,+,+,+$ and $\nabla$ 
is the covariant derivative associated with the metric $g$. 
We fix a preferred Minkowskian frame used throughout  and  adopt the notation $x=(t, \bx) \in \bR \times \bR^3$
for the coordinates of the frame.
In the rest of the paper we shall use also the following {\em light-cone coordinates} on $\bM$:
 \beq 
 u \doteq  \frac{t+ ||\bx||}{2} \in \bR\:, \quad 
v\doteq \frac{t- ||\bx||}{2} \in \bR\:, \quad \omega \doteq  \frac{\bx}{||\bx||} \in \bS^2\ ,
\label{nullcoordinates}
\eeq
where $u-v \geq 0$, and $||\bx||$ stands for the euclidean norm in $\bR^3$.

Another useful tool will be the {\em signed squared geodesic distance} $\sigma :\bM\times\bM\rightarrow \bR$ of pairs of points
in $\bM$ with coordinates $x=(t,\bx)$ and $x'=(t',\bx')$:
\beq 
\sigma(x,x') \doteq -(t-t')^2 + ||\bx-\bx'||^2\label{distance}\:.
\eeq 

Concerning the {\em causal structure} of spacetimes, we adopt the definitions as in \cite{Wald1}. In particular,
 if $K\subset \bM$, then
$J^+(K)$ and  $J^-(K)$ denote, as subsets of $\bM$, the {\em causal future} and {\em causal past} of $K$, respectively. Restricting to some open set $A \subset \bM$,
and if $K\subset A$,
$J^\pm(K;A)$ denote the analogs referred to $(A,g\spa\restriction_A)$ viewed as a spacetime on its own right.

\subsection{Relevant properties of conformal Killing vector fields of $\bM$}\label{secconfflow}
If $Y$ is a {\em conformal Killing field} \cite{schottenloher} of Minkowski spacetime,
 we indicate by $\Upsilon^Y$ the {\em local} one--parameter group of diffeomorphisms generated by $Y$ (sometimes simply called
 the {\em flow} of $Y$). It is the smooth map $$\Upsilon^Y: \cO\to \bM\ ,$$  where $\cO\subseteq \bR\times\bM$ an open subset 
 containing the set $\{0\}\times\bM$, and which satisfies the relations
\begin{itemize}
\item[$(a)$] $\Upsilon^Y(0,x)=x$ ;
\item[$(b)$] 
$\Upsilon^Y(\tau'+\tau,x)=\Upsilon^Y(\tau',\Upsilon^Y(\tau,x))$, 
whenever both sides are defined.
\end{itemize}

The induced action $\beta^Y_\tau$ on 
(scalar) functions $f$ is individuated by the requirement
\beq 
(\beta^Y_\tau f)(\Upsilon^Y(\tau,x))=  J^Y_{\tau}(x) ^{-1/4}f(x)\label{conf}\:,
\eeq
where $J^Y_\tau$ is the Jacobian {\em in Minkowskian coordinates} of the the map  $x \mapsto \Upsilon^Y(\tau, x) $ and it is assumed that it does not vanish. 
The exponent $-1/4$ may be fixed differently \cite{schottenloher}, however our choice will turn out to be useful later when discussing properties 
of  the massless Klein-Gordon equation.
In this paper we shall be concerned 
with the {\em infinitesimal generator} $\gamma^{Y}$ of $\beta_\tau^{Y}$:
\beq
(\gamma^{Y} f)(x) \doteq  \frac{d}{d\tau} \bigl\arrowvert_{\tau=0} (\beta^Y_\tau f)(x) = -Y(f)(x) -\frac{1}{4} (\nabla_a Y^a) f(x)\:, \label{delta}
\eeq
where, as in the rest of the paper, we make use of the Einstein's summation convention and from now on 
latin indices $a,b,c,\dots$ range from $1$ to $4$.

There is an interesting interplay between the squared distance $\sigma$ in $\bM$ and the transformations induced by the conformal Killing vectors: 
\beq \sigma(\Upsilon^Y(\tau,x), \Upsilon^Y(\tau,x')) = J_\tau^{Y}(x)^{1/4} J_\tau^{Y}(x')^{1/4} \sigma(x,x')\eeq
(see for instance \cite{Haag,CFT}, collecting all the results for every subgroup of conformal transformation of $\bM$). Taking the derivative at $s=0$,  the useful relation arises:
\beq 
Y_x(\sigma)(x,x')  +  Y_{x'} (\sigma)(x,x') =  \frac{1}{4} \left(\nabla_aY^a(x) + \nabla_a Y^a(x')\right) \sigma(x,x')
\label{trick0}
\eeq
where the notation $Y_x$ means that the action of $Y$ is on the variable showed as a lower index.
 The previous relation implies that, if $F$ is a differentiable function on the real line
and $\gamma^Y$ is defined in (\ref{delta}):

\begin{align}
\gamma_x^{Y} F(\sigma(x,x')) + &\gamma_{x'}^Y F(\sigma(x,x'))\nonumber\\
\qquad &= \frac{1}{4} \left( \nabla_a Y^a(x) + \nabla_a Y^a(x')\right) (F(\sigma(x,x')) +\sigma 
F'(\sigma(x,x'))) \label{trick}\ .
\end{align}

\subsection{The standard double cones $D(p,q)$}

{\em Double cones} in $\bM$, also known as {\em diamonds}, are open regions generated by the choice of two points, $p$ and $q$, such that $q$ lies in the chronological future of $p$, and are defined by the position
\[
D(p,q)\doteq \Int(J^+(p)\cap J^-(q))\ .
\]
Given $D(p,q)$, it is always possible to construct a Minkowskian reference frame, said to be {\em adapted to} $D(p,q)$, such that
 $p$ and $q$ stay on the $t$-axis, $q$ in the future of  $p$ which, in turn, coincides to the origin of the coordinates. 
That coordinate frame is determined up to a spatial $3$-rotation and,  referring to those coordinates,  $D(p,q)$  takes the canonical form, where $a \doteq \sqrt{-\sigma(p,q)}/2$ is the {\em radius} of the double cone:  
\beq
D(p,q) \equiv  \{ (t,\bx) \in \bR^4\:\:|\:\:  |t-a| + ||\bx|| < a\}  \:.\label{doublecone}
\eeq 
 The boundary $\partial D$ is decomposed into {\em three} disjoint sets:
 A couple of lightlike conical $3$-surfaces,  $V(p,q)$ and  $W(p,q)$ with tips $p$ and $q$ respectively, and a 
 $2$-sphere, $C(p,q)$, defining the common base of $V(p,q)$ and $W(p,q)$. Using light-coordinates as 
in \eqref{nullcoordinates}:
\begin{align}
V(p,q) &\doteq  \{ (\omega,u,v) \in\bS^2\times \bR^2  \:|\: v=0\:, u \in [0,a) \}\:, \label{boundary}\\
W(p,q) &\doteq  \{ (\omega,u,v) \in\bS^2\times \bR^2  \:|\: u=a\:, v \in (0,a] \} \:,\\
C(p,q)  &\doteq  \{ (\omega,u,v) \in  \bS^2\times \bR^2 \:|\:  v=0\:, u = a \}\:.
   \label{boundary2}
\end{align}
The embedding of the cylinder $\bS^2 \times (0,a)$ into $\bM$:
\beq
\kappa : \bS^2 \times (0,a) \ni (\omega,u) \mapsto (u, u\omega) \in \bM
\eeq
realizes a {\em blow up} of the cone $V(p,q)$, removing its lower tip $p$. 

 To conclude, we notice that $(D(p,q), g\!\!\!\restriction_{D(p,q)})$ is globally hyperbolic.
Indeed, it is strongly causal, since $\bM$ is such and, for   $r,s \in D(p,q)$, the intersection of causal 
sets  $J^+(r;D(p,q)) \cap J^-(s;D(p,q))$ is compact since it coincide with the analogous set
$J^+(r) \cap J^-(s)$ referred to the whole spacetime $\bM$, which is compact in turn.

 \subsection{A relevant conformal Killing vector for $D$}
For a fixed double cone $D(p,q)$, taking into account the general expression \cite{schottenloher} of smooth conformal Killing fields everywhere defined in $\bM$, one easily sees that,
  up to a constant nonvanishing factor, there is a unique (nonvanishing) conformal Killing field $X^{(p,q)}$ of $\bM$ such that 
  its flow leaves  $V(p,q) \cup C(p,q)$ fixed,  $X^{(p,q)}$
is nonspacelike and invariant under $3$-rotations on $V(p,q) \cup C(p,q)$. It is:
\beq
X^{(p,q)} \doteq (t-a) \bx \cdot \nabla_\bx + \left(\frac{t^2 + \bx^2}{2}-at\right)\partial_t \quad \mbox{so $\nabla_b{X^{(p,q)}}^b = 4(t-a)$} \label{HL}\:,
\eeq
where we have adopted a Minkowskian coordinate frame adapted to $D(p,q)$. It will be useful to have also the form of the field in light-cone coordinates, namely,
\[
X^{(p,q)}  = u(u-a) \partial_u + v(v-a) \partial_v\ .
\]

Notice that $X^{(p,q)}$ is tangent to both $V(p,q)$ and $W(p,q)$ and it vanishes on $p$, $q$ and $C(p,q)$, therefore its flow leaves $D(p,q)$, $C(p,q)$,  $V(p,q)$, $W(p,q)$, $p$, $q$  separately fixed.  
As $\overline{D(p,q)} = D(p,q) \cup \partial D(p,q)$  is compact,  the  one-parameter group 
of diffeomorphisms generated by $X^{(p,q)}$ is global on $\overline{D(p,q)}$, i.e. its parameter ranges to the 
full real line $\bR$ when the orbit starts from a point in $\overline{D(p,q)}$ and the composition of flows in the right-hand side of  (b) in Section \ref{secconfflow} is always defined if $x \in \overline{D(p,q)}$ and $\tau,\tau'\in \bR$.

\begin{remark} In the following,  without loss of generality, we shall deal with a preferred double cone, indicated by $D$, assuming that $a=1$.
In that case, $p$, $q$,  $C(p,q)$,  $V(p,q)$, $W(p,q)$ and $X^{(p,q)}$ will be denoted by $o$, $o^+$,  $C$,  $V$, $W$ and $X$ respectively.
\end{remark}

\section{Symplectic and conformal structures for wave equations}\label{sec3}
 In the rest of the paper $\sS_m(\bM)$ is the real vector space of the 
smooth solutions $\phi$ of the  Klein-Gordon equation with mass   $m\geq 0$,  
\beq
g^{a b}\nabla_a\nabla_b \phi - m^2\phi =0 \label{KG}\:,
\eeq
which are compactly supported on smooth spacelike Cauchy surfaces $\Sigma$ of $\bM$. Moreover
\beq
\sigma_{\bM}(\phi_1,\phi_2) \doteq \int_\Sigma \left(\phi_2\nabla_{n_\Sigma} \phi_1 -\phi_1 
\nabla_{n_\Sigma}  \phi_2\right)  \: d\mu_\Sigma \label{sigma}
\eeq
is the standard  nondegenerate symplectic form on $\sS_m(\bM)$,  where $n_\Sigma$ 
is the future-oriented normal-to-$\Sigma$ unit vector 
 and $\mu_\Sigma$ is the standard measure induced on $\Sigma$ by $g$.  $\sigma_{\bM}$ is  Cauchy-surface-independent
in view of the Klein-Gordon equation and Stokes-Poincar\'e's theorem.
$(\sS_m(\bM), \sigma_{\bM})$ is a real symplectic space.

In the following $\Delta_m : C_0^\infty(\bM) \to C^\infty(\bM)$ denotes the {\em causal propagator} of the equation (\ref{KG}), that is 
the difference of the {\em advanced} and {\em retarded fundamental solutions}, which exist because $(\bM, g)$ is globally hyperbolic.

\subsection{Symplectic spaces associated to $D$ and $V$ and the Goursat problem}
All the mentioned structures can be defined replacing $(\bM,g)$ with a generic smooth globally hyperbolic spacetime    \cite{BGP}.
 In particular,  as the double cone $D\subset \bM$ is open
and  $(D,  g\spa\restriction_D)$ is globally-hyperbolic, one can define the symplectic space $(\sS_m(D), \sigma_D)$  of Klein-Gordon solutions in $D$  
with compactly supported Cauchy data on Cauchy surfaces $\Sigma_D$ of $D$, $\sigma_D$ being defined as the right-hand side of
  (\ref{sigma}) replacing $\Sigma$ for $\Sigma_D$.

\begin{remark} \label{gammam} It is worth noticing that there is a unique linear map 
preserving the relevant symplectic forms (i.e. a {\em symplectic homomorphism})
$L_{m}^{\bM D} :  \sS_m(D) \to \sS_m(\bM)$,  such that 
\beq
\left(L_{m}^{\bM D} \phi\right)\spa\restriction_D = \phi \quad \mbox{and}\quad \supp(L_m^{\bM D}\phi) \subset J^+(D) \cup J^-(D) \quad  \mbox{for every $\phi \in \sS_m(D)$}.\label{restadd}
\eeq
Indeed,  every smooth spacelike Cauchy surface $\Sigma_D$ of $D$ extends to a smooth spacelike 
Cauchy surface $\Sigma$  of $\bM$.
As the Cauchy problem is well posed in every globally hyperbolic spacetime   \cite{Wald1,BGP} 
 and  $\sigma_{\bM}\spa\restriction_{\sS_m(D)\times \sS_m(D)}= \sigma_{D}$, 
there exists  a linear function $L_{m}^{\bM D} :  \sS_m(D) \to \sS_m(\bM)$
 mapping  $\phi \in \sS_m(D)$ 
to  $L_{m}^{\bM D} \phi \in \sS_m(\bM)$ with the same Cauchy data on $\Sigma$ as $\phi$ on $\Sigma_D$.
By the uniqueness property of the Cauchy problem, (\ref{restadd}) turns out to be satisfied. The same argument proves that
$L_{m}^{\bM D}$ is uniquely determined.
Finally, $L_{m}^{\bM D}$ is injective because a linear map between two symplectic spaces that preserves the symplectic 
forms is injective provided the symplectic form of the domain is non degenerate, the proof being elementary.
\end{remark}

\noindent Let us consider the symplectic structures associated with the lower boundary $V$ of $D$.
Equipping $V$ with the topology induced by $\bM$, we define the real vector space
\beq
\sS(V) \doteq \left\{ \left. \Phi : V \spa \to \spa \bR\:\: \right| \:\: \Phi = uf\spa \restriction_V,\: 
\mbox{for some}\ f \in C^\infty(\bM),\: \supp(\Phi) \mbox{ is compact} \right\} \:,
\eeq  
where $u$ is the light-coordinate appearing in (\ref{boundary}).
By direct inspection one easily sees that every $\Phi \in \sS(V)$ results to be smooth away from the tip $o$, it vanishes in a neighborhood of $u=1$,
 $|\Phi(\omega,u)| \leq C_\Phi |u|$, uniformly in $\omega$, for some $C_\Phi\geq 0$ 
and all $u$-derivatives (of every order) of $\Phi$ are bounded functions.\\
$\sS(V)$ becomes a symplectic  space when equipped with the non degenerate symplectic form:
\beq
\sigma_V(\Phi_1,\Phi_2) \doteq \int_V   \left( \Phi_2(\omega,u) \frac{\partial \Phi_1(\omega,u)}{\partial u} -\Phi_1(\omega,u)  
  \frac{\partial \Phi_2(\omega,u)}{\partial u} \right)\: \:
  d\omega du\: \quad \Phi_1,\Phi_2\in \sS(V) \label{sigmaV}
\eeq
where $d\omega$ is the standard measure on the the $2$-sphere surface  $\bS^2$ with unitary radius.

The relation between $(\sS_m(D), \sigma_D)$ and $(\sS(V),\sigma_V)$ is stated within the following proposition, where henceforth $\varinjlim_{V}$ denotes the pointwise limit toward $V$ 
of functions.

\begin{proposition}\label{propdefjmath} For every fixed $m\geq 0$, the map
\beq  L_m^{VD} :  \sS_m(D) \ni \phi \mapsto  u \lim_{\to V} \phi \:, \label{jmath}
\eeq 
is an injective symplectic homomorphism from $(\sS_m(D), \sigma_D)$
to $(\sS(V),\sigma_V)$.
\end{proposition}

\begin{proof} $\phi\in \sS_m(D)$ is the restriction to $D$ of an element $L_{m}^{\bM D}\phi \in \sS_m(\bM)$.
Therefore $ L_m^{VD}\phi$ is the restriction to $V$ of a smooth function $f\doteq L_{m}^{\bM D}\phi \in C^\infty(\bM)$.
Furthermore $\supp(L_{m}^{\bM D} \phi)$ does not intersect $C$
so that the support of $u(L_{m}^{\bM D}\phi)\spa\restriction_V$ is compact in $V$ as requested in the definition of $\sS(V)$.
 (Let $\Sigma$ be the Cauchy surface of $\bM$ at $t=1$
and $\Sigma_D$ its restriction to $D$. The Cauchy data of $L_{m}^{\bM D}\phi$ are included 
in a compact $K \subset \Sigma_D$ so that
$\supp(L_{m}^{\bM D}\phi) \subset J^+(K) \cup J^-(K)$;
however  $(J^+(K) \cup J^-(K))\cap C = \emptyset$ by construction and thus
 $\supp(L_{m}^{\bM D}\phi) \cap C = \emptyset$.)
Hence $ L_m^{VD}\phi \in \sS(V)$, by definition of $\sS(V)$.
To conclude it is enough proving that the linear map $L_m^{VD}: \sS_m(M) \to  \sS(V)$  preserves the symplectic forms. 
This fact entails that $L_m^{VD}$ is injective as $\sigma_D$ is nondegenerate.
For $\phi_1,\phi_2 \in \sS_m(M)$, the  $3$-form
$\eta \doteq  \sqrt{|g|} g^{ea} 
\left(\phi_1 \partial_e \phi_2 - \phi_2 \partial_e  \phi_2 \right) 
\epsilon_{abcd} dx^b \wedge dx^c \wedge dx^d$ 
where $\epsilon$ is completely antisymmetric and $\epsilon_{1234}=1$, satisfies $d \eta = 0$ in view of Klein-Gordon equation.
 Applying Stokes-Poincar\'e's theorem to $\eta$ for the set $B\subset D$ whose boundary is the union of $V$ and the intersection of $\overline{D}$ and the 
 surface  $\Sigma$ at $t=1$, one easily gets $\sigma_D(\phi_1,\phi_2) = \sigma_V( L_m^{VD}(\phi_1), L_m^{VD}(\phi_2))$. 
 \end{proof}

\noindent We aim now to investigate the possibility of inverting, on the whole $ \sS(V)$,  the injective  symplectic homomorphism 
$ L_m^{VD} : \sS(D)\to \sS(V)$. In other words we are concerned with the surjectivity of $ L_m^{VD}$.
Hence we have to study the problem of determining a solution of Klein-Gordon equation in $D\cup V$ when its  restriction to $V$ is assigned.
That is  the {\em Goursat problem} for the Klein-Gordon equation in $D\cup V$,
 using $V$ as the {\em characteristic surface}. We have the following crucial result.

\begin{theorem}\label{propgoursat} Fix $m^2\in \bR$. If $\Phi \in \sS(V)$,  there is  a unique function $\phi: D \to \bR$ 
such that:
\begin{itemize}
\item[$(i)$] $\phi$ is the restriction to $D$ of a function $C_0^\infty(\bM)$,

\item[$(ii)$] $\nabla^a \nabla_a \phi - m^2 
 \phi =0$ , 

\item[$(iii)$]  $u\varinjlim_{V}\phi = \Phi$.
\end{itemize}
\end{theorem}
 
\begin{proof}
By definition, if   $\Phi \in \sS(V)$,  $\Phi/u$ is
 the restriction to $V$ of a smooth function $f_\Phi$ defined in $\bM$ whose support does
not intersects $C$ so that we can always assume  $\supp(f_\Phi)$ under  the $t=1$ Cauchy surface. With this choice and 
extending  $\Phi/u$ to the zero function in $\partial J^+(o)\setminus V$, working in $J^+(o)$ and then restricting to $D$,
Theorems 5.4.1 and  5.4.2 in \cite{friedlander} straightforwardly imply  the thesis.  
\end{proof}

\noindent Notice that the extension of $\phi$ to $\overline{D}$ is uniquely determined by continuity. Coming back to our problem to invert $ L_m^{VD}$, we notice that
the only possible candidate $\phi$ for satisfying $\phi \in \sS_m(D)$ and $ L_m^{VD}(\phi) = \Phi$ 
with a given $\Phi \in \sS(V)$ is just the 
  $\phi$ of Proposition \ref{propgoursat}. If $\Sigma_D$ is a Cauchy surface of $D$ that extends to a Cauchy surface 
$\Sigma$ of $\bM$, there is however no guarantee for having $\phi \in \sS_m(D)$.
 Because the Cauchy data of $\phi$ on  $\Sigma$ may not have compact support in  $\Sigma_D$, as requested in the definition of $\sS_m(D)$. $\phi\spa\restriction_\Sigma$ 
may vanish  exactly on $C$ but not on $\Sigma_D$ and 
 $\partial_t\phi\spa\restriction_\Sigma$ may be strictly different from $0$ on the whole $\Sigma_D \cup C$. In two dimensional models, for $m=0$ similar counterexamples can be constructed very easily.
We conclude that  the injective symplectic map $ L_m^{VD} : S_m(D) \to \sS(V)$ is not surjective. 
A surjective restriction map can be defined by using a {\em larger} space of KG solutions, where $m\geq 0$ is fixed:
\beq
\widetilde{\sS}_m(D) \doteq \{\phi :  D \to \bR \:|\: (\nabla^a \nabla_a - m^2) 
 \phi =0\:, \phi = f\spa\restriction_{D}\ , \  f\in C_0^\infty(\bM)\ ,\ uf\spa\restriction_V \in \sS(V).\}
\eeq
Now, the map that naturally extends $ L_m^{VD}$:
\beq  L_m^{V\widetilde{D}} : \widetilde{\sS}_m(D) \in \phi   \mapsto u\lim_{\to V} \phi \in  \sS(V)\eeq 
is surjective in view of Proposition \ref{propgoursat}.  Following this way, and using again  $\sigma_D$ as symplectic form on $\widetilde{\sS}_m(D)$, 
$(\widetilde{\sS}_m(D), \sigma_D)$ turns out to be a well defined real symplectic space with a nondegenerate symplectic form and includes $(\sS_m(D), \sigma_D)$ as a symplectic subspace through
 the natural inclusion map that will be indicated as: \beq L_m^{\widetilde{D}D}: \sS_m(D) \to 
\widetilde{\sS}_m(D)\:.\eeq
Notice that $L_m^{VD} = L_m^{V\widetilde{D}} L_m^{\widetilde{D}D}$ by construction.

\begin{remark} 
{\bf (1)} The extension of $(\sS_m(D), \sigma_D)$ to $(\widetilde{\sS}_m(D), \sigma_D)$ is canonically determined by the requirement that 
$\widetilde{\sS}_m(D)$ contains the solutions of the Goursat problem with data in $S(V)$.\\
{\bf (2)} What we loose in the  extension of $(\sS_m(D), \sigma_D)$ to $(\widetilde{\sS}_m(D), \sigma_D)$ is the fact that
$(\widetilde{\sS}_m(D), \sigma_D)$ is not identifiable as a symplectic subspace  of $(\sS_m(\bM),\sigma_{\bM})$.
It is because  the KG solutions of $\widetilde{\sS}_m(D)$ do not extend into a unique canonical way 
to smooth KG solutions in the whole $\bM$ through a suitable extension of the embedding $L_{m}^{\bM D}$
 in remark \ref{gammam}.
That is extending the Cauchy initial data on $\Sigma_D$ as the zero data on $\Sigma\setminus \Sigma_D$, as we did defining $L_{m}^{\bM D}$.
 In the case of a generic element of $\widetilde{\sS}_m(D)$ this procedure gives rise to discontinuous Cauchy datum $\partial_t \phi\spa \restriction_\Sigma$ crossing 
 $C\subset \Sigma$, generating 
singularities of the $v$-derivative of the KG solutions, propagating along null geodesics emanating from $C$,
when interpreting the KG equation in a suitable weaker sense. Other smooth extensions to $\bM$ are however possible, but they depend on the initial 
element $\phi \in \widetilde{\sS}_m(D)$  and on  on arbitrary choices (as smoothing functions).
\end{remark}

\noindent To go on, proving that $L_m^{V\widetilde{D}} : \widetilde{\sS}_m(D) \to \sS(V)$ is injective and surjective, we state a technical lemma, that we 
shall use in several occasions later.
 Its proof is given in the appendix and it is based on uniqueness theorems for weak solutions of 
Klein-Gordon equation in suitable Sobolev spaces (to prove that (\ref{nFMM''})
 and (\ref{nFMM'})  hold if $t\geq 0$)
 and exploits the H\"ormander's theorem about propagation of singularities (to prove that they hold for $t\leq 0$).

\begin{lemma}\label{lemmafourier} If $\phi \in \widetilde{\sS}_m(D)$ with $m\geq  0$, for $(t,\bx)\in D$  it holds 
\begin{align}
\phi(t,\bx) & = \frac{1}{(2\pi)^{3/2}}\int_{\bR^3} \frac{d{\bk}}{\sqrt{2E({\bk})}}\ e^{i{\bk}\cdot {\bx} -itE(\bk)} \left( 
 \widehat{\phi}({\bk}) +\overline{ \widehat{\phi}(-{\bk})}
\right) \label{nFMM''}\\ 
\partial_t \phi(t,{\bx}) & = \frac{-i}{(2\pi)^{3/2}} \int_{\bR^3} d{\bk}\  \sqrt{2E({\bk})}\ e^{i{\bk}\cdot {\bx} - itE(\bk)} \left( 
 \widehat{\phi}({\bk}) - \overline{\widehat{\phi}(-{\bk})}
\right)
\label{nFMM'}\:,
\end{align}
where, for every fixed $t$, the integrals above have to be understood in the sense of the Fourier-Plancherel transform, 
or as standard integrals when $\phi \in \sS_m(D)$ and where:
\beq
 \widehat{\phi}({\bk}) \doteq \frac{e^{iE(\bk)}}{(2\pi)^{3/2}} \int_{\Sigma} 
\left(\sqrt{2E({\bk})} \phi(1,{\bx}) + i\sqrt{\frac{2}{E({\bk})}} \partial_{t} \phi (1, {\bx})\right)\ e^{-i{\bk}\cdot {\bx}}\ d{\bx}\:. \label{nFMM}
\eeq
Above
$\Sigma$ is the $t=1$ Cauchy surface of $\bM$ and $E({\bk}) \doteq 
\sqrt{m^2 + {\bk}^2}$.
\end{lemma}

\begin{proposition}\label{propdeftjmath} For every fixed $m\geq 0$, the map that extends $ L_m^{VD} : \sS_m(D)\to \sS(V)$,
\beq  L_m^{V\widetilde{D}} :  \widetilde{\sS}_m(D) \ni \phi \mapsto  u \lim_{\to V} \phi \:, \label{tjmath}
\eeq
is a symplectic isomorphism from $(\widetilde{\sS}_m(D), \sigma_D)$
onto $(\sS(V),\sigma_V)$. Moreover, for every $t$
there is a sequence $\{\epsilon_n\}_{n\in \bN}\subset \bR_+$ with $\epsilon_n \to 0$ as $n\to +\infty$, such that:
\beq
\left((L_m^{V\widetilde{D}})^{-1} \Phi\right)(t,\bx) =  
\lim_{n\to +\infty} -2 \int_{\bS^2} d\omega \int_0^1du \Delta_m((t-i\epsilon_n,\bx), (u, u\omega)) 
u \partial_u \Phi(\omega,u) \:, \label{ntempE}
\eeq
almost everywhere in $\bx$.
\end{proposition}

\begin{proof} Barring surjectivity, which is assured by Theorem \ref{propgoursat},  the proof of the rest of the first statement 
in the thesis is exactly the same as for Proposition \ref{propdefjmath}. Let us pass to the second statement for the case $m>0$. In view of the Lemma \ref{lemmafourier},
 (\ref{nFMM''}) 
 makes sense, at every fixed $t\in \bR$,  interpreted as Fourier-Plancherel
transforms.
By standard results of $L^2$-convergence, there must be a sequence of reals $\epsilon_n>0$ with $\epsilon_n \to 0^+$ as $n\to +\infty$, 
such that,  almost everywhere in $\bx$:
\beq \phi(t,{\bx})  = \lim_{n\to +\infty} \frac{1}{(2\pi)^{3/2}}\int_{\bR^3} \frac{d{\bk}}{\sqrt{2E({\bk})}}e^{i{\bk}\cdot {\bx} -i(t-i\epsilon_n) E(\bk)} \left( 
 \widehat{\phi}({\bk}) +\overline{ \widehat{\phi}(-{\bk})}
\right)\:, \label{epsilonn}\eeq
   where the integral in the right-hand side in interpreted as a standard Fourier integral (this is because, 
   if $m>0$, $E(\pm\bk)^{-1/2} |\widehat{\phi}(\bk)| \leq $ constant, as follows from (\ref{nFMM})). 
The right-hand of (\ref{nFMM})  is the integral over $\Sigma_D$ of the three-form
\beq  
\eta_{\bk,\phi}\doteq -i\frac{1}{3!} \sqrt{|g|} g^{ea} 
\left( \frac{e^{-ik_m x^m}} {(2\pi)^{3/2} \sqrt{2E}}  \partial_e \phi - \phi \partial_e    \frac{e^{-ik_m x^m}} {(2\pi)^{3/2} \sqrt{2E}}  \right) 
\epsilon_{abcd} dx^b \wedge dx^c \wedge dx^d \label{3forma}\ .
\eeq  
Dealing with as in the proof of Proposition \ref{propdefjmath}, employing Stokes-Poincar\'e's theorem and integrating by parts, we can rearrange the 
right-hand side of  (\ref{nFMM}), obtaining:
\beq 
 \widehat{\phi}({\bk}) = \frac{-2i}{(2\pi)^{3/2} \sqrt{2E(\bk)}} \int_{\bS^2\times [0,1]} \sp\sp\sp d\omega du \:
 ue^{-i({\bk} \cdot u \omega - E u)} \partial_u L_m^{V\widetilde{D}}(\phi) (\omega,u) 
\:. \label{nmode2}
 \eeq
 Inserting it in (\ref{epsilonn}), after some manipulations one gets 
$$
\phi(t,\bx) =  
\lim_{n\to +\infty} -2 \int_{\bS^2} d\omega \int_0^1du \Delta_m((t-i\epsilon_n,\bx), (u, u\omega)) 
u \partial_u L_m^{V\widetilde{D}} \phi \:, 
$$ 
 for every $t$, a.e. in $\bx$ where
  $\Delta_m(t,\bx,t',\bx')$ is a smooth function when analytically extended to complex values of $t$.
  Defining $\Phi \doteq L_m^{V\widetilde{D}} \phi$ (noticing that, due to the definition of $\widetilde{\sS}_m(D)$, $\Phi$ varies everywhere in $\sS(V)$ if $\phi$ ranges in
   $\widetilde{\sS}_m(D)$), we have $\phi = (L_m^{V\widetilde{D}})^{-1} \Phi$, so that the found identity is (\ref{ntempE}).
 In the case $m=0$  (\ref{ntempE}) is still true. Indeed,  using  the explicit expression 
 \beq \Delta_0((t-i\epsilon_n,\bx), (t', \bx')) = \frac{i}{4\pi^2}\left[  \frac{1}{\sigma((t+i\epsilon_n,\bx), (t', \bx'))} -
    \frac{1}{\sigma((t-i\epsilon_n,\bx), (t', \bx'))}\right]\:,\label{delta0}\eeq
after some straightforward manipulations, the right-hand side of (\ref{ntempE}) turns out  to coincide with
 the explicit solution of the Goursat problem for the massless Klein-Gordon equation with characteristic datum $\Phi/u$
 as represented in  (5.4.17) in \cite{friedlander}, independently from the choice of the sequence $\{\epsilon_n\}_{n\in \bN}$.  
 \end{proof}
 
 \subsection{Action of conformal Killing fields on massless KG solutions}\label{secmassless}
As is well-known (e.g. see \cite{CFT}), if  $\phi$ satisfies Klein-Gordon equation with $m=0$, then $\beta^Y_\tau \phi$ in (\ref{conf}) satisfies the same equation, provided that $Y$ 
is a conformal Killing vector and 
$\beta^Y_\tau \phi$ is well defined (the one-parameter group generated by $Y$ is only local in the general case). Let us specialize to the case of $Y=X$ given in (\ref{HL}).
In that case, the flow of $X$ leaves separately $D$ and $\partial D$ and $V$ invariant, in particular, giving rise to a well-defined action $\beta^X_\tau : C^\infty(D) \to C^\infty(D)$
for every $\tau\in \bR$. There is a similar action of $\beta^X_\tau$ on the functions defined on $V$. We want to study the interplay of these actions restricting
 to $\widetilde{\sS}_0(D)$
and $\sS(V)$ respectively. To this end we notice that from the definition (\ref{HL}) of $X$, if $(\omega^X_\tau(\omega,u),
u^X_\tau(\omega,u)) \doteq \Upsilon^X(\tau,(\omega,u))$, one has
 \beq
 X\spa\restriction_V(\omega,u) = u(u-1) \partial_u \quad \mbox{so that}\quad u^X_\tau(\omega, u) = \frac{u}{u+ e^\tau (1-u)}\:,\:\: 
\omega^X_\tau(\omega, u) = \omega\:. \label{XV}
 \eeq
 
Referring to (\ref{delta}), one easily proves from the first identity in (\ref{XV}) that, whenever a map
$f: \bS^2 \times (0,1) \to \bR$ is differentiable, then  $u \gamma^X (f) = -u X(f) - u(u-1) f$.
 If we define $\Phi \doteq uf$, the found identity 
implies that the infinitesimal generator $\gamma^X$ satisfies  $u \gamma^X(f) =  -X(\Phi)$. Passing form the infinitesimal  to the finite action one finds that
 transforming $f \to \beta_\tau^{X}(f)$  (interpreted as in (\ref{conf}) with $Y=X$) 
is equivalent to transform the associated $\Phi \doteq u f$ as $\Phi \to \beta_\tau^X(\Phi)$ where, with a harmless
misuse of notation, we define:
\beq
\left( \beta^{X}_\tau(\Phi)\right) (\omega, u) \doteq     \Phi\left(\omega, \frac{u}{u+e^{-\tau}(1-u)}  \right)\:.
\label{stau}
 \eeq 
 
 From now on, $\beta^{X}_\tau(\Phi)$ with $\Phi$ defined on $V$ will always mean (\ref{stau}) in spite of the definition (\ref{conf}) with $Y=X$,
   which we will only adopt for smooth functions defined 
 in the bulk $D$ or in the whole $\bM$. We have the following technical result.\\
 
 \begin{proposition} \label{symp} $\beta^{X}_\tau : \sS(V) \to \sS(V)$  is well-defined for every $\tau \in \bR$, giving rise to a one-parameter group of  symplectic isomorphisms
 of $(\sS(V), \sigma_V)$.
 \end{proposition}

 \begin{proof}
 Define the vector field $X'\doteq \chi X$ on $\bM$, where $\chi \in C_0^\infty(\bM)$
is such that $\chi=1$ on $\overline{D}$. As $X'$ has compact support, its integral lines are complete and the associated 
 one-parameter group of diffeomorphisms is global. Consequently the Jacobian $J_\tau^{X'}$ 
 does not vanish anywhere and thus it is strictly positive (as it is for $\tau=0$).
  Defining $\beta_\tau^{X'}(f)$ as in (\ref{conf}), if $f\in C_0^\infty(\bM)$ we get a well-behaved action for every value of
 $\tau\in \bR$ and $\beta_\tau^{X'}(f) \in C_0^\infty(\bM)$ for every $\tau\in \bR$. Furthermore, in $\overline{D}$, the action of $X'$ is the same as that of $X$.
If $\Phi \in \sS(V)$, $\Phi = u f\spa \restriction_V$ for some $f \in  C_0^\infty(\bM)$. Therefore, from the discussion before Proposition \ref{symp},
$\beta^X_\tau \Phi = u (\beta^{X'}_\tau f) \spa \restriction_V $ so that $\beta^X_\tau \Phi = u g\spa \restriction_V$ for $g= \beta^{X'}_\tau f \in C_0^\infty(\bM)$.
Finally, notice that computing the support in $V$,  $\supp(\beta^X_\tau \Phi)$ cannot reach the points at $u=1$ 
since the flow of $-X\spa \restriction_V$ transforms the support of  $\beta^X_\tau \Phi$ to the support of $\Phi$, but the points at $u=1$ 
are fixed for the flow of $-X$ and $\supp(\Phi)$ does not include those points by definition and thus they cannot belong to $\supp(\beta^X_\tau \Phi)$.\\ 
To prove that $\beta^{X}$ preserves the symplectic form $\sigma_V$, thinking of $V$ as the product $\bS^2 \times [0,1)$,
we notice that  the symplectic form $\sigma_V$  (\ref{sigmaV}) is trivially invariant under the (pull back) action of 
 diffeomorphisms of $[0,1]$ onto $[0,1]$ preserving the endpoints.
For each $\tau \in \bR$, the map $[0,1] \ni u \mapsto  u(u+e^{\tau}(1-u))^{-1}$ is just such a diffeomorphism, so that the bijective 
(as $\beta^{X}_\tau \beta^{X}_{-\tau} = \beta^{X}_{-\tau} \beta^{X}_\tau = id$) linear map 
$\beta^{X}_\tau : \sS(V) \to \sS(V)$ in (\ref{stau}) preserves the symplectic form.  
\end{proof}

\noindent We are now in place to state the theorem concerning the interplay of the action of $\beta^X$ on $D$
 and that on $V$ for massless KG solutions. Below, the index $0$ means ${m=0}$.\\

\begin{theorem}\label{teomz} $\beta^{X}_\tau : \widetilde{\sS}_0(D) \to \widetilde{\sS}_0(D)$  
is well-defined for every $\tau \in \bR$ giving rise to a one-parameter group of  symplectic isomorphisms
 of $(\widetilde{\sS}_0(V), \sigma_V)$. Moreover it holds
 \beq
 \beta^X_\tau \circ L_0^{V\widetilde{D}} =   L_0^{V\widetilde{D}} \circ  \beta^X_\tau\ , \quad \mbox{for all $\tau \in \bR$.} \label{commut}
\eeq
\end{theorem}

 \begin{proof} Since the flow of $X$ preserves $D$  and it transforms massless solutions of Klein-Gordon equation to massless solutions of Klein-Gordon equation,
 $\beta^{X}_\tau$ transforms smooth solutions of massless Klein-Gordon equation defined in $D$ to smooth solutions of massless Klein-Gordon equation defined in $D$. With the same procedure 
 as that exploited in the proof of Proposition \ref{symp}, using the vector field $X'$ instead of $X$,
 one sees that if $f'\in C^\infty_0(\bM)$,
 $\beta^{X'}_\tau(f')\in C^\infty_0(\bM)$ for every fixed $\tau$ and  $(\beta^{X'}_\tau(f'))\spa\restriction_D = \beta^{X}_\tau(f'\spa\restriction_D)$.
Those results imply that, if  $\phi \in \widetilde{\sS}_0(D)$, then $\beta^X_\tau\phi$ is a solution of Klein-Gordon equation in $D$ and it is the restriction to $D$ of some smooth 
compactly supported function defined on $\bM$. To prove that $\beta^X_\tau\phi\in \widetilde{\sS}_0(D)$ it is enough establishing  that the $u$-rescaled restriction
of $\beta_\tau^X \phi$ to $V$ belongs to $\sS(V)$.  By construction (since $X$ is everywhere defined on $\overline{D}$ and all the functions are jointly smooth) 
$u\varinjlim_ {V}\beta^X_\tau \phi =  \beta^X_\tau L_0^{V\widetilde{D}} \phi$,
so that $\beta^X_\tau \phi  \in \widetilde{\sS}_0(D)$, because $\beta^X_\tau L_0^{V\widetilde{D}} \phi \in \sS(V)$ in view of Proposition \ref{symp}.
 The found identity can alternately be written $L_0^{V\widetilde{D}} \beta^X_\tau \phi = \beta^X_\tau L_0^{V\widetilde{D}} \phi$ proving (\ref{commut}). The fact that
 $\beta_\tau^X$ preserves the symplectic form $\sigma_D$ follows from (\ref{commut}), Proposition \ref{symp}  and the fact that $L_0^{V\widetilde{D}} : 
 \widetilde{\sS}_0(D) \to \sS(V)$
 preserves the corresponding  symplectic forms as established  in Theorem \ref{propdeftjmath}.  
 \end{proof} 

\section{Algebras and  states for the bulk and the boundary}\label{sec4}
This section is devoted to study the interplay of the various algebras of observables 
associated  with the quantum Klein-Gordon field in $D$ and restricted to $V$, 
and some relevant states. We shall prove that, remarkably, the restriction $\omega_m$ of the standard Minkowski vacuum to
the bulk algebra  $\cW_m(D)$ corresponds to a unique state $\lambda$ on the boundary algebra $\cW(V)$,
independently from the mass $m>0$ of the field.

\subsection{Weyl algebras on $D$ and $V$ and their relations}
As $(\sS_m(\bM), \sigma_{\bM})$ is a real symplectic space with 
non-degenerate symplectic form, there is a unique (up to $*$-algebra isomorphisms) $C^*$-algebra,  
$\cW_m(\bM)$, with (nonvanishing) generators  $W_{m,\bM}(\phi)$ satisfying  Weyl relations \cite{BGP,BR1}, 
for every choice of $\phi_i,\phi \in \sS_{m}(\bM)$:
$$
W_{m,\bM}(\phi_1) W_{m,\bM}(\phi_2) = e^{\frac{1}{2}i\sigma_{\bM}(\phi_1,\phi_2)} W_{m,\bM}(\phi_1+\phi_2)\:, \quad W_{m,\bM}
(\phi)^* = W_{m,\bM}(-\phi)\:.
$$
This is the {\em Weyl algebra} of Klein-Gordon field on $\bM$. 
As is well known, it can be equivalently realized in terms of generators smeared with real functions $f \in C_0^\infty(\bM)$
when defining:
\beq 
W_{m,\bM}([f]) \doteq W_{m,\bM}(\Delta_{m,\bM} f) \quad \mbox{for all $f\in C_0^\infty(\bM)$,}\eeq

where $\Delta_{m,\bM}$ is the 
causal propagator introduced previously and whose support properties immediately imply that
locality holds,
\beq
\left[ W_{m,\bM}([f]), W_{m,\bM}([h])\right] =0\ , \label{locality} 
\eeq
whenever $\supp(f)$ and $\supp(h)$ are causally separated. 

Similar Weyl algebras can be constructed for the symplectic spaces $(\sS_m(D), \sigma_D)$,  $(\widetilde{\sS}_m(D), \sigma_D)$ and $(\sS(V), \sigma_V)$. The corresponding Weyl algebras will be denoted by $\cW_m(D)$, $\widetilde{\cW}_m(D)$, $\cW(V)$ with generators, respectively, $W_m(\phi)$, $\widetilde{W}_m(\phi)$ and $W(\Phi)$. 

Using well-known theorems for lifting morphisms properties of symplectic spaces to analog properties of the corresponding Weyl algebras \cite{BR2}, we have the straightforward result:

\begin{proposition}\label{propJ} For  $m\geq 0$ the following holds. There is an isometric $^*$-homomorphism $\ell_m^{\bM D} : \cW_m(D) \to \cW_m(\bM)$ uniquely individuated by the requirements
\beq
\ell_m^{\bM D} (W_m(\phi)) = W_{m,\bM}(L_m^{\bM D}\phi) \quad \mbox{for all $\phi \in \sS_m(D)$.}
\eeq
There is a unique isometric $^*$-homomorphism $\ell_m^{\widetilde{D}D} : \cW_m(D) \to \widetilde{\cW}_m(D)$ such that
\beq
\ell_m^{\widetilde{D}D} (W_m(\phi)) = \widetilde{W}_{m}(L_m^{\widetilde{D}D}\phi) \quad \mbox{for all $\phi \in \sS_m(D)$.} \label{JDTD}
\eeq
There is a unique isometric $^*$-homomorphism $\ell_m^{VD} : \cW_m(D) \to \cW(V)$  such that 
\beq
\ell_m^{VD} (W_m(\phi)) = W(L_m^{VD}\phi) \quad \mbox{for all $\phi \in \sS_m(D)$.} \label{alpha}
\eeq
There is a $^*$-isomorphism $\ell_m^{V\widetilde{D}} : \widetilde{\cW}_m(D) \to \cW(V)$ uniquely individuated by
\beq
\ell_m^{V\widetilde{D}} (\widetilde{W}_m(\phi)) = W(L_m^{V\widetilde{D}}\phi) \quad \mbox{for all $\phi \in \widetilde{\sS}_m(D)$.} \label{alpha1}
\eeq
Finally \beq \ell_m^{VD} =   \ell_m^{V\widetilde{D}}   \ell_m^{\widetilde{D}D}\label{coer}\:.\eeq
\end{proposition}

\begin{remark} Similarly to the case of symplectic spaces, the Weyl algebra $\widetilde{\cW}_m(D)$ cannot be naturally identified with a sub $C^*$-algebra of $\cW_m(\bM)$. However, this 
drawback  will be fixed at the level of von Neumann algebras as we are going to prove.
\end{remark}

\subsection{States on $\cW_m(D)$, $\cW(S)$, $\widetilde{\cW}_m(D)$ and their relations} Henceforth we adopt the definition of {\em quasifree state} as stated in \cite{Wald,KW} (with the conventions of the latter concerning the signs of the
 symplectic forms) and assume the reader 
is acquainted with the basic theory of those states and their GNS Fock representations.
 In the following  $\omega_m: \cW_m(D) \to \bC$ denotes the restriction to $\cW_m(D)$ of the standard Poincar\'e-invariant 
 quasifree vacuum state on $\cW_m(\bM)$ \cite{Haag},  obtained  by standard decomposition of the $\phi \in \sS_m(\bM)$ 
 into positive and negative frequency parts. Taking advantage from
 the embedding  $L_m^{VD} : \cW_m(D) \to \cW(V)$ (\ref{jmath}) and (\ref{alpha}),
we intend to prove that $\omega_m = \lambda \circ L_m^{VD}$, for a
quasifree 
state, {\em independently from the value of $m$},
$\lambda : \cW(V) \to \bC$ 
completely individuated by the requirements, if  $\Phi,\Phi' \in \sS(V)$,
\beq 
\lambda(W(\Phi)) \doteq e^{-\frac{\mu_\lambda(\Phi,\Phi)}{2}}\:,\label{lambda}
\eeq
\beq
 \mu_\lambda (\Phi, \Phi') \doteq  
  \Rea\left( \lim_{\epsilon \to 0_+} -
\int_{\bS^2} \!d\omega\!
\int_0^1\! du\! \int_0^1\! du'  \:\: \frac{\Phi(\omega,u) \Phi'(\omega,u')}{\pi (u-u'-i\epsilon)^2}\right) \:.\label{lambdaagg}
\eeq

\begin{proposition}\label{prop1.2} There exists a unique quasifree state $\lambda$ on $\cW(\sS(V))$ satisfying (\ref{lambda})-(\ref{lambdaagg}). 
 Its one-particle space structure $(\sK_\lambda, \sH_\lambda)$, which uniquely individuates the state 
is the following:
\begin{itemize}
\item[$(1)$] $\sK_\lambda : \sS(V) \to \sH_\lambda$ associates every $\Phi\in \sS(V)$ with the corresponding $\widehat{\Phi}$:
\beq\widehat{ \Phi}(\omega,k) = \frac{1}{\sqrt{2\pi}}\int_\bR du e^{iku} \Phi(\omega,u) \quad
 \mbox{for every $(\omega,k) \in \bS^2\times \bR_+$,} \eeq
where, in the integrand, we have extended $\Phi$ to the null function for  $u \not \in [0,1)$;
\item[$(2)$] The one-particle space $\sH_\lambda$ is the closure 
 of $\sK_\lambda (\sS) + i\sK_\lambda (\sS)$ in  $L^2(\bS^2 \times \bR_+, 2kd\omega dk)$.
 \end{itemize}
 \end{proposition}

\begin{proof}
Following Section  3.2 in \cite{KW}, a quasifree state is defined by the requirements (\ref{lambda})-(\ref{lambdaagg}) as stated in the thesis,  if $(i)$
 $\mu_\lambda$ is a  (real) positive-definite  scalar product on $\sS(V)$ and $(ii)$
$|\sigma_V(\Phi,\Phi')|^2 \leq 4\mu_\lambda (\Phi, \Phi)\mu_\lambda (\Phi', \Phi')$ for all $\Phi,\Phi' \in \sS(V)$.
Then the  one-particle structure space  $(\sK_\lambda,\sH_\lambda)$ is individuated, up to Hilbert-space isomorphisms, by
$\langle \sK_\lambda \Phi, \sK_\lambda \Phi' \rangle_{\sH_\lambda} = 
\mu_\lambda(\Phi,\Phi') - i\sigma_V(\Phi,\Phi')/2$ and  $\sH_\lambda$ is  
the completion of $\sK_\lambda (\sS(V)) + i\sK_\lambda (\sS(V))$ with respect to the so-defined  Hermitean scalar product.
We have to check the validity of (i) and (ii). 
 To this end, taking in general $\Phi,\Phi' \in \sS(V) + i \sS(V)$ we  study the well-posedness of 
 \beq
\lambda (\overline{\Phi}, \Phi') \doteq  \lim_{\epsilon \to 0_+} -\frac{1}{\pi}  
\int_{\bS^2\times [0,1]\times [0,1]} \sp\sp\sp  \sp\sp\sp d\omega du du' 
\frac{\overline{\Phi(\omega,u)} \Phi'(\omega,u')}{(u-u'-i\epsilon)^2}\:, \quad
\mbox{for $\Phi,\Phi' \in \sS(V)+ i \sS(V)$.} \label{lambda2}
\eeq
 $\Phi$ and $\Phi'$ are  $L^2(\bS^2\times \bR, d\omega du)$ functions such that, by construction,
   vanish uniformly in the angle $\omega$ as $u \to 0_+$, moreover
 the $u$-derivatives of $\Phi$ and $\Phi'$ belong to $L^2(\bS^2\times \bR, d\omega du)$ (in particular, as one sees by direct inspection, 
the $u$-derivatives turn out to bounded around $u=0$ because $u^{-1}\Phi$ and $u^{-1}\Phi'$ are restrictions to $V$ of smooth functions defined around $V$). Therefore 
 the following lemma (proved in the appendix), leads to the wanted well-posedness.

\begin{lemma} \label{lemmalemme}  Consider  $\Phi,\Phi' \in L^2(\bS^2\times \bR, d\omega du)
\cap C^1(\bS^2\times (0,+\infty); \bC)$ such that they  vanish for $u \in \bR \setminus [0,1)$ and  $\omega \in \bS^2$ and 
 one (or both) of the following facts holds:
\begin{itemize}
\item[$(a)$] $\partial_u \Phi \in L^2(\bS^2\times \bR, d\omega du)$ and  $\Phi(\omega,u) \to 0$ as $u \to 0_+$ uniformly in $\omega \in \bS^2$ ;
 \item[$(b)$] $\partial_u \Phi' \in L^2(\bS^2\times \bR,d\omega du)$ and  $\Phi'(\omega,u) \to 0$ as $u \to 0_+$ uniformly $\omega \in \bS^2$ ; 
 \end{itemize}
  then:
\beq  \lim_{\epsilon \to 0_+}  -\frac{1}{\pi} 
\int_{\bS^2\times [0,1]\times [0,1]} \sp\sp\sp d\omega  du du'\frac{\overline{\Phi(\omega,u)} \Phi''(\omega,u')}{(u-u'-i\epsilon)^2}
= \int_{\bS^2 \times \bR^+} d\omega dk 2k  \overline{\widehat{\Phi(\omega,k)}} \widehat{\Phi'(\omega,k)} \label{mod}\:.\eeq
\end{lemma}

\noindent Coming back to the validity of requirements $(i)$ and $(ii)$ we notice that,
in view of Lemma \ref{lemmalemme}, $\lambda(\overline{\Phi},\Phi')$ is well-defined and $\lambda(\overline{\Phi},\Phi) \geq 0$.
 Let us now restrict to the real case: $\Phi \in \sS(V)$ so that, in particular,
 $\mu_\lambda(\Phi,\Phi) = \Rea \lambda(\Phi,\Phi)  = \Rea \lambda(\overline{\Phi},\Phi) = \lambda(\Phi,\Phi)$.
 To prove (i) we only have  to prove the strict positivity of $\mu_\lambda$.
If $\mu_\lambda(\Phi,\Phi)=0$, (\ref{mod}) entails that $\widehat{\Phi}(\omega,k) = 0$, a.e, for $k>0$, $\omega \in \bS^2$.
On the other hand, as  $\widehat{\Phi}(\omega,-k) = \overline{\widehat{\Phi}(\omega,k)}$ because $\Phi$ is real, it must  hold
$\widehat{\Phi}(\omega,k) = 0$, a.e, for $k\in \bR$, $\omega \in \bS^2$. Since the Fourier-Plancherel transform is injective, it implies $\Phi =0$ a.e. and thus
$\Phi=0$ everywhere as $\Phi\in \sS(V)$ is continuous by definition. 
We pass to (ii) noticing that, from (\ref{lambdaagg}) by trivial manipulations of the right-hand side employing integration by parts and 
Sochockij's formula: 
\beq
-2 \Imm \lambda (\Phi, \Phi') = \sigma_V(\Phi,\Phi') \quad \mbox{for all $\Phi,\Phi' \in \sS(V)$.} \label{Im}
\eeq
If $(i)$ holds, as the Hermitean quadratic form $\lambda(\overline{\Phi},\Phi')$ is semi-positive defined, the  
Cauchy-Schwartz inequality holds and so  $|\Imm \lambda(\overline{\Phi},\Phi')|^2 \leq 
|\lambda(\overline{\Phi},\Phi')|^2 \leq \lambda(\overline{\Phi},\Phi)  \lambda(\overline{\Phi'},\Phi')$. Restricting again to the real case
$\Phi,\Phi' \in \sS(V)$ and exploiting  $\mu(\Phi,\Phi) = \lambda(\Phi,\Phi)$
and (\ref{Im}), $(ii)$ arises  immediately. To conclude, we observe that
 in view of the definition of $\sH_\lambda$,  $\mu_\lambda$, (\ref{lambda2}), (\ref{mod}) and (\ref{Im}), one has:
$\overline{\sK_\lambda \sS(V)+ i \sK_\lambda \sS(V)}= \sH_\lambda$ and  $\langle \sK_\lambda \Phi,\sK_\lambda \Phi'  \rangle_{\sH_\lambda}
= \mu_\lambda(\Phi,\Phi') - \frac{i}{2}\sigma_V(\Phi,\Phi')$ if $\Phi,\Phi'\in \sS(V)$. 
 We have then verified the hypotheses of Proposition 3.1 in \cite{KW} for which $(\sK_\lambda,\sH_\lambda)$ 
 defines (up to unitary equivalence) the one-particle space structure of $\lambda$.
\end{proof}

\noindent We are now  in place to establish the main theorem of this section.
The quasifree state $\omega_m$ on $\cW_m(D)$ is completely individuated by:
\beq
\omega_m(W_m(\phi)) \doteq e^{- \frac{1}{2}\mu_m(\phi,\phi)}\:,
\eeq
 \beq 
\mu_m(\phi_1,\phi_2) \doteq \Rea\left(\int_{\bR^3} d{\bk} \:\: \overline{\widehat{\phi_1}({\bk})}\widehat{\phi_2}({\bk})\right)\   ,  \qquad 
\forall \phi_1,\phi_2
\in \sS_m(D) \:,
\eeq
and this is equivalent to say that one-particle space structure of $\omega_m$ is $(\sK_m, \sH_m)$ where, for 
 $\widehat{\phi}$ as defined  in (\ref{nFMM}):
\beq \sK_m\phi \doteq  \widehat{\phi}\ ,\quad \phi \in \sS_m(D)\:,\quad
\:\: \sH_m \doteq  \overline{\sK_m(\sS_m(D))+ i \sK_m(\sS_m(D))} \subset L^2(\bR^3,d\bk).
\label{vacuum}
\eeq

\begin{theorem}\label{teobast} Let  $\omega_{m} : \cW_m(D) \to \bC$ be 
the restriction to $\cW_m(D) $ of the standard  Poincar\'e invariant vacuum of the real Klein-Gordon field 
with mass $m \geq  0$ in $\bM$. It holds:
\beq \omega_{m} =  \lambda \circ \ell_m^{VD} \label{omegam} \:.\eeq
\end{theorem}
\begin{proof}
As all the involved states are  completely individuated by the real scalar products 
 $\mu_\lambda$ and $\mu_m$ and $\ell_m^{VD}$ satisfies (\ref{alpha}),  the thesis is a consequence of the
 following lemma whose very technical proof is given in the appendix.
\end{proof}
 
 \begin{lemma} \label{lemmabastardo}
If $m \geq 0$, for all $\phi_1,\phi_2 \in \sS_m(D)$ one has
\beq
 \int_{\bR^3} d{\bk} \:\: \overline{\widehat{\phi_1}({\bk})}\widehat{\phi_2}({\bk})  =  
 \lim_{\epsilon \to 0_+}-\frac{1}{\pi}  
\int_{\bS^2\times [0,1]\times [0,1]} \sp\sp\sp\sp\sp\sp\sp\sp\sp   d\omega  du du' \frac{L_m^{VD}(\phi_1)(\omega,u) L_m^{VD}(\phi_2)(\omega,u')}{(u-u'-i\epsilon)^2}\:, \label{id}
\eeq
where $\widehat{\phi}$ is  given 
in (\ref{nFMM}) for every $\phi \in \sS_m(D)$.
\end{lemma} 

\begin{remark} The statement of Theorem \ref{teobast}  resembles the uniqueness theorem proved in \cite{KW} for the restriction of the two-point function of a quasifree Hadamard states 
to bifurcate Killing horizons. However there is here an important difference that, in fact,  required a different proof: The conical light surface $V$ is not the union of the orbits 
of a causal Killing field.

\noindent The structure (\ref{lambda})-(\ref{lambdaagg}) of the state $\lambda$ arises in several different contexts, even in quantum field theories in curved spacetime and refers to 
conformal field theory and lightfront holography, where one deals with the ``restriction'' of states 
on null surfaces or with the procedure to induce states on the bulk algebra form states defined on the boundary
\cite{KW,H00,Sch03,Sch08,DMP1,DMP2,DMP3,BMRW,DPP}. In the second case, the states induced in the bulk by those individuated on the boundary with a form similar to  (\ref{lambda})-(\ref{lambdaagg}) turn out to be of Hadamard class \cite{H00,Mor2,DMP3,DMP4,P,DPP} and are invariant under the action of natural symmetries 
induced by the spacetime background \cite{DMP1,Mor1,DMP2,DPP}. In the presently discussed case these features are automatically satisfied as the state induced on the algebra of the bulk is  the Poincar\'e invariant vacuum. 
\end{remark}

\noindent To conclude we  pass from $\cW_m(D)$ to the larger algebra $\widetilde{\cW}_m(D)$
proving that it admits a natural extension of 
$\omega_m$  that coincides with $\lambda$
through the *-isomorphism $\ell_m^{\widetilde{D}M}$.

\begin{theorem} For $m> 0$, there exists a quasifree state $\widetilde{\omega}_m$ on $\widetilde{\cW}_m(D)$  individuated by the property,
\beq
\widetilde{\omega}_m(\widetilde{W}_m(\phi)) = e^{- \frac{1}{2}\widetilde{\mu}_m(\phi,\phi)}\quad\mbox{with}\quad
\widetilde{\mu}_m(\phi_1,\phi_2) \doteq \Rea\left(\int_{\bR^3} d{\bk} \:\: \overline{\widehat{\phi_1}({\bk})}
\widehat{\phi_2}({\bk})\right) \:,\label{tomegam}
\eeq 
for all $\phi_1,\phi_2\in\widetilde{\sS}_m(D)$ and where $\widehat{\phi}$ is given 
in (\ref{nFMM}). Moreover
$\widetilde{\omega}_m$ coincides with $\lambda$ through the *-isomorphism 
$\ell_m^{\widetilde{D}M} : \widetilde{\cW}_m(D) \to 
\cW(V)$:
\beq
\widetilde{\omega}_m \doteq \lambda \circ \ell_m^{\widetilde{D}M}\:. \label{omegatilde}
\eeq
\end{theorem}
\begin{proof}
It is immediate to realize that, using again Section  3.2 in \cite{KW}, $(i)$
 $\widetilde{\mu}_m$ is a  (real) positive-defined  scalar product on $\widetilde{\sS}_m(D)$ and $(ii)$
$|\sigma_D(\phi,\phi')|^2 \leq 4\widetilde{\mu}_m (\phi, \phi)\widetilde{\mu}_m(\phi', \phi')$ for all $\phi,\phi' \in \widetilde{\sS}_m(D)$.
 Moreover if $\widetilde{\sK}_m\phi \doteq \widehat{\phi}$ given in (\ref{nFMM})
for $\phi \in \widetilde{\sS}_m(D)$ and we define the Hilbert space $\widetilde{\sH}_m \doteq \overline{\widetilde{\sK}_m(\widetilde{\sS}_m(D)) + i \widetilde{\sK}_m(\widetilde{\sS}_m(D))}$
in $L^2(\bR^3,d\bk)$,
it holds
  $\langle\widetilde{\sK}_m \phi,\widetilde{\sK}_m \phi'  \rangle_{\widetilde{\sH}_m}
= \widetilde{\mu}_m(\phi,\phi') - \frac{i}{2}\sigma_D(\phi,\phi')$ if $\phi,\phi'\in \widetilde{\sS}_m(D)$. 
 These are the hypotheses of Proposition 3.1 in \cite{KW} which entail that $(\widetilde{\sK}_m,\widetilde{\sH}_m)$ 
is (up to unitary equivalence) the one-particle space structure of $\widetilde{\omega}_m$.
Now the proof of (\ref{omegatilde}) is the same as that of (\ref{omegam}) since  Lemma 
\ref{lemmabastardo} is valid, with the same proof, also referring to $\phi_1,\phi_2 \in \widetilde{\sS}_m(D)$  replacing 
$L_m^{VD}$ by $L_m^{V\widetilde{D}}$. 
\end{proof}

\section{The modular groups for a double cone and its boundary}\label{sec5}
In view of Reeh-Schlieder's \cite{Haag}
and locality properties  (\ref{locality})  promoted to the von Neumann algebra level, the cyclic GNS vector $\Psi_\omega$ of the state $\omega_m$ is separating for the von Neumann algebra $\pi_\omega(\cW_m(D))''$
generated by the GNS representation of $\cW_m(D)$ in the GNS Hilbert space. In other words, $(\pi_m(\cW_m(D))'', \Psi_m)$ is in {\em standard form} and
thus it admits a unique modular group  \cite{hug, BR2,Haag}. We want to study it exploiting
 the link between modular theory and KMS theory \cite{hug,BR2,Haag}:
 If a state $\omega$ on the unital $C^*$-algebra  $\mA$ is KMS with inverse temperature $\beta \neq 0$,  then the pair $(\pi_\omega(\mA)'', \Psi_\omega)$
is in ``standard form'' and, if $U_\tau$ is the unitary implementing $\alpha_\tau$ in the GNS representation of $\omega$ leaving $\Psi_\omega$ fixed, the modular group is
$\sigma_\tau (A) = U_{-\beta \tau} A U^*_{-\beta \tau}$ for all $\tau \in \bR$ and  $A \in \pi_\omega(\mA)''$.

\subsection{The modular group for the boundary $V$}
We start our analysis building up the modular group associated with the von Neumann algebra of the 
GNS representation of   $\lambda$
defined on the algebra of the boundary $\cW(V)$.

 \begin{theorem} \label{T3} Consider the quasifree state $\lambda$ on $\cW(V)$ defined in (\ref{lambda}) and its
 GNS triple $(\cH_\lambda,\pi_\lambda, \Psi_\lambda)$.
 Let $\{\alpha^{X}_\tau\}_{\tau\in\bR}$ be the one-parameter group of $*$-automorphisms 
of $\cW(V)$ completely 
  individuated by the requirement
 \beq
 \alpha^{X}_\tau\left( W(\Phi)\right) =  W\left(\beta^{X}_\tau(\Phi)\right)\quad \mbox{for all $\Phi \in \sS(V)$,} \label{alphatau}
 \eeq
 with $\beta^X$ as in (\ref{stau}). Then the following hold:
 \begin{itemize}
\item[$(a)$]  $\lambda$  is invariant and $KMS$ with inverse temperature $\beta=2\pi$ with respect to $\alpha^{X}$.
\item[$(b)$]  The modular group of $(\pi_\lambda(\cW(V))'', \Psi_\lambda)$ is unitarily implemented by $U^{(\lambda)}_\tau\Psi_\lambda = \Psi_\lambda$ if $\tau \in \bR$
and $U^{(\lambda)}$ is the second-quantization of the one-parameter group of unitaries  $V_\tau^{(\lambda)} : \sH_\lambda \to \sH_\lambda$ for which,
\beq
 V^{(\lambda)}_\tau \sK_\lambda = \sK_\lambda \beta^{X}_\tau \quad\mbox{for all $\tau \in \bR$.} \label{ONE}
 \eeq
 \end{itemize}
 \end{theorem}
 
\noindent {\em Sketch of the proof (details in the Appendix)}.  Our strategy consists in constructing a quasifree state $\lambda'$ on $\cW(V)$ which admits $(\sK'_\lambda, \sH'_\lambda)$  
as one-particle space structure with 
\begin{align}
 &\sH'_\lambda \doteq   L^2(\bS^2 \times \bR; m(h)  d\omega dh)\quad\mbox{where}\quad
 m(h)  \doteq  \frac{h e^{2\pi h}}{e^{2\pi h}-e^{-2\pi h}}\:, \label{mu}\\
& \left(\sK'_\lambda(\Phi)\right)(\omega,h) \doteq \int_\bR d\ell
\frac{e^{i\ell h}}{\sqrt{2\pi}}\Phi(\omega,\ell) \doteq  \widetilde{\Phi}(\omega,h) \:,\label{Klambda}
\end{align}
 where $\Phi$  has been represented in coordinates $(\omega,\ell) \in \bS^2 \times \bR$, with  $u\doteq  1/(1+ e^{-\ell})$, over $V$.
Then we prove that  $\alpha^{X}$ is implemented in the bosonic Fock space $\cH'_\lambda$ by second-quantization of the strongly continuous one-parameter group of unitary operators
in the one-particle space $\sH'_\lambda$ defined by
\beq
V^{(\lambda')}_\tau \doteq e^{i \tau \hat{h}} \quad \mbox{with $(\hat{h}\psi)(\omega,h) \doteq h\psi(\omega,h)$ for  $(\omega,h) \in \bS^2 \times \bR$} \label{Vs}
\eeq
($\Dom(\hat{h})$ is the usual one for multiplicative operators). We will establish that $\lambda'$ is KMS with respect to $\alpha^X$ with $\beta = 2\pi$, 
so that the second-quantization of $V^{(\lambda')}$ (which satisfies  $V^{(\lambda')}_\tau \sK'_\lambda = \sK'_\lambda \beta^{X}_\tau$)
defines the modular group of $(\pi_{\lambda'}(\cW(V))'',\Psi_{\lambda'})$ in view of the link between KMS condition and modular group outlined at the beginning of Section \ref{sec5}. 
 Finally we will prove  that  $(\sK'_\lambda, \sH'_\lambda)$  is unitarily equivalent to that of $\lambda$, so that $\lambda=\lambda'$ 
in view of Proposition 3.1 in \cite{KW}.  In this picture, $V_\tau^{(\lambda)}$ is the unitary
 corresponding to $V_\tau^{(\lambda')}$ in the structure $(\sK_\lambda, \sH_\lambda)$. $\Box$

\subsection{The modular group of the double cone and its relation with that of $V$}
From Proposition \ref{propJ}, we know that, for every $m\geq 0$,  $\cW_m(D)$ identifies to a sub algebra of $\cW(V)$ in view of the  existence of the isometric $*$-homomorphism
$\ell_m^{VD}$  and that  $\widetilde{\cW}_m(D)$ identifies to the whole algebra of $\cW(V)$
due to the $*$-isomorphism $\ell_m^{V\widetilde{D}}$, but
it does not hold that $\ell_m^{VD}(\cW_m(D)) = \cW(V)$ and $\ell_m^{\widetilde{D}D}(\cW_m(D)) = \widetilde{\cW}_m(D)$.
However,  the three algebras $\cW_m(D), \widetilde{\cW}_m(D),  \cW(V)$ do coincide when promoted to von Neumann algebras in the GNS representations of 
$\omega_m$, $\widetilde{\omega}_m$ and $\lambda$ as we are going to discuss.  We start with a fundamental technical lemma proved in the appendix.

\begin{lemma}\label{lemmagoursat} For $m\geq 0$, if $\Phi \in \sS(V)$ and defining the corresponding
 $\phi \doteq  (L_m^{V\widetilde{D}})^{-1}(\Phi)$, so that $\phi \in \widetilde{S}(D)$,
there is a sequence $\{\phi_n\}_{n\in \bN} \subset \sS_m(D)$ such that both the limits hold
$\sK_\lambda L_m^{VD}(\phi_n) \to \sK_\lambda \Phi$ in $\sH_\lambda$ and 
$\widetilde{\sK}_m L_m^{\widetilde{D}D}(\phi_n) \to \widetilde{\sK}_m \phi$ in $\widetilde{\sH}_m$
as $n\to +\infty$, so that, in particular:
 \beq
 \widetilde{\sH}_m=\sH_m \quad \mbox{and}\quad \sK_m = \widetilde{\sK}_m L_m^{\widetilde{D}D}\:.\label{HHJJ}
 \eeq
 \end{lemma}
 
A consequence of the lemma concerns the existence of a unitary $Z_m$ that implements the (nonsurjective) embedding $\ell_m^{VD}: \cW_m(D) \to \cW(V)$ 
at the level of the GNS representations. 

\begin{proposition}\label{propZ}  For $m\geq 0$, consider the states 
$\lambda$ on $\cW(V)$,  $\omega_m$ on $\cW_m(D)$ and  $\widetilde{\omega}_m$ on $\widetilde{\cW}_m(D)$ with
 GNS representations $(\cH_m,  \pi_m, \Psi_m)$,  $(\cH_{\lambda},  \pi_{\lambda}, \Psi_{\lambda})$ and
 $(\widetilde{\cH}_{\lambda},  \widetilde{\pi}_{\lambda}, \widetilde{\Psi}_{\lambda})$
  respectively. The following properties hold:
  \begin{itemize}
 \item[$(a)$] There is a unique unitary operator $Z_m: \cH_m  \to \cH_{\lambda}$ such that
\begin{align}
&Z_m \Psi_m = \Psi_\lambda\:, \label{P1}\\
&Z_m \pi_m(A) Z_m^{-1} = \pi_\lambda(\ell_m^{VD}(A))\quad \mbox{for all $A \in \cW_m(D)$}\:.\label{P2}
\end{align}
 \item[$(b)$] Up to unitary equivalence $(\cH_{m},  \pi_{m}, \Psi_{m}) =
   (\widetilde{\cH}_{m},  \widetilde{\pi}_{m} \circ \ell_m^{\widetilde{D}D}, \widetilde{\Psi}_{m})$.\\
\item[$(c)$] Taking (b) and (\ref{HHJJ}) into account, the unitary $Z_m$ also verifies:
\begin{itemize}
\item[$(i)$] $Z_m  (\sH_m) = Z_m  (\widetilde{\sH}_m) = \sH_\lambda$

\item[$(ii)$] $Z_m \sK_m = \sK_\lambda L_m^{VD}$ and $Z_m \widetilde{\sK}_m = \sK_\lambda L_m^{V\widetilde{D}}$

\item[$(iii)$]  $Z_m \widetilde{\pi}_m(a) Z_m^{-1} = \pi_\lambda(\ell_m^{V\widetilde{D}}(a))$ for all $a \in \widetilde{\cW}_m(D)$.
\end{itemize}
\end{itemize}
\end{proposition}

\begin{proof} 
$(a)$ Define the $\bR$-linear map 
$z_m : \sK_m\phi  \mapsto \sK_\lambda L_m^{VD}(\phi)$ for every $\phi \in \sS_m(D)$.
First notice that $z_m$ is isometric. This follows decomposing the scalar products into real and imaginary part and then 
 using Lemma \ref{lemmalemme} for the real parts and the fact that the imaginary parts equal the corresponding symplectic forms which, in turn, are preserved 
 by the map $L_m^{VD}$ as discussed in the proof of Proposition \ref{propdefjmath}.  
As a consequence of Lemma \ref{lemmagoursat},  the space spanned by $\sK_\lambda L_m^{VD}(\phi) + i \sK_\lambda L_m^{VD}(\phi)$ for $\phi \in \sS_m(D)$ is dense in $\sH_\lambda$.
By Lemma A.2 in \cite{KW}, we conclude that $z_m$ extends by linearity and continuity to a unitary map from
 $\sH_m$ to $\sH_\lambda$. In turn, by standard 
procedures of second quantization in bosonic Fock spaces, this map uniquely extends by linearity and continuity to a unitary map from the Fock
space $\cH_m$ (with one-particle space structure $(\sK_m, \sH_m)$)
onto the Fock space $\cH_{\lambda}$  (with one-particle space structure $(\sK_{\lambda}, \sH_{\lambda})$) verifying (\ref{P1}) and (\ref{P2}). 
Since $\Psi_m$ is cyclic for $\pi_m$ and $Z_m$ is unitary, we conclude from (\ref{P1})
and (\ref{P2}) that $\Psi_\lambda$ has to be cyclic for $ \pi_\lambda(\ell_m^{VD}(\cW_m(D)))$. 
This result implies the uniqueness  of any unitary $Z_m: \cH_m \to \cH_{\lambda}$
satisfying (\ref{P1}) and (\ref{P2}). The proof of the statement $(b)$ immediately follows from (\ref{HHJJ}) and the definition of $\ell_m^{\widetilde{D}D}$ (and the known Fock space structure 
of the GNS representation for quasifree states). Concerning $(c)$, we observe that the $(i)$ and first identity in $(ii)$ are true from the construction of $Z_m$ above, also taking 
(\ref{HHJJ}) into account.
Afterwards, a unitary 
$Z'_m: \widetilde{\cH}_m (= \cH_m) \to \cH_{\lambda}$ verifying $(iii)$ and $Z'_m \widetilde{\Psi}_m (= Z'_m \Psi_m) = \Psi_\lambda$ can be constructed as 
before, thus verifying the second in $(ii)$. In view of (\ref{coer}), $Z'_m$ fulfills the same requirement satisfied 
(\ref{P1}) and (\ref{P2}) so it must coincide to $Z_m$ by the uniqueness property of the latter. \end{proof}

\noindent The unitary $Z_m$ gives a natural way to embed the representation $\pi_{m}$ into $\pi_{\lambda}$ in accordance with 
the (non surjective) embedding of algebras $\ell_m^{VD} : \cW_m(D) \to \cW(V)$. Passing to the corresponding von Neumann algebras this embedding becomes an identification that allows us to characterize  the modular group of $(\pi_m(\cW_m(D))'' , \Psi_m)$ (and that of  $(\widetilde{\pi}_m(\widetilde{\cW}_m(D))'' , \widetilde{\Psi}_m)$)
 in terms of that of $(\pi_{\lambda}(\cW(V))'' , \Psi_{\lambda})$.

\begin{theorem}\label{teommg} Referring to the three von Neumann algebras $(\pi_m(\cW_m(D))'' , \Psi_m)$, $(\widetilde{\pi}_m(\widetilde{\cW}_m(D))'' , \Psi_m)$ 
and  $(\pi_{\lambda}(\cW(V))'' , \Psi_{\lambda})$ in standard form and to the unitary operator $Z_m: \cH_m \to \cH_{\lambda}$ 
introduced in Proposition \ref{propZ}, it holds:
\beq   \:\pi_m(\cW_m(D))'' \:\:=\:\:  \pi_m(\widetilde{\cW}_m(D))''  \:\:=\:\: Z_m^{-1} \: \pi_{\lambda}(\cW(V))'' \: Z_m\:. \label{VN}\eeq
The modular group $\{\sigma^{(m,D)}_\tau \}_{\tau \in \bR}$ of $(\pi_m(\cW_m(D))'' , \Psi_m)$ and $(\widetilde{\pi}_m(\widetilde{\cW}_m(D))'' , \Psi_m)$
is 
\beq 
 \sigma^{(m,D)}_\tau(A) =   U^{(m)}_{-2\pi \tau} A U^{(m)*}_{-2\pi \tau} \quad \mbox{for every
 $A\in \pi_m(\cW_m(D))''$ and $\tau \in \bR$,}
\eeq 
where $U^{(m)}_{\tau} \doteq Z_m^{-1} U^{(\lambda)}_\tau Z_m$ and $U_\tau^{(\lambda)}$ is that in (b) of Theorem \ref{T3}.
\end{theorem}

\begin{proof} The identity $\pi_m(\widetilde{\cW}_m(D))'' = Z_m^{-1} \: \pi_{\lambda}(\cW(V))'' \: Z_m$ holds trivially from $(iii)$ of $(c)$ in Proposition \ref{propZ}
since $Z_m$ is unitary and $\ell_m^{V\widetilde{D}}$ is a $^*$-isomorphism.
We have to prove the remaining identity in  (\ref{VN}) only, since the 
last statement in the thesis is an immediate consequence of it and of the fact that $Z_m \Psi_m= \Psi_\lambda$ and of 
$(b)$ in Theorem \ref{T3}. Since, by basic properties of the commutant, $Z_m  \:\pi_m(\cW_m(D))'' \:Z_m^{-1}= (Z_m \pi_m(\cW_m(D))\:Z_m^{-1})''$, we have to 
prove that:
\beq \left( Z_m \pi_m(\cW_m(D))\:Z_m^{-1}\right)'' = \pi_{\lambda}(\cW(V))'' \:.\label{VN'}\eeq
We already know that
$\left( Z_m \pi_m(\cW_m(D))\:Z_m^{-1}\right)'' \subset \pi_{\lambda}(\cW(V))'' $, because $\ell_m^{VD}(\cW_m(D)) \subset \cW_\lambda(V)$.
To pass to (\ref{VN'}), using the fact that a von Neumann algebra is closed in the strong (operatorial) topology and that the space of finite liner combinations of elements 
$\pi_\lambda(W(\Phi))$, $\Phi \in \sS(V)$, is dense in $\pi_\lambda(W(\Phi))''$ in that topology, it is enough to prove that, for every $\Phi \in \sS(V)$, there is a sequence
$\{\phi_n\}_{n\in \bN} \subset \sS_m(D)$ such that $\pi_\lambda(W(L_m^{VD}(\phi_n))) \to \pi_\lambda(W(\Phi))$ in the strong operator topology because
$\pi_\lambda(W(L_m^{VD}(\phi_n))) = Z_m  \:\pi_m(W_m(\phi_n)) \:Z_m^{-1}$ by (\ref{P2}).
To conclude we observe that if, for a fixed $\Phi \in \sS(V)$, the sequence  $\{\phi_n\}_{n\in \bN} \subset \sS_m(D)$
is chosen as in Lemma \ref{lemmagoursat} so that $\sK_\lambda L_m^{VD}(\phi_n) \to \sK_\lambda \Phi$ in $\sH_\lambda$ as $n\to +\infty$, $(1)$ in
Proposition 5.2.3 in \cite{BR2} implies that  $\pi_\lambda(W(\ell_m^{VD}(\phi_n))) \to \pi_\lambda(W(\Phi))$, for $n\to + \infty$, 
in the strong operatorial topology as requested.
\end{proof}
\section{The spacetime action of the modular group for $m\geq 0$}\label{sec6}
Referring to the last statement of Theorem \ref{teommg}, in the following $\{V_\tau^{(m)}\}_{\tau \in \bR}$
denotes the one-parameter group of unitaries acting in the one-particle space $\widetilde{\sH}_m$
 whose second-quantization on the corresponding Fock space $\widetilde{\cH}_m$ (assuming the cyclic vector $\Psi_m$ to be invariant)  is the unitary $U^{(m)}$
 implementing   the modular group of
$(\widetilde{\pi}_m(\widetilde{\cW}_m(D))'', \Psi_m) $.
We aim to study the action of the modular group {\em on the KG solutions in the spacetime}, for $m \geq 0$.
In other words, we are interested in the existence of a one-parameter group
$\{ s^{(m)}_\tau\}_{\tau\in \bR}$
 of symplectic isomorphisms of $\widetilde{\sS}_m(D)$ such that:
\beq
V^{(m)}_\tau \widetilde{\sK}_m \phi = \widetilde{\sK}_m s^{(m)}_\tau(\phi) \quad \mbox{for every $\phi \in \widetilde{\sS}_m(D)$
 and $\tau\in \bR$.} \label{mgw}
\eeq
 If   $\{ s^{(m)}_\tau\}_{\tau\in \bR}$ exists, it is unique because $\widetilde{\sK}_m$ in (\ref{mgw})
is injective.
From (\ref{ONE}), the second identity  in (ii) in Proposition \ref{propZ} and the last statement of the thesis of Theorem \ref{teommg}, 
we have
$V^{(m)}_\tau \widetilde{\sK}_m = \widetilde{\sK}_m (L_m^{V\widetilde{D}})^{-1}\: \beta_\tau^X \: L_m^{V\widetilde{D}}$.
So that it must be:
\beq
s_\tau^{(m)} \doteq  (L_m^{V\widetilde{D}})^{-1}\: \beta_\tau^X \: L_m^{V\widetilde{D}} \label{stau2}
\eeq
Notice that the right-hand side is well defined, it is a symplectic isomorphisms of $\widetilde{\sS}_m(D)$ for every $\tau$, and
(\ref{mgw}) holds as a consequence of the given definition.\\
A first result is obtained taking Theorem \ref{teomz} into account for the case $m=0$: 
Since it holds $(L_0^{V\widetilde{D}})^{-1} \beta_\tau^X =\beta_\tau^X (L_0^{V\widetilde{D}})^{-1}\: $, inserting it in the right-hand side of (\ref{stau2}) with $m=0$,
we obtain:
\beq
s_\tau^{(0)} = \beta_\tau^X \label{stau20}\:.
\eeq
This is nothing but the known result obtained in \cite{HL} written in our language. 

\begin{remark} It should be clear that  we are working in the space $\widetilde{\sS}_m(D)$ rather than $\sS_m(D)$ because $L_m^{VD} : \sS_m(D) \to \sS(V)$ is not surjective and so $(L_m^{VD})^{-1}$,
that would appear in the corresponding of (\ref{stau2}), would not exist. This is the ultimate reason to enlarge  $\sS_m(D)$  to the space $\widetilde{\sS}_m(D)$.
This enlargement makes sense since it does not affect the relevant  von Neumann algebra and the modular group as already  established in Theorem \ref{teommg}.
\end{remark}

\subsection{The relative action of the modular group and its infinitesimal generator}
We pass to study the {\em relative action} of the modular operator on KG solutions, that is:
\beq 
s^{(m)}_\tau \phi - s^{(0)}_\tau \phi\quad \mbox{for $\phi \in \widetilde{\sS}_m(D)$,}\label{Delta}
\eeq
 giving its explicit expression.
Three preliminary remarks are necessary. $(1)$ $ s^{(0)}_\tau \phi$ is well defined 
for every  $\phi \in \widetilde{\sS}_m(D)$ even if $m>0$ since it coincides to $\beta^{X}_\tau \phi$ as remarked above. 
However, differently from $s^{(m)}_\tau \phi$, one may have
 $ s^{(0)}_\tau \phi \not \in \widetilde{\sS}_m(D)$ if $m>0$ and  $\phi \in \widetilde{\sS}_m(D)$.
 $(2)$ Every map $L_m^{V\widetilde{D}}: \widetilde{\sS}_m(D) \to \sS(V)$, for all values of $m$, can be seen as the restriction 
to the corresponding space $\widetilde{\sS}_m(D)$ of a common map $L : C_0^{\infty}(\bM) \to C^\infty(\bS^2 \times (0,1))$
 associating a smooth function $f$ to the function  $(Lf)(\omega,u) \doteq u f\spa\restriction_V(\omega,u)$. To do it, every element $\phi \in \widetilde{\sS}_m(D)$
 has to be smoothly extended to the whole $\bM$ (as is allowed by the very definition of $\widetilde{\sS}_m(D)$) and the function $L \phi$ does not depend on the extension. 
$(3)$ if $\phi \in \widetilde{\sS}_{m'}(D)$ then $L \phi \in \sS(V)$ so that $(L_m^{V\widetilde{D}})^{-1} (L\ph))$ is however well-defined regardless if $m\neq m'$.  More generally,
$(L_m^{V\widetilde{D}})^{-1} (L f))$ is well defined for $f\in C_0^\infty(\bM)$ if $L f \in \sS(V)$. 

\begin{remark}
From now on we simplify the notation as  $L_m \doteq L_{m}^{V\widetilde{D}}$.
\end{remark}

 \begin{theorem} \label{propQF} If $m\geq 0$, $\phi \in \widetilde{\sS}_m(D)$ and $(t,\bx) \in D$ then:
 \begin{align}
 &s^{(m)}_\tau \phi - s^{(0)}_\tau \phi =  \left[L_m^{-1}L - L_0^{-1}L,\beta^{X}_\tau\right] \phi \label{comm} \quad \mbox{where,
 if $f\in C_0^\infty(\bM)$ and $L f \in \sS(V)$}\\
&\left(\left( L_m^{-1} L- L_0^{-1}L\right)f \right)(t,\bx)   = \int_{\bS^2}  \int_0^{u^*(t,\bx,\omega)}  \sp\sp\sp  F_m(\sigma((t,\bx), (u, u\omega))) 
u \partial_u u f\!\!\restriction_V\!\!(\omega,u)du d\omega \label{comm2}
 \end{align}
\beq
\mbox{with   }\qquad u^*(t,\bx,\omega) \doteq \frac{t^2-\bx^2}{2 (t -\omega \cdot \bx)}\:,  \quad F_m(z) \doteq \frac{m^2}{8\pi}\sum_{k=0}^{+\infty} \left(\frac{m}{2}\right)^{2k} \sp \sp \frac{z^k}{k!(k+1)!}\:, z\in \bR  \eeq 
\end{theorem}
 
\begin{remark}\label{remarkadd} For fixed $(t,\bx) \in D$ and $\omega \in \bS^2$,  the real  $u^* \doteq u^*(t,\bx,\omega)$ is the 
unique solution of $\sigma((t,\bx), (u^*, u^*\omega)) =0$, 
so that $0< u^*<1$. The function $F_m$ , due to the rapid convergence of the series, is analytic since it is the restriction to $\bR$ of an entire analytic function. 
$F_m$ vanishes,  as expected, for $m=0$.
\end{remark}
 
\begin{proof}[Proof of Theorem \ref{propQF}] By definition of $s^{(m)}_\tau$ and since $L_m^{-1} L_m = I_{\widetilde{\sS}_m(D)}$, we immediately obtain 
  $s^{(m)}_\tau \phi - s^{(0)}_\tau \phi = L_m^{-1} \beta^X_\tau L_m \phi - \beta^X_\tau L_m^{-1} L_m \phi$. That is
  $s^{(m)}_\tau \phi - s^{(0)}_\tau \phi = L_m^{-1} \beta^X_\tau L \phi - \beta^X_\tau L_m^{-1} L \phi$. Now notice that 
  $\beta_\tau^X L f = L \beta^X_\tau f$, essentially because $X$ is tangent to $V$ and $L$ is the restriction to $V$
  (however it can be checked by direct inspection). We have found that $s^{(m)}_\tau \phi - s^{(0)}_\tau \phi =
   L_m^{-1}  L \beta^X_\tau \phi - \beta^X_\tau L_m^{-1} L \phi$, where now, in general,
   $L_m^{-1}  L \beta^X_\tau \phi \neq L_m^{-1}  L_m \beta^X_\tau \phi  = \beta^X_\tau \phi$ 
   since $\beta^X_\tau \phi \not \in \widetilde{\sS}_m(D)$ if $m>0$. We can re-arrange the found result as
   $s^{(m)}_\tau \phi - s^{(0)}_\tau \phi =   (L_m^{-1}  L -   L_0^{-1}  L)\beta^X_\tau \phi 
 -   ( \beta^X_\tau L_m^{-1}  L -  L_0^{-1}  L \beta^X_\tau) \phi$.
The found identity gives rise to (\ref{comm}) if observing that $L_0^{-1}  L \beta^X_\tau=  L_0^{-1}  \beta^X_\tau L
 = \beta^X_\tau L_0^{-1} L$ where, 
   in the last passage, we have used Theorem \ref{teomz}.  Let us pass to (\ref{comm2}) exploiting
   (\ref{ntempE}). If $f\in C_0^\infty(\bM)$ and $L f \in \sS(V)$:
 \begin{align*}&\left(\left(L_m^{-1} L- L_0^{-1}L\right)f \right)(t,\bx)  = \\
&\qquad\qquad\qquad\qquad \lim_{n\to +\infty} \sp -2 \int_{\bS^2}\int_0^1 (\Delta_{m} - \Delta_{0})  ((t-i\epsilon_n,\bx), (u, u\omega)) u \partial_u (L f)(\omega,u)du d\omega \ .\end{align*}  
The integral defining the advanced and retarded fundamental solutions of Klein-Gordon equation \cite{RS2} produces, for $m>0$:
  \begin{align*}&\Delta_m((t-i\epsilon_n,\bx), (t', \bx')) =\\
  &\qquad\qquad\quad \frac{im}{4\pi^2}\left[  \frac{K_1 \left(m \sigma((t+i\epsilon_n,\bx), (t', \bx'))^{1/2} \right)}{\sigma((t+i\epsilon_n,\bx), (t', \bx'))^{1/2}} -
    \frac{K_1 \left(m \sigma((t-i\epsilon_n,\bx), (t', \bx'))^{1/2} \right)}{\sigma((t-i\epsilon_n,\bx), (t', \bx'))^{1/2}}\right]\ ,
    \end{align*}
 the analog for $m=0$ appears in (\ref{delta0}).
 Using the standard expansion of the Bessel functions $K_1$ and $I_1$  (see 8.446 and 2 of 8.447 in \cite{Grad}), one sees that, 
 at $(t,\bx) \in D$ fixed,  $-2(\Delta_{m} - \Delta_0)((t-i\epsilon_n,\bx), (u, u\omega)) $ is the smooth function 
 $F_m(\sigma((t,\bx), (u\omega,u)))$ multiplied with a function weakly
 tending to $\theta(-\sigma( (t,\bx), (u, u\omega)) )$ as $n\to +\infty$, where $\theta(x)=0$ for $x<0$ and $\theta(x)=1$ otherwise. This restricts
 the final integration to $\omega \in \bS^2$ and $u \in [0, u^*(t,\bx, \omega)]$ giving rise to (\ref{comm2}) almost everywhere in $\bx$. Since both 
 members of (\ref{comm2}) are continuous in $(t,\bx)$, the identity holds everywhere.
 \end{proof}

\noindent If $\phi \in \widetilde{\sS}_m(D)$, the function
  $(\tau, t,\bx) \mapsto (s_\tau^{(m)} \phi)(t, \bx)$ is jointly smooth for $\tau \in \bR$ and $(t,\bx) \in D$: The first and, for Theorem  \ref{propQF}, the third term of the sum in the right-hand side of
  $$s^{(m)}_\tau \phi = s^{(0)}_\tau \phi + (L_m^{-1}L - L_0^{-1}L)\beta^{X}_\tau \phi - 
\beta^{X}_\tau   (L_m^{-1}L - L_0^{-1}L) \phi$$
are smooth by construction, the second one is smooth as it can be trivially proved by recursively deriving the integral in (\ref{comm2})  using Leibniz rule
and the Theorem of derivation under the symbol of integration. So, it makes sense to compute the {\em infinitesimal generator} $\delta^{(m)}$ of $s^{(m)}_\tau$, taking the derivative of
$s^{(m)}_\tau \phi$ at $\tau=0$. If $\gamma^X$ is defined as in (\ref{delta}) with $Y$ replaced for $X$, (\ref{comm})  leads to:
\beq
\delta^{(m)} \phi =  
\gamma^X\phi +  \left[L_m^{-1}L - L_0^{-1}L,\gamma^X \right] \phi  \quad 
\mbox{for every $\phi \in \widetilde{\sS}_m(D)$, $m\geq 0$.}\label{deltam}
\eeq
The right-hand side of (\ref{deltam}) can explicitly be computed taking (\ref{comm2}) into account.\\
Notice that, for $m=0$, (\ref{deltam}) produces, as it has to do,
\beq
\delta^{(0)} \phi =  
\gamma^X\phi\:.
\eeq

\begin{theorem}\label{tquasiultimo} For every $\phi \in \widetilde{\sS}_m(D)$, $m\geq 0$, it holds:
\beq
(\delta^{(m)} \phi) (t,\bx) =  (\gamma^{X} \phi)(t,\bx) + \int_{\bS^2} \int_0^{u^*(t,\bx,\omega)}\sp\sp\sp \:\:\: I_\phi(t,\bx,u,\omega,\sigma)\ du d\omega \label{deltamm}
\eeq
where  $\sigma \doteq \sigma((t,\bx), (u\omega,u))$ for $(t,\bx) \in D$ and
\beq I_\phi(t,\bx,u,\omega,\sigma)\doteq u [2-(t+u)] \:[F_m(\sigma) + 
\sigma  F_m'(\sigma)]\: \partial_u u \phi\!\! \restriction_V\!\!(\omega,u)\:.\label{integrando}
\eeq 
\end{theorem}

\begin{proof}
The commutator in the right-hand side of (\ref{deltam}) can be explicitly computed using \eqref{comm2} and the identity \eqref{trick} for $Y=X$, that permits to pass the symbol $\gamma^X$ from the variables $(t,\bx)$ to the variables $(u,u\omega)$ in $V$ erasing some terms. The computation is long but straightforward and produces just
 the result (\ref{deltamm}) up to a further added term to  the right-hand side:
\beq
-\int_{\bS^2} d\omega  F_m(0)  \left\{ u \left[  u(1-u) + X^a(t,\bx) \partial_a u^*(t,\bx, \omega) \right] \partial_u u \phi\spa\restriction_V (\omega,u) \right\}_{u= u^*(t,\bx, \omega)}\:.
\label{termagg}
\eeq
This term vanishes for the following reason. By the definition of $u^*(x, \omega)$ (see Remark \ref{remarkadd}), $\sigma(x, (u^*(x, \omega), u^*(x,\omega) \omega))=0$
for all $x\equiv (t,\bx) \in D$. Therefore we also have that, applying $X$ to the complete dependence on $x$,
 $X^a(x) \frac{\partial}{\partial x^a} \sigma|_{u=u^*} + X^a(x) \frac{\partial u^*}{\partial x^a} \frac{\partial}{\partial u} \sigma|_{u=u^*}=0$. Making use of (\ref{trick0}), 
it  implies $(u(1-u) \partial_u \sigma) + ((t-1) + (u-1)) \sigma |_{u=u^*} + X^a(x) \frac{\partial u^*}{\partial x^a} \frac{\partial}{\partial u} \sigma|_{u=u^*}=0$.
 Since $\sigma|_{u=u^*} =0$,  and $\sigma = -t^2+\bx^2 + 2u (t-\bx \cdot \omega)$, the found identity is nothing but:
 $\left[ u(1-u) + X^a(t,\bx) \partial_a u^*(t,\bx, \omega)\right]|_{u=u^*} (t-\bx \cdot \omega)=0$. However
 $t-\bx \cdot \omega > 0$ because $t>||\bx||$ in $D$
 and $||\omega||=1$. Therefore $\left[ u(1-u) + X^a(t,\bx) \partial_a u^*(t,\bx, \omega)\right]|_{u=u^*} =0$ and the integral in (\ref{termagg}) vanishes in any case. 
\end{proof}

\subsection{$\delta^{(m)} -\delta^{(0)}$ is a pseudo-differential operator of class $L^{0}_{1,1}$ }
We consider the difference of generators $\delta^{(m)} -\delta^{(0)}$ (that is $\delta^{(m)} -\gamma^{X}$)  as an operator  
$C_0^{\infty}(\bM)\to C^\infty(D)$ in view of the fact that 
the integral in the  right-hand side of (\ref{deltamm}) is well-defined also for $\phi$ replaced with a generic $f\in C_0^{\infty}(\bM)$, though its physical meaning is guaranteed only when it acts on 
$\phi \in \widetilde{\sS}_m(D)$.
Hence we consider, for $f\in C_0^\infty(\bM)$ and $(t,\bx)\in D$:
\beq
\left((\delta^{(m)} -\delta^{(0)})f\right)(t,\bx) \doteq
\int_{\bS^2}\sp\spa d\omega \int_0^{u^*(t,\bx,\omega)} du\ 
I_f(t,\bx,u,\omega,\sigma)\ ,\label{dd}
\eeq
with $I_f(t,\bx,u,\omega,\sigma)$ as in \eqref{integrando}.

We pass to prove that $\delta^{(m)} -\delta^{(0)}$  is a {\em pseudodifferential operator} of class $L^{0}_{1,1}(D)$,
 proving and making more precise Fredenhagen's conjecture mentioned in the introduction.

\remark\label{rempseudo} We address the reader to specialized texts (e.g. \cite{Tay},\cite{GriSjoes},\cite{duistermaat})  for more information 
on this subject while recalling  just two relevant facts here.

 $(1)$ A {\em pseudodifferential operator}
$A: C_0^\infty(\bR^n)\to C^\infty(X)$ -- where $X\subset\bR^n$ is open --
is individuated 
by its {\em amplitude},  a smooth function on $a: X\times\bR^n\to \bC$,
which satisfies
\beq
\left |\partial^\alpha_x \partial^\beta_\theta a(x,k)\right | \le C_{\alpha,\beta,K}(a) 
(1+\langle k, k \rangle^{1/2})^{m-\rho |\beta|+\delta |\alpha|}
\label{stimeA}
\eeq
for every compact set $K\subset X$, every pair of multindices $\alpha,\beta\in\bN^n$ 
and  corresponding constants $C_{\alpha,\beta,K}(a)\ge 0$. 
Above $|\alpha| \doteq \alpha_1 +\ldots +\alpha_n$ when 
$\alpha = (\alpha_1,\ldots,\alpha_n)$ and $\langle\cdot,\cdot \rangle$ denotes the standard scalar product in $\bR^n$.
Then, by definition, the  {\em pseudodifferential operator} associated to $a$ is
\beq 
(Au)(x) \doteq\int_{\bR^n} a(x,k) e^{i\langle \theta,x\rangle} \widehat{u}(k) dk \quad 
\mbox{if $u \in C_0^\infty(\bR^n)$,}\label{PSA}
\eeq
where $\widehat{u}$ indicates the Fourier transform of $u$.
Holding (\ref{stimeA}),
we  say that  {\em the amplitude $a$ belongs to the functional class $S^m_{\rho,\delta}(X\times\bR^n)$}
and, equivalently, that {\em the pseudodifferential operator $A$ belongs to the class $L^m_{\rho,\delta}(X)$}.

$(2)$ Compositions of pseudodifferential operators can be made in various fashions. However we only stick to the following
 simple result (see, for instance, Proposition $3.3$ in \cite{Tay}). Given $A_i\in L^{m_i}_{\rho_i ,\delta_i}(X)$ $i=1,2$, 
suppose that $0\le\delta_2<\rho\le 1$, where $\rho=\min\{\rho_1,\rho_2\}$, then 
$A_1 A_2$ is well defined and belongs to $L^{m_1+m_2}_{\rho,\delta}(X)$ where $\delta=\max\{\delta_1 ,\delta_2\}$.\\

To prove that $\delta^{(m)} -\delta^{(0)}$  is a {\em pseudodifferential operator} of class $L^{0}_{1,1}(D)$, it would be sufficient
 to recast \eqref{dd} into the form (\ref{PSA}) where, of course 
$a\in S^m_{\rho,\delta}(X\times\bR^n)$ with $m=0$ and $\rho=\delta=1$. We will indirectly obtain that
 result exploiting the composition rule mentioned in (2) in Remark \ref{rempseudo}.
We preventively need two lemmata concerning the restriction of a smooth function on $V$ and the 
form of $\delta^{(m)} -\delta^{(0)}$ as integral operator.\\
 Henceforth we equip $V$
with the topology induced by $\bM$ that, in turn, coincides with that of the identification $V \equiv \bS^2 \times [0,1)$,  and the associated Borel $\sigma$-algebra.
 
\begin{lemma} \label{lemmarest}  If $g : V \to \bR$ is measurable and bounded and $f \in C_0^\infty(\bM)$ then:
\beq
\int_{V} g f\spa \restriction_V du d\omega 	= \frac{1}{(2\pi)^4} \int_{\bR^4} dk \int_{V}  e^{i(k_0 u +  \bk \cdot u\omega)}
g(u,\omega) du d\omega \int_{\bR^4} e^{-i \langle k,y \rangle} f(y) dy \label{rest}
\eeq
where $k=(k_0,\bk)$, $y=(t',\bx')$ and $\langle k,y \rangle \doteq k_0t' + \bk \cdot \bx'$.
\end{lemma}  

\begin{remark}  The map $\kappa: \bS^2 \times (0,1) \ni (\omega,u) \mapsto (u, u\omega) \in \bM$ is a smooth embedding into $\bM$ of a cylinder  with image 
given by  the embedded submanifold $V^* \doteq V \setminus \{o\}$ of $\bM$. This embedding,
 in fact, gives rise to {\em a blow up} of the singularity at the tip $o$ of $V$.  Since $o$ has zero measure with respect to $du d\omega$, removing the tip of the cone
 does not affect the integration on $V$ and it permits to exploit  standard results about restrictions of smooth functions to embedded manifolds $V^*$.
 \end{remark}

\begin{proof}[Proof of Lemma \ref{lemmarest}] 
 The pull back action of $\kappa$ on  $f \in C_0^\infty(\bM)$, 
 interpreted as a restriction $f \mapsto f\spa\restriction_{V^*}$ can be represented as a   {\em Fourier integral operator}  \cite{duistermaat}.
 In practice, if   $g\in C_0^\infty(\bS^2 \times (0,1))$ we have 
 $$\int_{V^*} g f\spa \restriction_V du d\omega 	= \frac{1}{(2\pi)^4} \int dk d\omega du dx     e^{-i \langle k,x - (u, u\omega) )\rangle}
g(u,\omega) f(x)$$ where  the right-hand side has to be interpreted as  an {\em oscillatory integral}  \cite{duistermaat}:
Inserting a further smooth  compactly supported function $\chi(\epsilon k)$ in the integrand, with $\chi=1$
 in a neighborhood of $k=0$, and taking the limit $\epsilon \to 0^+$. Rearranging the integrals,
the limit produces the right-hand side of (\ref{rest}), using the fact that $ \int_{\bR^4} e^{-i \langle k,y \rangle} f(y) dy$ belongs to Schwartz space in the variable $k$,
that $ \int_{V^*}  e^{i(k_0 u +  \bk \cdot u\omega)}
g(u,\omega) du d\omega$ is bounded in the variable $k$, and exploiting Lebesgue dominated convergence theorem.
 So (\ref{rest}) holds true if  $g\in C_0^\infty(\bS^2 \times (0,1))$
and $f \in C_0^\infty(\bM)$. By a standard corollary of Luzin's theorem, if the function $g: \bS^2 \times (0,1) \to \bC$ is bounded by $M <+\infty$
 and is Borel-measurable, there is 
a sequence of compactly-supported continuous functions $g_n$, all bounded by $M$, with 
 $g_n \to g$ almost everywhere, for $n\to +\infty$, with respect the natural Borel measure on $\bS^2 \times (0,1)$. Stone-Weierstrass theorem implies that the $g_n$
 can be chosen in $C_0^\infty(\bS^2 \times (0,1))$. Using the sequence of these $g_n$, (\ref{rest}) extends to the case of $g$ bounded and measurable, exploiting 
  Lebesgue dominated convergence theorem twice,  the fact that $ \int_{\bR^4} e^{-i \langle k,y \rangle} f(y) dy$ belongs to Schwartz space in the variable $k$,
  the uniform bound $M$ for $g$ and the $g_n$ and the further bound of
 $\int_{V^*}  e^{i(k_0 u +  \bk \cdot u\omega)} g(u,\omega) du d\omega$  and all the $\int_{V^*}  e^{i(k_0 u +  \bk \cdot u\omega)}
g_n(u,\omega) du d\omega$  by $4\pi M$.   We have proved that (\ref{rest}) holds with $V$ replaced for $V^*$ in both sides. However, 
since $o$ has zero measure with respect to $du d\omega$, the integration over $V^*$ can be extended to $V$
without affecting the final result.
\end{proof}
 

 \begin{lemma} The operator $\delta^{(m)} -\delta^{(0)}: C_0^{\infty}(\bM)\to C^\infty(D)$ in \eqref{dd} can be written as
 \beq
\left((\delta^{(m)} -\delta^{(0)})f\right)(x) = \int_{\bR^4} dk \:\:
e^{i\langle k, x\rangle} b(x,k) \:\widehat{Z(f)}(k) \quad \mbox{for $f \in C_0^{\infty}(\bM)$}  \label{L1}
 \eeq
 with $x=(t,\bx)$, $k=(k_0,\bk)$ and (with notation as in Theorem \ref{propQF}):
 \begin{eqnarray}
 b(x,k)  &\doteq&  e^{-i \langle k,x \rangle} \int_{\bS^2\times [0,1]} d\omega 
 du e^{i(k_0 u + u\bk \cdot \omega)} 
 g(x,\omega,u)\:, \label{AAA}\\
g(x, \omega,u) &\doteq& (2\pi)^{-4}\theta(u^*(t,\bx,\omega) -u) u [2-(t+u)] [F_m(\sigma) + \sigma F_m'(\sigma)]\:,\label{L2}\\
Z &\doteq&  1+ t\partial_{t} + \bx \cdot \nabla_{\bx}\label{ZED}\:.
\end{eqnarray}
Furthermore  $b \in S^{-1}_{1,1}(D \times \bR^4 )$.
\end{lemma}

\begin{proof} By direct inspection, one sees that
 $Z(f)\spa\restriction_{V^*} =  \partial_u u (f\spa\restriction_{V^*})$.
Next, for every $x \in D$, one has
\beq
\left((\delta^{(m)} -\delta^{(0)})f\right)(x) = \int_{\bR^4} dk \int_{\bR^4} dy \:\:,
e^{i\langle k, x-y\rangle} b(x,k) \:Z(f)(y) \quad \mbox{for $f \in C_0^{\infty}(\bM)$} \label{L1b}
\eeq
by (\ref{rest}) in the case of $g= g(x, \omega,u)$ as in (\ref{L2}), if taking (\ref{dd})
into account and where $b(x,k)$ is defined in \eqref{AAA}.
The right-hand side of (\ref{L1b}) can equivalently be written as (\ref{L1}).
Let us prove that $b \in S_{11}^{-1}(D)$. 
By definition $b \in C^\infty(D\times \bR^4)$, as one can easily check. Furthermore 
$$
|b(x,k)| = \left|\int_{\bS^2\times [0,1]} d\omega 
 du e^{i(k_0 u + u\bk \cdot \omega)} u g_0(x, \omega,u)\right|\ ,
$$
 where $|g_0(x, \omega,u)| \leq M_K< +\infty$ for  $(x,\omega,u) \in K\times \bS^2\times [0,1]$ for every compact $K\subset D$, and the derivatives 
 of $g_0$, which exists almost everywhere, are similarly bounded. Using polar coordinates with polar axis $z$ along $\bk$, to evaluate the integral above, integrating  
by parts one sees that $||\bk||\: |b(x,k)| \leq C'_K< +\infty$ if $k \in \bR^4$ and 
$x\in K$. Similarly, reducing to the previous case by 
integrating by parts again, one finds
 $|k_0| \: |b(x,k)| \leq C''_K< +\infty$ if  $k \in \bR^4$ and $x\in K$. 
Thus $$|b(x,k)| \leq C_K(1+\langle k, k\rangle^{1/2})^{-1}\:.$$
Let us improve this estimate passing to consider the derivatives.
By Leibniz and the differentiation rules under the integral sign, we get that
$\partial_{x_i}^{\alpha_i} b(x,k)$ with $\alpha_i \geq 1$ can be decomposed as a sum of terms as follows
\beq
\partial_{x_i}^{\alpha_i} b(x,k) = \int_{\bS^2} d\omega \left\{I_1(x,k,u^*,\omega) +I_2(x,k,u^*,\omega)+I_3(x,k,u^*,\omega)\right\}
\eeq
where, respectively,
\begin{align*}
I_1(x,k,u^*,\omega)&= \sum_{\gamma_i =0}^{\alpha_i -1} c_{\alpha_i ,\gamma_i} \ \partial_{x_i}^{\gamma_i}  I(x,k,u^*,\omega) \partial_{x_i}^{\alpha_i -\gamma_i}u^* (x,\omega)\ ,\\
I_2(x,k,u^*,\omega)&= \sum_{\lambda_i =1}^{\alpha_i -1}\sum_{\mu_i=0}^{\alpha_i-\lambda_i-1} \ d_{\lambda_i,\mu_i,\alpha_i}\ \partial_{x_i}^{\mu_i}\bigl[\left(\partial_{x_i}^{\lambda_i} I(x,u,k,\omega)\right)|_{u=u^*}\bigr]\ \partial_{x_i}^{\alpha_i-\lambda_i-\mu_i}u^*(x,\omega)\ ,\\
I_3(x,k,u^*,\omega)&= \int_0^{u^*(x,\omega)} \ du\ \partial_{x_i}^{\alpha_i} I(x,k,u,\omega)\ ,
\end{align*}
where, $c_{\alpha_i,\gamma_i},d_{\lambda_i,\mu_i,\alpha_i}$ are some numerical coefficients, and finally, $I(x,k,u,\omega)$ represents the integrand function in \eqref{AAA}.

By a direct inspection of the formulas it is clear that the dominant term w.r.t. the
 variable $k$, for $\langle k, k\rangle^{1/2}\ge 1$, comes from the term in $I_1$
 with the highest power of derivatives of the function $I(x,u^*,k,\omega)$. Indeed, the derivatives act on the exponential term because of the $u^*$-dependence on the space variables, dropping down terms depending solely on the variable $k$. By differentiating w.r.t. all variables $x_i$, $i=1,\cdots, 4$, one gets a polynomial in $k$ of order $|\alpha|-1$ with smooth coefficients. Differentiating w.r.t $k_i$, $i=1,\cdots,4$, $\beta_i$-times, respectively, gives a polynomial in $k$ of order $|\alpha|-|\beta|-1$, with smooth coefficients as well.
Hence we get, considering $K\subset D$ compact, that there exists $C_{\alpha,\beta,K}(a)\ge 0$ for which
\beq 
\left |\partial_k^\beta \partial_{x}^{\alpha} b(x,k)\right |\le C_{\alpha,\beta,K}(a) (1+\langle k, k\rangle^{1/2})^{-1 + |\alpha|-|\beta|}
\eeq
holds true, so that $b\in S^{-1}_{1,1}(D)$ by definition.
\end{proof}
\noindent We are in place to establish the main result of this section proving and going beyond Fredenhagen's conjecture.

\begin{theorem}
The operator $\delta^{(m)} -\delta^{(0)}$ is a pseudodifferential operator of class $L^0_{1,1}(D)$. 
\end{theorem}

\begin{proof}
It holds   $Z \in L^{1}_{1,0}(D)$.  As shown above $b \in S^{-1}_{1,1}(D \times \bR^4 )$, 
hence by \eqref{L1b} and the recalled  composition property of pseudodifferential 
operators ((2) in Remark \ref{rempseudo}), we have that the operator in the right-hand side of \eqref{L1b} is a pseudodifferential operator of the class $L^0_{1,1}(D)$, since in our case $m_1=-1$ $m_2=1$ and $\rho_1=\rho_2=1$ and $\delta_1=1$ $\delta_2=0$. 
\end{proof}

\section{Conclusions and Outlook}\label{sec7}
We have shown how to build up, in the symplectic space of KG solutions, the analytical form of the one-parameter group of symplectic isomorphisms  corresponding 
to the modular automorphism group of von Neumann  algebra of a free massive scalar theory in the vacuum of Minkowski spacetime localized in double cones.
The same procedure has lead to an explicit formula for the  generator of that group as well. 
Many different kind of techniques contributed to achieve this goal: Mainly ideas coming from the correspondence between the Cauchy and Goursat 
problems for Klein-Gordon partial differential operators, some ideas coming from the algebraic approach to quantum field theory, even in curved spacetime,
 relating theories
on different regions of spacetime and ideas coming from the microlocal approach to distributions. 

A remarkable result concerns the fact that the theory in the bulk and that on the boundary are unitarily equivalent when promoted to von Neumann algebras theories so that, the modular groups are identified through the same unitary intertwiner.
In particular, referring to the theory defined on the boundary, the modular group preserves a conformal-geometric meaning regardless the presence of a mass for the theory in the bulk where, conversely, any geometric meaning of the modular group is absent (for nonvanishing mass). 
 
One of  the main results has however been the classification of the difference of the generators of the modular groups for the massive and 
massless cases as a pseudodifferential operator acting on the solution space of the Klein-Gordon equation. This achievement 
goes much beyond the conjecture stated in the literature as soon as  it makes the statement more precise even concerning the type of the operator, as it has been established to be of order $0$ and type $1,1$. This kind of pseudodifferential operators seems to appear rather frequently in the discussions around nonlinear hyperbolic equations, see e.g. \cite{Hoe}.
  
 Two  natural directions deserve to be further investigated. One regards the extension of some of our achievements to  curved  spacetimes.
Switching on the curvature, at least dealing with geodesically convex neighborhoods, the geometric picture is not very far from that in Minkowski spacetime. 
 At the level of Weyl algebras -- and it can be extended to include the algebras of Wick polynomials -- the main features  of the bulk-boundary interplay for geodesic double cones survive the appearance of the curvature,
 as recently established in \cite{DPP}. In that case, states strictly analogous to $\lambda$ can be defined on the boundary algebra 
 and the induced states in the bulk -- corresponding 
 to our state $\omega_m$ -- turn out to be of Hadamard class.  It would be interesting to focus on the features of modular groups of the respective von Neumann algebras.

 The other possible generalization could obviously consist of relaxing the requirement of having a free quantum field. In that case, however, the symplectic structure, which plays the crucial role in building up 
 the bulk-boundary correspondence would not exist, stopping the construction at the first step. It seems plausible that a suitable extension of the time-slice axiom to null Cauchy surfaces may help to cover the gap.
    
As recalled in the introduction, there are many
 other ideas related to the results presented in this paper. A particularly nice thing to do would be to try to push even further the 
 intuition that the modular group is a form of dynamics, by going to precisely define, in the spirit of the paper, the idea of a 
  ``local dynamics'' generated by these modular groups and study its properties. We have some interesting preliminary result, 
  at the moment only at the level of symplectic spaces and we hope to report soon on them. Another appealing idea would be 
  to try to study the properties of the group generated by all the (local) modular operators. 
  
\section*{Acknowledgements} The authors thank D. Buchholz, K. Fredenhagen, R. Giloni, N. Pinamonti and B. Schroer  for useful discussions.


\appendix
\section{Proof of some propositions}

\begin{proof}[Proof of Lemma \ref{lemmafourier}] 
Here $\Sigma_D$ is the 
restriction to $D$ of the Cauchy surface for $\bM$ at $t=1$ so that $C$ is the boundary  $\partial \Sigma_D$ of $\Sigma_D$ in $\Sigma \equiv \bR^3$. 
If $\phi \in \sS_m(D)$,  its Cauchy data on $\Sigma$ are in $C_0^\infty(\Sigma_D)$ so that they also belong to the Schwartz space on $\bR^3$, hence we can define
$\widehat{\phi}$ as in (\ref{nFMM}).
 Since $m>0$,  $\widehat{\phi}$ belongs again to the Schwartz space.
Thus (\ref{nFMM''}) and (\ref{nFMM'}) hold true interpreting the integrals in the standard way,
 since the smooth function defined in the right-hand side of  (\ref{nFMM''}) must coincide with $\phi$ (more precisely, $L^{m}_{\bM D} \phi$) everywhere on $\bM$,
in view  of the standard uniqueness theorem for the smooth Cauchy problem for Klein-Gordon equation (e.g. Theorem 10.1.2  in \cite{Wald1}); 
 indeed  both functions  satisfy the Klein-Gordon equation in $\bM$ 
 and match the same Cauchy data at $t=0$ on the whole $\Sigma$.   
If $\phi \in \widetilde{\sS}_m(D)$ the situation is more complicated. 
However, in view of the definition of $\widetilde{\sS}_m(D)$,
(a) $\phi(1,\cdot)$ and $\partial_t\phi(1, \cdot)$ are smooth in $\Sigma_D \cup \partial \Sigma_D$
 (which has finite Lebesgue measure)  vanishing outside it and (b) $\phi(1,{\bx})=0$ if ${\bx} \in \partial \Sigma_D$ and outside $\Sigma_D \cup \partial \Sigma_D$ on $\Sigma$. 
From (a)  $\partial_t\phi(1, \cdot)\in L^2(\bR^3,d{\bx})$ is bounded and has compact support
and thus its Fourier transform $\widehat{\partial_t\phi(1, \cdot)}(\bk)$ is bounded  and  
$\widehat{\partial_t\phi(1, \cdot)}(\bk), E(\bk)^{-1}\widehat{\partial_t\phi(1, \cdot)}(\bk), E(\bk)^{-1/2}\widehat{\partial_t\phi(1, \cdot)}(\bk)$ belong to 
$L^2(\bR^3, d\bk)$ (even for $m=0$).
From (a) and (b)  $\phi(1,\cdot) \in H^1_0(\Sigma_D) \subset H^1(\bR^3)$ by the Sobolev trace theorem 
(e.g. see Theorem 9.5.1 in \cite{Aubin}). Consequently 
$\widehat{\phi(1, \cdot)}(\bk), E(\bk)\widehat{\phi(1, \cdot)}(\bk)$ belongs to $L^2(\bR^3, d\bk)$.
We conclude from (\ref{nFMM}) that $\widehat{\phi}$ belongs to $L^2(\bR^3, d{\bk})$ when  $\phi \in \widetilde{\sS}_m(D)$, the right-hand sides of (\ref{nFMM''}) and (\ref{nFMM'}), 
{\em henceforth  indicated by  $\psi$ and $\partial_t\psi$  respectively}, 
are well defined {\em if we interprete the integrals in the sense of Fourier-Plancherel transform}.
By direct inspection one sees that the right-hand side of (\ref{nFMM''}),  $\psi$, is a weak solution of Klein-Gordon equation in $\bM$. 
Furthermore $\psi \in C([1,T]; H^{1}(\bR^3)) \cap C^1([1,T]; H^0(\bR^3))$, for every $T>1$, and it admits as initial conditions for $t=1$
just the functions $\phi(1,\cdot) \in H^1(\bR^3)$ and $\partial_t\phi(1,\cdot)\in H^0(\bR^3)$ (extended to the zero functions outside
 $\Sigma_D$).  We want to prove that $\psi = \phi$ in $D$. First consider the set $[1,+\infty) \times \bR^3$, that is the causal future of $\Sigma_D$.
We can smoothly extend $\phi \in \widetilde{\sS}_m(D)$ in the whole set $J^+(o)$ as the unique solution of the Goursat problem
(Theorem 5.4.2 in \cite{friedlander})
 with characteristic datum $u^{-1}L^m_{V\widetilde{D}}\phi$ 
(smoothly) extended to the zero function outside $\partial J^+(o)\setminus V$. Then, we can further extend $\phi$ to the the half space  
$[1,+\infty) \times \bR^3$, simply defining 
the extension as the null function outside $J^+(o)$. This extension, by construction, is a smooth solution of the Klein-Gordon equation in $(1,+\infty) \times \bR^3$ away from
$\partial J^+(o)$. However it is a weak solution of the Klein-Gordon equation  in a neighborhood of $\partial J^+(o)$ in $(1,+\infty) \times \bR^3$. For a test function 
$\kappa \in C_0^\infty(\bM)$, it  can be verified by direct inspection 
using the fact that the boundary term on  $\partial J^+(o)$ arising by employing the integration by parts and Stokes-Poincar\'e lemma is proportional to
$\int_{\partial J^+(o) \cap ((1,+\infty) \times \bR^3)}  d\omega du (\kappa \partial_u \phi - \phi \partial_u \kappa)  =0$, since the (extended) function $\phi$ vanishes on
 $\partial J^+(o) \cap ([1,+\infty) \times \bR^3) = \partial J^+(o)  \setminus V$
with its derivative $\partial_u$ (i.e. along the affine parameter $u$ describing the null geodesics forming $\partial J^+(o)$). Summing up, the extended function $\phi$
on the whole $[1,+\infty) \times \bR^3$ is a weak solution of Klein-Gordon equation in $(1,+\infty) \times \bR^3$  that matches the initial conditions 
 $\phi(1,\cdot) \in H^1(\bR^3)$ and $\partial_t\phi(1,\cdot)\in H^0(\bR^3)$ (extended to the zero functions outside
 $\Sigma_D$). One sees directly that the constructed extended $\phi$ is in the class
$C([1,T]; H^{1}(\bR^3)) \cap C^1([1,T]; H^0(\bR^3))$ for every $T>1$ (in particular using the fact that 
$\phi(t,\cdot) \in H^1_0(J^+(o) \cap \Sigma_t)$ for every $t\geq 1$, $\Sigma_t$ being the Cauchy surface of $\bM$ at fixed $t$). In view of Theorem 3.2 in \cite{sogge}, 
in $[1,+\infty) \times \bR^3$,
 the extended $\phi$ must coincide to the other weak solution $\psi$ which satisfies the same requirements 
and initial data. In particular, $\phi=\psi$  in $D \cap ([1,2] \times \bR^3)$ so that
 (\ref{nFMM''}) holds true therein.
 To conclude we prove  (\ref{nFMM''}), i.e.
 that $\phi=\psi$, even in $D \cap ([0,1] \times \bR^3)$. If $\psi$ were smooth in the open set $D \cap ([0,1] \times \bR^3)$, 
in view of the standard 
 uniqueness theorem for the Cauchy problem in  $D \cap ([0,1] \times \bR^3)$ (e.g. Theorem 10.1.2  in \cite{Wald1}), 
 we would have $\phi=\psi$  in $D \cap ([0,1] \times \bR^3)$. To prove that
 $\psi$ is smooth therein we notice that if it were not the case for a point $p \in  D \cap ([0,1] \times \bR^3)$, 
there would be a full  complete null geodesic $\gamma$ 
 through $p$ with $\gamma$ included in the singular support of the distribution $\psi$, in view of the {\em H\"ormander's theorem of propagation of singularities},
  since $\psi$ solves the Klein-Gordon equation.
 On the other hand $\gamma$ would enter the region $D \cap ([1,2] \times \bR^3)$, where we already know that $\psi$ is smooth and thus its singular support is empty, 
 giving rise to a contradiction. Thus 
 $\psi$ is smooth in $D \cap ([0,1] \times \bR^3)$ as wanted and
 (\ref{nFMM''}) holds in $D \cap ([0,1] \times \bR^3)$, too.
 \end{proof}

\begin{proof}[Proof of Lemma \ref{lemmalemme}] In the hypotheses of the lemma, assuming in particular that
the case (b) holds (the other is similar) we have that
both $\widehat{\Phi'(\omega,k)}$ and  $k\widehat{\Phi'(\omega,k)}$ define functions in $L^2(\bS^2\times \bR,d\omega dk)$ varying
$(\omega,k)\in \bS^2\times \bR$.
Now defining $\theta(k) = 1$ for $k\geq 1$ and $\theta(k)=0$ otherwise, we have:
$$\int_{\bS^2 \times \bR^+} d\omega dk 2k  \overline{\widehat{\Phi(\omega,k)}} \widehat{\Phi'(\omega,k)}=
\lim_{\epsilon \to 0_+} \int_{\bS^2 \times \bR} d\omega dk 2  \overline{\widehat{\Phi(\omega,k)}} \: \left(e^{-\epsilon k}\theta(k) k\widehat{\Phi'(\omega,k)}\right)\:.$$
Using the fact that $\bR \ni k \mapsto e^{-\epsilon k} \theta(k)$
is  $L^2(\bR, du)$ and is the Fourier-Plancherel transform of $\bR \ni u' \mapsto (2\pi i)^{-1}/(u'-i\epsilon)$ and making use of  the convolution theorem for $L^2(\bR,du)$ functions
\cite{friedlander2} as well as Plancherel's theorem,  the integral in the right-hand side can be re-written as:
$$\int_{\bS^2 \times \bR} d\omega dk 2  \overline{\widehat{\Phi(\omega,k)}} \: \left(e^{-\epsilon k}\theta(k) k\widehat{\Phi'(\omega,k)}\right) =\int_{\bS^2\times \bR} d\omega du \overline{\Phi(\omega,u)} \int_{\bR} du' 
\frac{\partial_u \Phi'(\omega,u')}{u-u'-i\epsilon}\:.$$
Integrating by parts in the last integral, taking into account that $\Phi'(\omega,u) \to 0$ as $u \to 0_+$, and explicitly writing the supports of the functions,
we end up with:
\beq \int_{\bS^2 \times \bR} d\omega dk 2  \overline{\widehat{\Phi(\omega,k)}} e^{-\epsilon k}\theta(k) k\widehat{\Phi'(\omega,k)}=
-\frac{1}{\pi} 
\int_{\bS^2\times [0,1]} \sp\sp\sp d\omega du \int_{[0,1]} du'\frac{\overline{\Phi(\omega,u)} \Phi'(\omega,u')}{(u-u'-i\epsilon)^2}\:. \label{LLast}\eeq
Eventually, we notice that
$\int_{\bS^2\times [0,1]} d\omega  du \int_{[0,1]} du' 
\left|\frac{\overline{\Phi(\omega,u)} \Phi'(\omega,u')}{(u-u'-i\epsilon)^2}\right| < + \infty$
because the function $\bS^2 \times [0,1]\times [0,1] \ni (\omega,u,u') \mapsto |\Phi'(\omega,u')|/  (|u-u'|^2+\epsilon^2)$ is bounded by construction 
 ($\Phi'(\omega,u)$ vanishes uniformly in the angles as $u\to 0$) and  
 $\bS^2 \times [0,1] \ni (\omega,u) \mapsto \Phi(\omega,u)$ 
is $L^1(\bS^2\times [0,1],d\omega du)$ since it is $L^2(\bS^2\times [0,1], d\omega du)$. This implies that 
we can apply Fubini-Tonelli's theorem in the last iterated integral in (\ref{LLast})  computing it in the measure 
product, proving (\ref{lambda2}) and concluding the proof. 
\end{proof}
 
\begin{proof}[Proof of Lemma \ref{lemmabastardo}] In this proof, for every $\phi \in \sS_m(D)$,
 the same symbol is used for its unique smooth 
extension $L^m_{\bM D}\phi \in \sS_m( \bM)$ defined in the whole 
$\bM$. 
The right-hand side in \eqref{nFMM}  is nothing but the integral over $\Sigma$ of the $3$-form 
$\eta_{{\bk}, \phi}$ that appeared in \eqref{3forma}. 
As $\eta_{{\bk}, \phi} $ is constructed out of solutions of Klein-Gordon equation, it satisfies the conservation relation $d \eta_{{\bk}, \phi}  =0$ (which 
corresponds to the conservation of 
the associated current  $J_{{\bk}, \phi}  \doteq *\eta_{{\bk}, \phi} $).
Now we specialize $\Sigma$ to the surface at  $t=1$,  and  consider $\phi\in \sS_m(D)$ so that its Cauchy data are supported
 in the portion  of $\Sigma$ included in $D$.
Stokes-Poincar\'e theorem  applied to  $\eta_{{\bk}, \phi}$ 
leads immediately to
$\widehat{\phi}({\bk}) = \int_\Sigma \eta_{{\bk}, \phi}  = \int_V \eta_{{\bk}, \phi} \:, $
where both $3$-surfaces have future-oriented normal vector.  
Exploiting the fact that, on $V$, $t=u= ||{\bx}||$ so that one can write ${\bx} = u\omega$
for some vector $\omega \in \bS^2$, the second integral gives
(integrating by parts noticing that $ L^m_{VD}(\phi)$  vanishes for $u=0$ and  $u\geq 1$) 
\beq 
 \widehat{\phi}({\bk}) = \frac{-2i}{(2\pi)^{3/2} \sqrt{2E}} \int_{\bS^2\times [0,1]} \sp\sp\sp d\omega du \:
 ue^{-i({\bk} \cdot u \omega - E u)} \partial_u  L^m_{VD}(\phi) (\omega,u) 
\:. \label{mode2}
 \eeq
 Now we can insert the right-hand side of (\ref{mode2}), computed for $\phi_1$ and $\phi_2$, in the left-hand side of 
(\ref{id}) and try to interchange  the order of integrations. This cannot be done directly, but an $\epsilon$ prescription needs. 
However, a straightforward application 
of Fubini-Tonelli theorem and Lebesgue dominated convergence theorem lead to, if $T \doteq \bS^2\times \bS^2 \times [0,1]\times [0,1]$
 \begin{align}
 \int_{\bR^3} d{\bk} \:\: \overline{\widehat{\phi_1}({\bk})}\widehat{\phi_2}({\bk}) =
 \lim_{\epsilon \to 0_+} 
 \frac{2}{(2\pi)^3} \int_{T} d\omega d\omega'  & du du'  \left(\partial_u \Phi_1(\omega,u)\right)  \left(\partial_{u'} \Phi_2(\omega',u')\right) \nonumber\\
&\qquad \times\int_{\bR^3}\sp\sp d {\bk} \frac{uu'}{E} e^{-i({\bk} \cdot (u'\omega' - u \omega)  - E( u'-u +i\epsilon))}\:,  \label{inter}
\end{align}
where $\Phi_1 \doteq  L^m_{VD}(\phi_1)$ and $\Phi_2 \doteq  L^m_{VD}(\phi_2)$. We start considering the case $m>0$.
The last integral in (\ref{inter}) can be explicitly computed passing in polar coordinates in the variable ${\bk}$ and exploiting the result 3.961 in \cite{Grad}:
$$ \int_{\bR^3}d {\bk} E^{-1} e^{-i({\bk} \cdot (u'\omega' - u \omega)   - E( u'-u + i\epsilon))} =  I(u,u',\omega,\omega',m,\epsilon)
$$
where
$$I(u,u',\omega,\omega',m,\epsilon)\doteq\frac{4\pi m K_1\spa\left(   m \sqrt{    2 uu' (1- \cos\widehat{\omega \omega' })  +2i\epsilon (u-u' -i\epsilon/2)}      \right)     }{ 
\sqrt{    2 uu' (1- \cos\widehat{\omega \omega' })  +2i\epsilon (u-u' -i\epsilon/2)}  }
 \:.$$
 (From now on, considering the complex functions we shall encounter which admit the origin as branch point,  we always assume that the cut stays along the negative  real axis.)
For $\epsilon>0$ fixed,  this function is smooth and exponentially vanishes as the argument of $K_1$ tends to $+\infty$.
We have obtained that:
\begin{align}
&\int_{\bR^3} d{\bk} \:\: \overline{\widehat{\phi_1}({\bk})}\widehat{\phi_2}({\bk})\nonumber\\
&\qquad=\lim_{\epsilon\to 0_+} \int_{T} d\omega d\omega'  du du' u u' \left(\partial_u \Phi_1(\omega,u)  \partial_{u'} \Phi_2(\omega',u')\right)I(u,u',\omega,\omega',m,\epsilon) \:. \label{tot}
\end{align}

Our last step consists of working out  the limit above explicitly showing that it 
can be re-arranged as the right-hand side of (\ref{lambdaagg}).
First of all, we notice that the various integrations computed before taking the limit  can be interchanged in view of Fubini-Tonelli theorem, as the integrand is continuous with
 compact support and thus they are integrable in the product measure.  So we start performing the integral over the angles $(\omega,\omega') \in \bS^2 \times \bS^2$ decomposing it into two integrals 
 computed over two corresponding regions. 
 $A_\delta$ individuated by $1- \cos\widehat{\omega \omega' } > \delta$ and $B_\delta$ individuated by $0\leq 1- \cos\widehat{\omega \omega' } \leq  \delta$
 for a fixed $\delta >0$. Let us separately consider the two terms of the decomposition:
$$\lim_{\epsilon \to 0_+}\int_{T} = \lim_{\epsilon \to 0_+}\int_{A_\delta \times [0,1]^2} +\lim_{\epsilon \to 0_+} 
\int_{B_\delta \times [0,1]^2}$$ of the integral in the right-hand side of (\ref{tot}).
The kernel containing the parameter $\epsilon$ is jointly continuous   in all variables including $\epsilon \in [0, \epsilon_0]$ and thus it is $\epsilon$-uniformly
 bounded in $A_\delta \times [0,1]^2$. Using Lebesgue theorem again:
 \begin{align*}
  \lim_{\epsilon\to 0_+} \int_{A_\delta \times [0,1]^2 } \sp\sp\sp\sp\sp\sp\sp d\omega d\omega'  du du' u u' & \left(\partial_u \Phi_1(\omega,u) \partial_{u'} \Phi_2(\omega',u')  \right)
I(u,u',\omega,\omega',m,\epsilon)\\
 & = \int_{A_\delta \times \bR_+^2 }\sp\sp\sp\sp d\omega d\omega'  dv dv' v v' f_\lambda(v,\omega) 
g_\lambda(v',\omega')  
I(v,v',\omega,\omega',m,0)
\end{align*}
where we have replaced the range $[0,1]$ of $u,u'$ with $\bR_+$ since $\psi$ and $\psi'$ vanish for $u, u'\geq 1$ and, only in the last 
step, we have changed variables: $u \to v \doteq\lambda u$ and $u' \to v' \doteq u'/\lambda$, where $\lambda >0$ is a fixed real
and we have defined $f_\lambda(x,\omega) \doteq \partial_u \Phi_1(\omega,u)|_{u= x/\lambda}$
and 
 $g_\lambda(x,\omega) \doteq \partial_u \Phi_2(\omega,u)|_{u= \lambda x}$
Notice that the result cannot depend on $\lambda$, so we are allowed to take the limit as $\lambda \to +\infty$.
Notice that $\bS^2\times \bS^2\times \bR_+\times \bR_+  \ni(\omega,\omega',v,v') \mapsto f_\lambda(v,\omega) g_\lambda(v',\omega')$
is $\lambda$-uniformly bounded by construction because $\Phi =u\phi\spa $  has 
bounded $u$-derivative as it can be proved by direct inspection starting from the fact that the functions $\phi \in \sS_m(\bM)$
are smooth.  Furthermore, point-wisely,     $f_\lambda(v,\omega) g_\lambda(v',\omega')\to \partial_u\Phi_1(0,\omega) 0 =0$ constantly as $\lambda \to +\infty$ because $\Phi_2(\omega,u)$ smoothly vanishes if $u >1$. 
Finally the function 
$$A_\delta \times \bR_+\times \bR_+  \ni(\omega,\omega',v,v') \mapsto \left|vv' \frac{m K_1\left(   m \sqrt{    2 vv' (1- \cos\widehat{\omega \omega' }) }      \right)     }{ 
\pi^2 \sqrt{    2 vv' (1- \cos\widehat{\omega \omega' })  }  }\right|$$
is integrable. Summing up, Lebesgue theorem of dominated convergence eventually proves that, for every $\delta>0$:
 \begin{align*}
  \lim_{\epsilon\to 0_+} \int_{A_\delta \times [0,1]^2 } \sp\sp\sp\sp\sp\sp\sp d\omega d\omega'  & du du' u u' \ \partial_u \Phi_1(\omega,u)  \partial_{u'} \Phi_2(\omega',u') \ 
 I(u,u',\omega,\omega',m,\epsilon)\\
& =\lim_{\lambda \to +\infty}\int_{A_\delta \times \bR_+^2 }\sp\sp\sp\sp d\omega d\omega'  dv dv' v^2 v'^2 f_\lambda(v,\omega) 
g_\lambda(v',\omega')  
I(v,v',\omega,\omega',m,0) = 0\ .
\end{align*}
We conclude that,
for every $\delta>0$, $\int d{\bk}  \overline{\widehat{\phi_1}}\widehat{\phi_2}$ amounts to:
 $$ \lim_{\epsilon\to 0_+} \int_{B_\delta \times [0,1]^2 } \sp\sp\sp\sp\sp\sp\sp d\omega d\omega'  du du' u u' \left(\partial_u \Phi_1(\omega,u)  \partial_{u'} \Phi_2(\omega',u')\right)  
 I(u,u',\omega,\omega',m,\epsilon)\:.$$
To compute that limit, we notice that, while evaluating the integral above,
 the argument of $K_1$ varies, in fact,  in a bounded set,
since $\epsilon$ can be fixed in a interval $[0, \epsilon_0)$ and the supports of the $\partial_u\phi_i$ are compact,
 therefore we can exploit the expansion about $z=0$:
\beq
\frac{m K_1(m z)  }{z} = \frac{1}{z^2} - \frac{m^2}{4} \ln(z^2) + O(1)\:.
\eeq
The last two terms inserted in the integrand over  $B_\delta \times [0,1]^2$ give rise to functions which are integrable also for $\epsilon=0$. In that case, using Lebesgue theorem again, one sees that the limit as $\epsilon \to 0_+$ can be computed by direct 
substitution of $\epsilon$ with $0$ in the integrand. Next, as $\delta>0$ can be taken arbitrarily small and using the fact that
the measure of $B_\delta \times [0,1]^2$ vanishes as $\delta \to 0$, one easily concludes that the terms $- \frac{m^2}{4} \ln(z^2) + O(1)$
give no contribution to $\int d{\bk}  \overline{\widehat{\phi_1}}\widehat{\phi_2} $. We end up with, 
where we make explicit the factor $u$ in the definition of $\Phi(\omega,u) = u \phi(\omega,u)$:
\beq 
\int_{\bR^3} d{\bk}\overline{\widehat{\phi_1}({\bk})}\widehat{\phi_2}({\bk})= 
\lim_{\epsilon\to 0_+} 
\int_{B_\delta \times [0,1]^2 } \sp\sp\sp\sp\sp\sp\sp d\omega d\omega'  du du' 
\frac{u u' \partial_u u\phi_1(\omega,u)
  \partial_{u'} u'\phi_2(\omega',u')}{2\pi^2 \left(uu' (1- \cos\widehat{\omega \omega' })  +i\epsilon (u-u' -i\epsilon/2)\right)}.\label{Last}
\eeq	
Integrating by parts in the variable $u$ the right-hand side becomes:
\beq
\int_{B_\delta \times [0,1]^2 } \sp\sp\sp\sp\sp\sp\sp d\omega d\omega'  du du' 
 \frac{\phi_1(\omega,u)  \left(\partial_{u'} u'\phi_2(\omega',u')\right) (i\epsilon u u'^2   + \epsilon^2 uu' /2)}{2\pi^2 \left(uu' (1- \cos\widehat{\omega \omega' })  +i\epsilon (u-u' -i\epsilon/2)\right)^2}\:.\label{last}
 \eeq
The part of integral   proportional to $\epsilon^2$ vanishes as $\epsilon\to 0_+$ as we go to prove. Below,  
 $\omega$ is individuated by its polar angles, $\varphi \in (-\pi,\pi)$, $\theta\in (0,\pi)$. Furthermore,  we have fixed the polar axis $z$ to 
be $\omega'$ performing the integration in $d\omega$, so that, in fact,  $\theta$ and $\varphi$ 
depends on $\omega'$ parametrically, too. We have, for some constant $C\geq 0$  depending on $\phi_1$ and $\phi_2$
\begin{align}
&\left|\int_{B_\delta \times [0,1]^2 } \sp\sp\sp\sp\sp\sp\sp d\omega d\omega'  du du' 
\phi_1(\omega,u)  \left(\partial_{u'} u'\phi_2(\omega',u')\right)
 \frac{\epsilon^2 u u'}{2\pi^2 \left(uu' (1- \cos\widehat{\omega \omega' })  +i\epsilon (u-u' -i\epsilon/2)\right)^2}\right|\nonumber \\
 &\leq  2\pi  C
\int_{[0,1]^2 }\sp \sp\sp  du du'
 \frac{\epsilon}{|u-u'|} \arctan\left(\frac{|u-u'|}{\epsilon}\right) \to 0 \:.
\end{align}
In the last step we used the fact that $\bR_+ \ni \mapsto x^{-1} \arctan x$ is bounded (so, the integrand above is bounded  when $(u,u')\in [0,1]^2$), and 
$ \frac{\epsilon}{|u-u'|} \arctan\left(\frac{|u-u'|}{\epsilon}\right) \to 0$ pointwisely as $\epsilon \to 0_+$ so that we can exploit Lebesgue dominated 
convergence theorem.

\noindent Starting form \eqref{Last}
%
and performing an integration by parts we get:
\begin{align}
&\int_{\bR^3} d{\bk} \:\: \overline{\widehat{\phi_1}({\bk})}\widehat{\phi_2}({\bk})\nonumber\\
 &= \frac{1}{2\pi^2}\lim_{\epsilon\to 0_+}\left\{ 
\int_{\bS^2 \times [0,1]^2 } \sp\sp\sp\sp\sp\sp\sp d\omega'  du du' \partial_{u'} u'\phi_2(\omega',u') 2\pi
  \left(\frac{i \epsilon u' \phi_1(u, \omega' )   }{i\epsilon (u-u' -i\epsilon/2)}\right) \nonumber\right. \\
&\qquad\qquad\quad
-\int_{\bS^2 \times [0,1]^2 } \sp\sp\sp\sp\sp\sp\sp d\omega'  du du' \partial_{u'} u'\phi_2(\omega',u')
 \int_{0}^{2\pi} \sp \sp d\varphi  \left(\frac{i \epsilon u' \phi_1(u, 1-\delta, \varphi )   }{uu' \delta  +i\epsilon (u-u' -i\epsilon/2)}\right)\label{LLLast} \\
&\qquad\qquad \left.
-\int_{\bS^2 \times [0,1]^2 } \sp\sp\sp\sp\sp\sp\sp d\omega'  \spa du du' \partial_{u'} u'\phi_2(\omega',u')\sp \int_{0}^{2\pi} \sp\sp \spa d\varphi 
 \int_{1-\delta}^1 \sp \sp\sp d \cos \theta
 \left(\frac{i \epsilon u' \partial_{\cos\theta} \phi_1(u,\theta, \varphi )   }{uu' (1- \cos \theta)  +i\epsilon (u-u' -i\epsilon/2)}\right) \right\} \nonumber.
\end{align}	
 As we prove shortly, the last two integrals vanish so that we get, integrating by parts in the remaining integral:
\beq \int_{\bR^3} d{\bk} \:\: \overline{\widehat{\phi_1}({\bk})}\widehat{\phi_2}({\bk}) = \lim_{\epsilon \to 0_+}-\frac{1}{\pi}  
\int_{\bS^2\times [0,1]\times [0,1]} \sp\sp\sp d\omega  du du' \frac{u\phi_1(\omega,u) u'\phi_2'(\omega,u')}{(u-u'-i\epsilon)^2}\:.
\eeq
This is just what we need to conclude the proof of the theorem. 
 To complete the proof we have to prove that the last two integrals in (\ref{LLLast})  vanish.
Defining, if $\cos \theta_\delta = 1-\delta$,
$$\Psi(u) \doteq \int_{0}^{2\pi} d\varphi  \phi_1(u, \theta_\delta, \varphi ) \:,$$
the next to  last integral 
in (\ref{LLLast}) can be re-arranged as:
$$i\epsilon \int_{\bS^2 \times [0,1]^2 } \sp\sp\sp\sp\sp\sp\sp d\omega'  du du'
\left(\partial_{u'} u'\phi_2(\omega',u')\right) 
\Psi(u)
 \frac{u'}{u'\delta +i\epsilon } \partial_u\ln\left( uu' \delta + i\epsilon (u-u' -i\epsilon/2) \right)\:.$$
In turn this integral decomposes into three parts if integrating by parts:
\begin{align}
&-i\epsilon \int_{\bS^2 \times [0,1]^2 } \sp\sp\sp\sp\sp\sp\sp d\omega'  du du'
\left(\partial_{u'} u'\phi_2(\omega',u')\right) 
 \partial_u\Psi(u)
 \frac{u'}{u'\delta +i\epsilon } \ln\left( uu' \delta + i\epsilon (u-u' -i\epsilon/2) \right) \nonumber \\
 &+ i\epsilon  \int_{\bS^2 \times [0,1]^2 } \sp\sp\sp\sp\sp\sp\sp d\omega'  du du'
\left(\partial_{u'} u'\phi_2(\omega',u')\right) 
 \Psi(0)
 \frac{u'}{u'\delta +i\epsilon } \ln\left(u' +i\epsilon/2) \right)\nonumber  \\
 &+ i\epsilon \ln(-i\epsilon) \int_{\bS^2 \times [0,1]^2 } \sp\sp\sp\sp\sp\sp\sp d\omega'  du du'
\left(\partial_{u'} u'\phi_2(\omega',u')\right) 
 \Psi(0)
 \frac{u'}{u'\delta +i\epsilon }\:.\label{LLLL}
\end{align}
Then notice that the logarithm as well as the fraction $u'/(u'\delta -i\epsilon)$ are 
integrable  for every values of $\epsilon\geq 0$, including $\epsilon =0$. 
Using the fact that $|u' /(u'\delta -i\epsilon)|$ is bounded uniformly in $\epsilon$ and that 
$u'\phi_2(\omega',u') \partial_u\Psi(u)$ is bounded on the domain of integration, 
each integral above  can be proved to vanish as $\epsilon \to 0_+$ due to the presence of the factor $i\epsilon$.
Let us discuss in details the most difficult integral, that is the first one.
Assuming $\delta, \epsilon >0$ sufficiently small,  taking the absolute value of the integrand,  the first integral in (\ref{LLLL}) can be bounded by, where 
$K\geq 0$ is some constant depending on $\phi_1$ and $\phi_2$,
\[
\epsilon K  \int_{[0,1]^2 } \sp\sp\sp\sp  du du' 
 \sqrt{ \left[\ln(uu' \delta)\right]^2 +\pi^2/4} \to 0 \quad \mbox{as $\epsilon \to 0_+$.}
 \]
 The remaining two integrals  in (\ref{LLLL}) can be treated similarly proving much more straightforwardly  that they vanish as well.
Finally, turning our attention on (\ref{LLLast})  again,
defining: 
$$\frac{u}{\sqrt{1-\cos\theta}} \Psi(u,  \theta) \doteq \int_{0}^{2\pi} d\varphi \partial_{\cos\theta} \phi_1(u, \theta, \varphi ) \:,$$
where $\Psi$ is bounded.
The factor $u/\sqrt{1-\cos\theta}$ arises from the fact that, in view of the restriction to $V$, $\phi\restriction_V = 
\phi(u, u\cos \theta, u \sin\theta \cos \varphi,  u \sin\theta \sin \varphi)$ and $d\sin \theta /d\cos \theta = \cos \theta /\sqrt{(1-\cos \theta)(1+\cos \theta)}$.
Thus the last integral in (\ref{LLLast}) can be re-arranged as:
$$i\epsilon \int_{\bS^2 \times [0,1]^2 } \sp\sp\sp\sp\sp\sp\sp d\omega'  du du'
\left(\partial_{u'} u'\phi_2(\omega',u')\right) 
  \int_{0}^\delta \sp  \frac{d \mu}{\sqrt{\mu}}
 u\Psi(u,\theta)
 \frac{u'}{u'\mu +i\epsilon } \partial_u\ln\left( uu' \mu + i\epsilon (u-u' -i\epsilon/2) \right)\:.$$
By integrating by parts it can be re-arranged to:
\beq 
&-i2\epsilon\spa  \int_{\bS^2 \times [0,1]^2 } \spa d\omega'  du du' 
\left(\partial_{u'} u'\phi_2(\omega',u')\right) 
   \int_{0}^{\sqrt{\delta}} \sp d \sqrt{\mu}
 \frac{ \partial_u u\Psi(u,\mu)  u'}{u' (\sqrt{\mu})^2 +i\epsilon } \ln\left(u' \mu + i\epsilon (u-u' -i\epsilon/2) \right)\:. \nonumber 
\eeq
As before, assuming $\epsilon$ and $\delta$ small enough,  this integral can be bounded by:
\[
\epsilon K'  \int_{[0,1]^2\times [0,\sqrt{\delta}] } \sp\sp\sp\sp  \sp\sp\sp\sp  du du' dx
 \sqrt{ \left[\ln(uu'x)\right]^2 +\pi^2/4} \to 0 \quad \mbox{as $\epsilon \to 0_+$.}
 \]
 where $K'\geq 0$ is some constant  depending on $\phi_1$ and $\phi_2$. \\
 The case $m=0$ is very similar but much more simple. In fact, coming back to (\ref{inter}),  making use of 3.895 in \cite{Grad},
 (\ref{Last}) arises directly, the analog integral on $A_\delta$ being vanishing similarly. Afterwards the proof 
 is the same as for $m>0$. 
 \end{proof} 
 
\begin{proof}[Proof of Theorem \ref{T3}]
Replace the coordinate $u\in (0,1)$ on $V$ with the coordinate 
$\ell \in (-\infty,+\infty)$ such that
$u= 1/(1+ e^{-\ell})$ and $X = \partial_\ell$  induces trivial displacements $\ell \to \ell+\tau$. 
 The functions $\Phi \in \sS(V)$
turns out to be smooth jointly in the variables $(\omega,\ell)$, they vanishes for $\ell$ sufficiently 
large and vanish with all of their  
derivatives, uniformly in  $\omega$, for $\ell \to -\infty$ with order $O(e^{\ell})$. In particular 
$\Phi \in \sS(V)$ belongs to 
the Schwartz  space $\mS(\bS^2\times \bR)$ of the complex smooth functions which vanishes as $|\ell| \to +\infty$, 
uniformly in $\omega$ and with all of their derivatives, faster than every inverse power of $|\ell|$. 
Therefore the $\ell$-Fourier transform of $\Phi \in \sS(V)$
is well defined and belongs to $\mS(\bS^2\times \bR)$ again.  
Then the function
$\sK'_\lambda : \sS(V) \to L^2(\bS^2 \times \bR;  m(h)  d\omega
 dh)$ is well defined (notice that $|m(h) | \leq C(1+|h|^2)$ for some constant $C\geq 0$ 
 and all $\ell \in \bR$).
 By direct inspection, using standard theorems of Fourier transform theory, and the fact that
 $\widetilde{\Phi(\omega,-h)} = \overline{\widetilde{\Phi(\omega,h)}}$
if $\Phi \in \sS(V)$ since it is real, 
   one has that
 \beq
 \sigma(\Phi,\Phi') = -2\Imm \langle \sK'_\lambda \Phi,  \sK'_\lambda  \Phi'\rangle_{L^2(\bS^2 \times \bR; 
 m(h)  d\omega dh)} \quad \mbox{for all $\Phi,\Phi' \in \sS(V)$.} \label{simp3}
 \eeq
Let us now prove that  $\overline{  \sK'_\lambda  \sS(V) + i \sK'_\lambda \sS(V)} = L^2(\bS^2 \times \bR; m(h)  d\omega dh)$. 
This is an immediate consequence of the following lemma, whose proof can be found in this appendix. 

\begin{lemma} \label{lemmaA}
 Referring to coordinates $(\omega,\ell)\in \bS^2\times \bR$ to define $C_0^\infty(\bS^2\times \bR)$, the space 
 $\sK'_\lambda  C_0^\infty(\bS^2\times \bR) + i \sK'_\lambda C_0^\infty(\bS^2\times \bR)$
is dense in $L^2(\bS^2 \times \bR; d\mu_{\bS^2} m(h)  dh)$. 
\end{lemma}

\noindent If we define $\mu_{\lambda'}(\Phi,\Phi') \doteq \Rea\langle \sK'_\lambda \Phi,\sK'_\lambda\Phi' \rangle_{\sH'_\lambda}$ for $\Phi,\Phi'\in \sS(V)$, 
it turns out that $\mu_{\lambda'}$ is a real scalar product (strictly positive) on $\sS(V)$. Furthermore it satisfies
$|\sigma_V(\Phi,\Phi')|^2 \leq 4\mu_{\lambda'}(\Phi, \Phi)\mu_{\lambda'} (\Phi', \Phi')$ for all $\Phi,\Phi' \in S(V)$ in view of (\ref{simp3})
and the Cauchy-Schwarz inequality for $\langle\:,\: \rangle_{\sH'}$. Therefore, due to Proposition 3.1 in \cite{KW},  there is a quasifree state on $\cW(V)$ 
(the unique satisfying the  identity in (\ref{lambda})  with $\mu_\lambda$ replaced for $\mu_{\lambda'}$) with one-particle space structure 
$(\sK'_\lambda, \sH'_\lambda)$. To conclude the proof of (a) we have to prove that $(\sK'_\lambda, \sH'_\lambda)$ is unitarily equivalent to $(\sK_\lambda, \sH_\lambda)$
(so that $\lambda'$ coincides to $\lambda$).
Changing variables in the two-point function of state $\lambda$ and obtaining:
\beq
\lambda(\Phi,\Phi') = \lim_{\epsilon \to 0_+} -\frac{1}{4\pi} \int_{\bS^2 \times \bR^2} d\omega d\ell d\ell' 
\frac{\Phi(\omega, \ell) \Phi'(\omega, \ell')}{\left(\sinh(\frac{\ell- \ell'}{2}) -i\epsilon \cosh(\ell/2) \cosh(\ell'/2)  \right)^2}\:, \label{GG'}
\eeq
we can exploit  the following technical result proved below in this appendix.

\begin{lemma} \label{lemmaaggiunto}  Let $N$ be a smooth Riemannian manifold and let $\mu$ be the natural measure induced by the metric on $N$. 
If $f,h \in C_0^\infty(N)$ with $h>0$, $df \neq 0$, $fdh -hdf \neq 0$  on $N$, then:
\beq
 \lim_{\epsilon\to 0^+} \int_N \frac{g(q)}{(f(q) \pm i\epsilon h(q))^2}d \mu = \lim_{\epsilon\to 0^+} \int_N \frac{g(q)}{(f(q) \pm i\epsilon)^2} d\mu
\quad \mbox{for every $g\in C_0^\infty(N)$,} \label{GGQQqq}
\eeq
and both limits do exist and are finite. 
\end{lemma}

\noindent Assuming $N= \bS^2 \times \bR^2$ equipped with the natural product metric (referred to coordinates $(\ell,\ell')$ on $\bR^2$)
 and taking $f(\omega,\ell,\ell') = \sinh((\ell-\ell')/2)$, 
so that $df \neq 0$ everywhere on  $N$, and 
$h(\omega,\ell,\ell') = \cosh(\ell/2)\cosh(\ell'/2)$ which is strictly positive and $f dh -h d f \neq 0$ as it can be checked by direct inspection,
Lemma \ref{lemmaaggiunto} and (\ref{GG'}) entail:
\beq
\lambda(\Phi,\Phi') = \lim_{\epsilon \to 0_+} -\frac{1}{4\pi} \int_{\bS^2 \times \bR^2} d\omega d\ell d\ell' 
\frac{\Phi(\omega, \ell) \Phi'(\omega, \ell')}{\left(\sinh(\frac{\ell- \ell'}{2}) -i\epsilon  \right)^2} \label{GGpepe} \quad\mbox{if $\Phi,\Phi' \in C_0^\infty(\bS^2\times \bR)$.}
\eeq
The $\ell$-Fourier transform of the distribution $ 1/(\sinh(\ell/2) - i0_+)^2$ is $-2\sqrt{2\pi} m(h) $ with $\delta$
defined in (\ref{mu}), so that the convolution theorem applied  to the right hand side of  (\ref{GGpepe}) proves that,
if $\Phi,\Phi' \in C_0^\infty(\bS^2\times \bR)$:
\beq
\langle \sK_\lambda \Phi,\sK_\lambda\Phi' \rangle_{\sH_\lambda} = \int_{\bS^2 \times \bR^2}\sp\sp\sp\sp m(h)  d\omega  dh \:\:  \overline{\widetilde{\Phi}}(\omega,h) 
 \widetilde{\Phi'}(\omega,h) = \langle \sK'_{\lambda} \Phi,\sK'_{\lambda}\Phi' \rangle_{\sH_{\lambda'}}\:. \label{GGG}
\eeq
We have found that the $\bR$-linear, bijective  map $U_0 :  \sK'_{\lambda}(C_0^\infty(\bS^2\times \bR)) \to \sK_{\lambda}(C_0^\infty(\bS^2\times \bR))$
which associates\footnote{This map is well defined since $\sK_\omega \Phi = \sK_\omega \Phi_1$ implies $\Phi = \Phi_1$ as
 the $\bR$-linear map $\sK_\omega : \sS \to \sH_\omega$ is injective for every one-particle space structure $(\sK_\omega, \sH_\omega)$
  of every quasifree state $\omega$ of a Weyl algebra of a symplectic space $(\sS,\sigma)$, because $\Imm\langle \sK_\omega \Phi, \sK_\omega \Phi_1\rangle$
  is proportional to $\sigma(\Phi,\Phi_1)$ that is nondegenerate by hypotheses.} $\sK'_\lambda \Phi$ to $\sK_\lambda \Phi$
is isometric. On the other hand the following lemma (proved below in the appendix) holds.

\begin{lemma} \label{lemmaA2}
 Referring to coordinates $(\omega,\ell)\in \bS^2\times \bR$ to define $C_0^\infty(\bS^2\times \bR)$, the space 
 $\sK_{\lambda}(C_0^\infty(\bS^2\times \bR)) + i  \sK_{\lambda}(C_0^\infty(\bS^2\times \bR))$
is dense in $\sH_\lambda$. 
\end{lemma}

\noindent Since also $\sK'_\lambda  C_0^\infty(\bS^2\times \bR) + i \sK'_\lambda C_0^\infty(\bS^2\times \bR)$
is dense in $\sH'_{\lambda}$ (by Lemma \ref{lemmaA}),  Lemma A.1 in \cite{KW} implies that $U_0$ linearly and continuously 
extends to a unique ($\bC$-linear) Hilbert space isomorphism 
$U : \sH'_{\lambda} \to \sH_\lambda$. By construction $U \sK'_{\lambda} = \sK_\lambda$, so that the two
one-particle space structures $(\sK'_\lambda, \sH'_\lambda)$ and $(\sK_\lambda, \sH_\lambda)$ are unitarily equivalent.

Let us prove $(a)$ and $(b)$ in the thesis of Theorem \ref{T3}. Since $X= \partial_\ell$, it holds $(\beta^{X}_\tau \Phi)(\omega,\ell) =  \Phi(\omega,\ell - \tau)$.
Thus,  referring the the one-particle space structure $(\sK'_\lambda, \sH'_\lambda)$,
the action of $\beta^{X}_\tau$ is equivalent to the appearance of a phase $e^{ih\tau}$ in front of the $\sK'_\lambda \Phi$ by 
standard properties of the Fourier 
transform.  Defining  $(V^{(\lambda')}_\tau \Psi)(\omega,h) \doteq e^{ih\tau} \Psi(\omega,h) $, 
it holds $V^{(\lambda')}_\tau K'_\lambda \Phi = K'_\lambda \beta^{X}_\tau \Phi$.
 In view of Stone theorem and a few of elementary observations, $\hat{h}$ is the generator of $V^{(\lambda')}$. 
As $\mu_\lambda(\beta_\tau^{X}\Psi, \beta_\tau^{X}\Psi) =  \langle V^{(\lambda')}_t K'_\lambda\Phi, V^{(\lambda')}_t 
K'_\lambda \Phi\rangle_{\sH'_\lambda} =
 \langle  K'_\lambda\Phi, K'_\lambda \Phi\rangle_{\sH'_\lambda} = \mu_\lambda(\Phi,\Phi)$ because phases cancel, $\lambda$
 turns out to be $\alpha^{X}$-invariant due to (\ref{lambda}),  (\ref{lambdaagg}) and (\ref{alphatau}). Passing to the unitarily equivalent one-particle space structure
$(\sK_\lambda,\sH_\lambda)$, the result remains unchanged if replacing $V^{(\lambda')}$ with the corresponding unitary
$V^{(\lambda)}$.
 Passing to the second-quantization
 $U^{(\lambda)}_\tau$ of $V^{(\lambda)}_\tau$ (assuming $U^{(\lambda)}_\tau\Psi_\lambda = \Psi_\lambda$)
 implements $\alpha^X$ as follows from (\ref{Vs}) by standard procedures. 
To conclude the proof of  (b), for $\beta = 2\pi$, define the antilinear map $j : \sK'_\lambda \sS(V) \to \sK'_\lambda \sS(V)$ such that
$(j \widetilde{\Phi})(\omega, h) \doteq -e^{-\beta h/2}\overline{\widetilde{\Phi}(\omega, -h)}$.
The definition of $j$ is well-posed since $\overline{\widetilde{\Phi}(\omega, -h)} = 
\widetilde{\Phi}(\omega,h)$ if $\Phi \in \sS(V)$ because $\Phi$ is real-valued and 
$\sK'_\lambda (\sS) \subset \Dom\left( e^{-\frac{1}{2}\beta \hat{h}}\right)$  by direct inspection.
Notice that $jj=I$ and $j$ commutes with $V^{(\lambda)}_t$. We can extend $j$ to  the whole $\sH'_\lambda$ by antilinearity and continuity eventually obtaining 
an antilinear operator
$j: \sH'_\lambda \to \sH'_\lambda$ with $j j = I$ and the following facts are verified:
(i) $\sK'_\lambda (\sS) \subset \Dom\left( e^{-\frac{1}{2}\beta \hat{h}}\right)$,
 (ii) $[j, V^{(\lambda)}_t] =0$ if $t\in \bR$, 
(iii) $e^{-\frac{1}{2}\beta \hat{h}} \sK'_\lambda  \psi = -j \sK'_\lambda \psi$ if $\psi \in \sK'_\lambda\sS(V)$.
By the standard procedure of second-quantization with Fock space $\cH_\lambda$, $j$ 
uniquely determines an  antiunitary  operator $J$ leaving fixed $\Psi_\lambda$. By standard
procedures using (i),(ii)  and (iii) and referring to $J$ and $U^{(\lambda)}_t$, 
one easily proves that the requirements (1), (2) and (3) in the (second equivalent) definition of KMS state presented at the end of this appendix are satisfied
(see Section  3.2 in \cite{KW} and references therein for details) so that $\lambda$ turns out to be KMS at the inverse temperature $\beta=2\pi$ with respect to 
$\alpha^{X}$. The remaining part  of the statement (a) arises by  the link between KMS condition and modular group 
outlined at the beginning of Section \ref{sec5}. (b) arises immediately defining $V_\tau^{(\lambda)}$
as the unitary corresponding to $V_\tau^{(\lambda')}$ in the structure $(\sK_\lambda, \sH_\lambda)$.
\end{proof}

\begin{proof}[Proof of the Lemma \ref{lemmaA}]
 $\sS(V)+i \sS(V) \supset C_0^\infty(\bS^2\times \bR) + i C_0^\infty(\bS^2\times \bR)
 \doteq C_0^\infty(\bS^2\times \bR;\bC)$. Since the latter is dense in $\mS(\bS^2\times \bR)$ and 
 the Fourier transform is continuous with respect to the topology of $\mS(\bS^2\times \bR)$ and leaves that space fixed, we conclude that 
$\overline{\sK'_\lambda \sS(V) + i\sK'_\lambda \sS(V)}$ (where the closure is that of the topology 
of $L^2(\bS^2 \times \bR;  m(h)  d\omega dh)$) 
includes the subspace $\mS_0 \doteq 
\sK'_\lambda  C_0^\infty(\bS^2\times \bR) + i \sK'_\lambda C_0^\infty(\bS^2\times \bR)$,  which is dense in $\mS(\bS^2\times \bR)$ 
with respect the  $\mS(\bS^2\times \bR)$  topology (where, now
 $\bS^2\times \bR$ is referred to the coordinates $(\omega, h)$). Since the convergence in the $\mS(\bS^2\times \bR)$ 
 topology implies that in the topology of each
$L^2(\bS^2 \times \bR;  (1+ |h|^n) d\omega dh)$ for $n=0,1,2,\ldots$, and  $|m(h) | \leq C(1 + h^2)$ for some $C\geq 0$
and all $h\in \bR$, we conclude that $\mS_0$ is dense in $\mS(\bS^2\times \bR)$
also referring to the topology of $L^2(\bS^2 \times \bR; m(h)  d\omega dh)$. In particular  it has to hold
$\mS(\bS^2\times \bR)\subset \overline{\sK'_\lambda \sS(V) + i \sK'_\lambda \sS(V)}$ and thus
$C_0^\infty(\bS^2\times \bR; \bC) \subset \overline{\sK'_\lambda \sS(V) + i\sK'_\lambda\sS(V)}$ where, again,
 $\bS^2\times \bR$ is referred to the coordinates $(\omega, h)$.
 To conclude, it is sufficient to establish 
that $C_0^\infty(\bS^2\times \bR; \bC)$ is dense in $L^2(\bS^2 \times \bR;  m(h)  d\omega dh)$. 
This is a trivial consequence of the fact that $C_0^\infty(\bS^2\times (0,+\infty); \bC)$ is dense in $L^2(\bS^2 \times (0,+\infty); d\omega dx)$,
 passing from the variable 
$x$ to the variable $h$ such that $x(h) = \int_0^h m(h') dh'$ and noticing that $g\in C^\infty_0(\bS^2\times (0,+\infty); \bC)$ if and only if 
$\bR \ni h \mapsto g(x(h))$ belongs to $C^\infty_0(\bS^2\times \bR; \bC)$. 
\end{proof}

\begin{proof}[Proof of Lemma \ref{lemmaaggiunto}] With our hypothesis on $f$,
all the nonempty sets of the form $\Sigma_c := \{x \in N \:|\: f(x) = c\}$ 
  are embedded submanifolds of $N$. Thus, in a neighborhood $\Omega_p$ of every
 $p \in N$, there is a coordinate patch $\psi: \Omega_p \ni q \mapsto  (y^1(q),y^2(q),\ldots, y^n(q))$ satisfying  $f\circ \psi^{-1} = y^1$
 so that
 $\Sigma_c\cap \Omega_p$ is made of  the points $q\in N$ with $y^1(q) =c$. We henceforth assume  that $\Sigma_0 \neq \emptyset$
  and we restrict to work in $\Omega_p$.
Distributions $\delta(f)$ and $\cP(1/f)$  in $\cD'(\Omega_p)$ can be defined as follows, for every $g \in C_0^\infty(\Omega_p)$:
\begin{align}
 & \langle \delta(f), g \rangle \doteq  \int_{\bR^{n-1}} g(0,y^2,\ldots,y^n)  \frac{d\mu_{\Sigma_c}}{|df|}\nonumber \:,\\
 &\left\langle \cP\left(\frac{1}{f}\right), g \right \rangle \doteq \lim_{\epsilon\to 0^+} \left(\int_{-\infty}^{-\epsilon}
   \frac{dy^1}{y^1} + \int_{\epsilon}^{+\infty}  \frac{dy^1}{y^1}\right)
 \int_{\bR^{n-1}} g(y^1,\ldots,y^n) \frac{d\mu_{\Sigma_c}}{|df|} \:,\nonumber
\end{align} 
where $\mu_{\Sigma_0}$ is the measure induced on $\Sigma_0$ by the metric of $N$, and $|df|$ the norm of $df$ referred 
to the metric.
The reader can verify that 
the definitions are well-posed, reduce to the standard ones for $N=\bR$ (equipped with the natural Euclidean metric)
 and are independent  from the used coordinate patch  on $\Omega_p$ verifying the hypotheses.
Finally, working in the said coordinate patch and through a very straightforward generalization of the proof of Sochockij formulas as in \cite{Vladimirov}, 
one extends these formulas 
to our case. More precisely, it turns out that: {\em Assuming $df\neq 0$ on $N$,  in a sufficiently small open neighborhood $\Omega_p$ of any point $p\in N$}:
\beq
\lim_{\epsilon \to 0^+} \int_{N} \frac{g(q)}{f(q) \pm i \epsilon h(q)} \: d\mu  = \left\langle \cP\left(\frac{1}{f}\right),g \right \rangle \mp 
i \pi \langle \delta(f), g\rangle \quad \mbox{\em for every $g \in C^\infty_0(\Omega_p)$}  \:,
\eeq
{\em where the assigned function  $h: N \to \bR$ is  every  smooth function that is {\em strictly positive} on $\Omega_p$.}
In particular 
one can choose $h \equiv 1$. The obtained result entails that, in weak sense: 
$$\frac{1}{f(q) \pm i \epsilon h(q)} \to  \cP(1/f) \mp i\pi \delta(f)
\leftarrow \frac{1}{f(q) \pm i \epsilon}  \quad \mbox{as $\epsilon \to 0^+$.}$$
As  weakly convergent sequences of distributions
can be derived term-by-term obtaining again a sequence that weakly converges to a distribution, using again coordinates $y^1,\ldots,y^n$, taking the $y^1$ 
derivative
--  taking into account some further  smooth factors arising from the density function 
in the integration measure --
 and coming back to general coordinates on $\Omega_p$, we conclude that, if $f$ and $h$ satisfy our hypotheses
 and  $g \in C^\infty_0(\Omega_p)$:
\beq
\lim_{\epsilon \to 0^+} \left(\int_{N} \frac{g(q)}{(f(q) \pm i \epsilon h(q))^2} \: d\mu  
 \pm i\epsilon  \int_{N} \frac{g(q) \partial_{y^1} h(q)}{(f(q) \pm i \epsilon h(q))^2} \: d\mu \right) \nonumber \\ =
\lim_{\epsilon \to 0^+} \int_{N} \frac{g(q)}{(f(q) \pm i \epsilon )^2} \: d\mu \in \bC \:. \label{GGQQ}
\eeq
Notice that the result above not only states that the two sides coincide, but it even states that both limit  exist and are finite.
The second integral in the left-hand side can be re-arranged as: 
$$\int_{N} \frac{g(q) (\partial_{y^1} h(q))/h(q)^{1/2}}{((f(q)/h(q)) \pm i \epsilon )^2} \: d\mu$$
The existence of the limit in the right-hand side in
 (\ref{GGQQ}) for $f(q)$ replaced by $f(q)/h(q)$,
and $g(q)$  replaced by $g(q) (\partial_{y^1} h(q))/h(q)^{1/2}$
 implies that 
$\lim_{\epsilon\to 0^+} \int_{N} \frac{g(q) \partial_{y^1} h(q)}{(f(q) \pm i \epsilon h(q))^2} \: d\mu$
 exists in $\bC$, provided that $d(f/g) \neq 0$, i.e.
$hdf -fdh \neq 0$, shrinking the neighborhood $\Omega_p$ if necessary. Consequently, under these hypotheses,
 the second term in the left-hand side 
of (\ref{GGQQ}) vanishes as $\epsilon \to 0^+$ in view of the overall factor $\epsilon$. 
We conclude that: {\em  if $df \neq 0$, $h>0$,  $hdf -fdh \neq 0$ on $N$ then, in a sufficiently small 
neighborhood $\Omega_p$ of every $p\in N$, one has}:
\beq
\lim_{\epsilon \to 0^+}\!\int_{N}\! \frac{g(q)}{(f(q) \pm i \epsilon h(q))^2}\ d\mu =
\lim_{\epsilon \to 0^+} \!\int_{N} \!\frac{g(q)}{(f(q) \pm i \epsilon )^2}\ d\mu \in \bC,\quad \mbox{$g\in C_0^\infty(\Omega_p)$.} \label{GGQQq}
\eeq
Now, with the given hypotheses on $f$ and $h$ on $N$, we extend the result to the whole manifold $N$.   Take $g\in C_0^\infty(N)$. For every $p\in \supp(g)$ there is a neighborhood 
$\Omega_p$ where (\ref{GGQQq}) can be applied. As $\supp(g)$ is compact, we can extract a finite open covering $\{\Omega_{p_i}\}_{i=1,\ldots, N}$  
of $\supp(g)$ and construct a  partition of the unit $\{ \psi_i\}_{i=1,\ldots,N}$ subordinate to that finite covering.
Notice that in each $\Omega_{p_i}$ the identity (\ref{GGQQq}) holds because we can construct the relevant coordinate systems therein.
Then, since the sum over $i$  is finite so that it can be interchanged with the integral symbol, and all the limits are finite, we achieve the thesis:
\begin{align}
& \lim_{\epsilon \to 0^+}\int_{N} \frac{g(q)}{(f(q) \pm i \epsilon h(q))^2} \: d\mu 
=\lim_{\epsilon \to 0^+}  \sum_{i=1}^N  \int_{\Omega_{p_i}} \frac{g(q)\psi_i(q)}{(f(q) \pm i \epsilon h(q))^2} \: d\mu \nonumber \\
& =  \sum_{i=1}^N \lim_{\epsilon \to 0^+}   \int_{\Omega_{p_i}} \frac{g(q)\psi_i(q)}{(f(q) \pm i \epsilon h(q))^2} \: d\mu  
=\sum_{i=1}^N \lim_{\epsilon \to 0^+}\int_{\Omega_{p_i}} \frac{g(q)\psi_i(q)}{(f(q) \pm i \epsilon)^2} \: d\mu  \nonumber \\
& = 
 \lim_{\epsilon \to 0^+} \sum_{i=1}^N\int_{\Omega_{p_i}} \frac{g(q)\psi_i(q)}{(f(q) \pm i \epsilon)^2} \: d\mu =
 \lim_{\epsilon \to 0^+}\int_{N} \frac{g(q)}{(f(q) \pm i \epsilon)^2} \: d\mu\:.\nonumber
\end{align}
\end{proof}

\begin{proof}[Proof of Lemma \ref{lemmaA2}] Let $\Phi \in \sS(V)$, then $\Phi$ is smooth in coordinates $(\omega, u)\in \bS^2 \times \bR_+$, 
 approaching $u=0$ it vanishes with order $O(u)$ uniformly in the angles, 
$\Phi$ vanishes also if $u>u_0>0$ for some $u_0<1$, finally
$\partial_u \Phi_u$ is bounded.
Define $\chi : \bR \to [0,1]$ such that $\chi \in C^\infty(\bR)$ and $\chi(u)=1$ for $u\geq 2$ whereas $\chi(u) = 0$ for $u\leq 1$. Next define
$\chi_n(u) \doteq \chi(nu)$ for $n=1,2,3 \ldots$. By direct inspection one sees that the sequence of functions $\Phi_n \doteq \chi_n \cdot \Phi$
satisfies the following as $n\to +\infty$: (a) $|\Phi_n| \leq |\Phi|$ and $\Phi_n \to \Phi$ pointwisely, (b) 
$\Phi_n \to \Phi$ in the sense of $L^2(\bS^2\times \bR_+, d\omega du)$, (c) $\partial_u\Phi_n \to \partial_u \Phi$
in the sense of $L^2(\bS^2\times \bR_+, d\omega du)$.
Passing to the $u$-Fourier transforms, from (b) and (c)  and Cauchy-Schwarz inequality, one has that $\sK_\lambda \Phi_n \to \sK_\lambda \Phi$ in the topology of
 $L^2(\bS^2 \times \bR_+, k^nd\omega dk)$
with $n=0,1,2$. This entails the thesis immediately because, by construction, 
$\Phi_n \in C^\infty_0(\bS^2 \times \bR)$, where now $\bR$ is referred to the coordinate $\ell$.
\end{proof}

\begin{proof}[Proof of Lemma \ref{lemmagoursat}] We adopt the same conventions and notations as in the proof 
of Lemma \ref{lemmafourier}. For the fixed $\Phi$ and $\phi$ as in the hypotheses, the sequence  $\{\phi_n\}_{n\in \bN} \subset \sS_m(D)$ 
is individuated by the sequence of corresponding
 smooth compactly-supported Cauchy data on $\Sigma_D$, $\{ (\phi_{n}(1,\cdot),
 \partial_t\phi_{n}(1,\cdot)) \}_{n\in\bN}$ with:
 (a) $\phi_{n}(1,\cdot) \to \phi(1,\cdot)$ in the topology of $H^1(\Sigma_D)$  and 
 (b) $\partial_t\phi_{n}(1,\cdot) \to \partial_t\phi(1,\cdot)$ 
 in the topology of $H^0(\Sigma) \equiv L^2(\Sigma_D,d{x})$.\\
  The former sequence  exists because  $\phi(1,\cdot) \in H^1_0(\Sigma_M)$ and  $C_0^\infty(\Sigma)$ is 
dense in $ H^1_0(\Sigma_D)$.
The latter sequence exists because  $C_0^\infty(\Sigma_D)$ is dense in $ H^0(\Sigma_D)$. 
Passing to the ${\bx}$ Fourier-Plancherel
 transform, the fact that 
$\{ (\phi_{n}(1,\cdot),
 \partial_t\phi_{n}(1,\cdot)) \}_{n\in\bN}$ converges to the pair $(\phi(1,\cdot),
 \partial_t\phi(1,\cdot))$ in the said topologies
  entails that $\widehat{\phi}_{n} \to \widehat{\phi}$ in the topology of $L^2(\bR^3, d{\bk})$
 as it  follows from (\ref{nFMM}). Indeed, denoting by $\cF$ the standard Fourier-Plancherel transformation on $\bR^3$, 
 (\ref{nFMM}) 
and the Cauchy-Schwarz identity  
 imply that the nonnegative number $||\widehat{\phi}_{n} - \widehat{\phi} ||_{L^2(\bR^3, d{\bk})}^2/2$
is bounded by  $$\int_{\bR^3}\sp\sp \sqrt{m^2+ {\bk}^2}
 |\cF(\phi_{n}(1,\cdot))({\bk})  - 
 \cF(\phi(1,\cdot)) ({\bk}) |^2 d{\bk} +   \int_{\bR^3}\sp\sp\frac{|\cF(\partial_t\phi_{n}(1,\cdot))({\bk})  - 
 \cF(\partial_t\phi(1,\cdot))({\bk}) |^2} {\sqrt{m^2+ {\bk}^2}} d{\bk}$$
 $$ +  ||\phi_{n}(1,\cdot)  - 
 \phi(1,\cdot)||_{L^2(\Sigma_D, d{\bx})}\:\:
 ||\partial_t\phi_{n}(1,\cdot)  - \partial_t\phi(1,\cdot)||_{L^2(\Sigma_D, d{\bx})}\to 0\quad \mbox{for $m>0$ as $n\to +\infty$,}$$
 because $(\phi_{n}(1,\cdot),
 \partial_t\phi_{n}(1,\cdot))\to (\phi(1,\cdot),
 \partial_t\phi(1,\cdot))$ in the said Sobolev topologies.  If $m=0$, the only apparent problem concerns the second integral.
 The obstruction can be avoided using a sequence of {\em uniformly bounded}  smooth compactly supported (in $\Sigma_D$) functions 
 $\partial_t\phi_{n}(1,\cdot)$ converging almost everywhere to $\partial_t\phi(1,\cdot))$ (and thus also in $L^2(\Sigma, dx)$ since $\Sigma_D$ has finite measure).
  It exists as a consequence of Luzin and Stone-Weierstrass theorems, and the sequence of the Fourier transforms $\cF(\partial_t\phi_{n}(1,\cdot))$ converges to
  $\cF(\partial_t\phi(1,\cdot))$  both pointwisely (due to Lebesgue dominated convergence theorem) and in the $L^2(\bR^3,dk)$ sense. 
 Since $\cF(\partial_t\phi(1,\cdot))$ and the $\cF(\partial_t\phi_{n}(1,\cdot))$ are uniformly bounded by a constant independent from $n$ and $1/||\bk||$ being $d\bk$-integrable, one easily proves that the second integral above vanishes 
 for $n\to +\infty$.\\
  We have so far obtained that, for $m\geq 0$,
 $\widetilde{\sK}_m L^m_{\widetilde{D}D} (\phi_n) \to  \widetilde{\sK}_m \phi$.  
 To prove that the other convergence property in the thesis  holds,  we notice that  the thesis of  Lemma \ref{lemmabastardo}, in particular
the identity
$
\int_{\bR^3} |\widehat{\phi}(\bk)|^2 d\bk = \langle \sK_\lambda \Phi , \sK_\lambda \Phi  \rangle_{\sH_\lambda}
$
is valid 
 for a generic smooth-in-$\overline{D}$ function  $\phi$ satisfying the Klein-Gordon equation in $M$, such that
 $\widehat{\phi}$ defined as in (\ref{nFMM}) is $L^2(\bR^3,d\bk)$
and $\Phi = L^m_{V\widetilde{D}}\phi$ defines an element
 of $\sS(V)$ as in our case, since just those 
hypotheses were used in the proof. Therefore,  we can re-adapt that proof to our case replacing $\phi$ for $\phi-\phi_n$
and $\Phi$ for $\Phi- L^m_{VD}(\phi_n)$, obtaining
$||\sK_\lambda \Phi - \sK_\lambda L^m_{VD}(\phi_n)||_{\sH_\lambda}^2 = ||\widehat{\phi_n}
 -\widehat{\phi}||_{L^2(\bR^3,d\bk)}^2 \to 0$ for $n\to +\infty$.
The last statement in the lemma is trivially true noticing that  $\widetilde{\sH}_m \supset \sH_m$ are closed subspaces of $L^2(\bR^3,d\bk)$ and thus
 $\sK_m = \widetilde{\sK}_mL^m_{\widetilde{D}D}$ is valid by the definitions of $\widetilde{\sK}_m$ and $\sK_m$. Since  convergence property 
 $\widetilde{\sK}_m L^m_{\widetilde{D}D} \phi_n \to \widetilde{\sK}_m  \phi$ shows that  $\sH_m$ is dense in $\widetilde{\sH}_m$, we conclude that $\widetilde{\sH}_m = \sH_m$.
\end{proof}

\subsection*{On KMS states}
 Comparing Definition 5.3.1 and Proposition 5.3.7 in \cite{BR2} we say that a state $\omega$ on a $C^*$-algebra $\mA$ is  {\em KMS at inverse temperature $\beta\in \bR$}
 with respect to a one-parameter group of 
$*$-automorphisms $\{ \alpha_t\}_{t\in \bR}$ 
if the function
$\bR \ni t \mapsto  F^{(\omega)}_{A,B}(t) \doteq \omega\left(A\alpha_t(B)\right)$ satisfies the  following three requirements for every pair $A,B \in \mA$.
{\bf (a)} It  extends to a continuous complex function $F^{(\omega)}_{A,B} = F^{(\omega)}_{A,B}(z)$
on  $\overline{D_\beta} \doteq \{z \in \bC\:|\:  0 \leq \Imm z \leq \beta\}$  if $\beta\geq 0$, or
$\overline{D_\beta} \doteq \{z \in \bC\:|\: \beta \leq \Imm z \leq 0\}$ if $\beta\leq 0$;
{\bf (b)} the extension is analytic in the interior of $\overline{D_\beta}$;
{\bf (c)}  the {\em KMS condition} holds:
$
F^{(\omega)}_{A,B}(t+i\beta) = \omega\left(\alpha_t(B)A \right)  
$ for all $t \in \bR$.
Another definition is the following.  We say that a state $\omega$ on a $C^*$-algebra $\mA$ is  {\em KMS at inverse temperature $\beta\in \bR$}
 with respect to a one-parameter group of 
$*$-automorphisms $\{ \alpha_t\}_{t\in \bR}$ 
if its GNS triple
$(\cH_\omega, \pi_\omega, \Psi_\omega)$ satisfies the following three requirements. 
{\bf (1)} $\omega$ is $\alpha$-invariant and  the unique unitary group $\bR \ni t \mapsto U_t$
which leaves $\Psi_\omega$ invariant  implementing $\alpha$  (i.e.
$\pi_\omega\left(\alpha_{t}(A)\right) = U_t \pi_\omega(A) U^*_t$  if $A\in \mA$ and $t\in \bR$)
 is strongly continuous. {\bf (2)} $\pi_\omega\left(\mA\right) \Psi_\omega\subset \Dom\left(e^{-\beta H/2}\right)$ where  $e^{itH} = U_t$, for $t\in \bR$, with
  $H$ self-adjoint in $\cH_\omega$. 
 {\bf (3)} There exists  $J: \cH_\omega \to \cH_\omega$ antilinear with $JJ=I$ such that:
$J e^{-it H} = e^{-it H}J$, for all $t\in \bR$,  and
 $e^{-\beta H/2} \pi_\omega(A)\Psi_\omega = J \pi_\omega(A^*)\Psi_\omega$ for all $A\in \mA$.\\
These two definitions are equivalent as we prove here. A state satisfying (a), (b), (c) 
is $\{ \alpha_t\}_{t\in \bR}$-invariant  \cite{BR2} and fulfils 
the conditions (1), (2) and (3) due to Theorem 6.1 in \cite{hug}.
Conversely, consider a state $\omega$ on $\mA$ fulfilling
 the conditions (1), (2) and (3).
When $A$ and $B$ 
are entire analytic elements of $\mA$ (see \cite{BR2}), $\bR \ni t \mapsto F^{(\omega)}_{A,B}(t)$ uniquely 
extends to an analytic function on the whole $\bC$ and thus (a) and (b)  are true.
(1), (2), (3) and $e^{z H} \Psi_\omega = \Psi_\omega$, for all 
$z\in \overline{D_\beta}$ (following from (2) and (3)) entail  (c), too:
$\omega(\alpha_t(B) A) = \langle\Psi_\omega,\: U_t \pi_\omega(B) U^*_t \pi_\omega(A) \Psi_\omega \rangle 
= \langle  \pi_\omega(B^*)\Psi_\omega,\: U_t^* \pi_\omega(A) \Psi_\omega\rangle $
$ = \langle  J U_t^* \pi_\omega(A)\Psi_\omega  ,\: J \pi_\omega(B^*) \Psi_\omega\rangle$\\
$= \langle   U_t^* e^{-\beta H/2}\pi_\omega(A^*) \Psi_\omega ,\: e^{-\beta H/2} \pi_\omega(B) \Psi_\omega\rangle $
$ = \langle\Psi_\omega,\: \pi_\omega(A) e^{i(t+i\beta)H} \pi_\omega(B)e^{-i(t+i\beta)H}  \Psi_\omega\rangle$\\
$= F_{A,B}^{(\omega)}(t+i\beta)$.
The validity of conditions (a), (b) and (c) for  entire analytic elements $A,B \in \mA$ implies the validity  
for all $A,B\in \mA$, as established in \cite{BR2} (compare Definition 5.3.1 and Proposition 5.3.7 therein).


\begin{thebibliography}{References}

\bibitem[Au00]{Aubin} J.-P. Aubin, {\em Applied Functional Analysis}. Second edition,
Wiley-Interscience, New York  (2000), 


\bibitem[BGP07]{BGP} C. B\"ar, N, Ginoux, F., Pf\"affle,  {\em Wave equations on Lorentzian manifolds and quantization} 
 ESI Lectures in Mathematics and Physics,  the European Mathematical Society Publishing (2007)
 
 \bibitem[BJL02]{BJL} H. Baumg\"artel, M. Jurke, and F. Lled\'o, {\em Twisted duality of the CAR-Algebra,} J. Math. Phys. {\bf 43}, (2002), 4158Ð4179.

\bibitem[BMRW09]{BMRW} M. Bischoff, D. Meise, K.-H. Rehren, I. Wagner, {\em  Conformal quantum field theory in various dimensions}
arXiv:0908.3391v1 [math-ph] 


\bibitem[BW75]{BW1} J. Bisognano and E. H. Wichmann, {\em  On the duality condition for a Hermitian quantum field,}  J. Math. Phys. {\bf 16} (1975) 985-1007.

\bibitem[BW76]{BW2} J. Bisognano and E. H. Wichmann,  {\em  On the duality condition for quantum fields,}  J. Math. Phys. {\bf 17} (1976) 303Ð321.

\bibitem[Bo92]{Bor} H. J. Borchers, {\em The CPT-theorem in two-dimensional theories of local observables,} Commun. Math. Phys. {\bf 143}, 315-332 (1992).

\bibitem[Bo00]{Borchers} H. J. Borchers, {\em On revolutionizing quantum field theory with Tomita's modular theory}, J. Math. Phys. {\bf 41}, (2000), no. 6,  3604-3673.

\bibitem[BY99]{BY} H. J. Borchers and J. Yngvason, {\em Modular groups of quantum fields in thermal states},  J. Math. Phys.  {\bf 40},  (1999),  no. 2, 601--624.   

 \bibitem[BR96I]{BR1} O.~Bratteli, D.~W.~Robinson,
{\em Operator Algebras and Quantum Statistical Mechanics}. Vol.1,
Springer Berlin, Germany  (1996)
 
 
\bibitem[BR96II]{BR2} O.~Bratteli, D.~W.~Robinson,
{\em Operator Algebras and Quantum Statistical Mechanics}. Vol. 2,
Springer Berlin, Germany  (1996)

\bibitem[BGL93]{BGL1} R. Brunetti, D. Guido and R. Longo, {\em Modular structure and duality in conformal quantum field theory,}  Comm. Math. Phys.  {\bf 156}, (1993), 201-219.

\bibitem[BGL95]{BGL2} R. Brunetti, D. Guido and R. Longo, {\em Group cohomology, modular theory and space-time symmetries,}  Rev. Math. Phys.  {\bf 7},  (1995), 57-71.

\bibitem[BGL02]{BGL3} R. Brunetti, D. Guido and R. Longo, {\em  Modular localization and Wigner particles,}  Rev. Math. Phys.  {\bf 14},  (2002), 759-785.

\bibitem[Bu77]{Bu} D. Buchholz, {\em On the structure of local quantum fields with non-trivial interaction,} in Proc. Intern. Conf. Operator Algebras, Ideals, and Their Applications in Physics, ed. H. Baumg\"artel (Teubner, 1977), 146Ð153.

\bibitem[BDFS]{BS1} D. Buchholz, O. Dreyer, M. Florig and S. Summers,
{\em Geometric modular action and spacetime symmetry groups}, Rev. Math. Phys.  {\bf 12},  (2000), 475-560.

\bibitem[BMS]{BS2} D. Buchholz, J. Mund and S. Summers, {\em Covariant and quasi-covariant quantum dynamics in Robertson-Walker spacetimes},  Classical Quantum Gravity  {\bf 19},  (2002),  6417-6434.

\bibitem[CH09]{CH} H. Casini and M. Huerta, {\em Reduced density matrix and internal dynamics for multi- component regions,} Class. Quant. Grav. {\bf 26}, (2009) 185005

\bibitem[Co74]{Connes} A. Connes, {\em Une classification des facteurs de type III},  Ann. Ec. Norm. Sup., {\bf 6}, (1974) 415-445.

\bibitem[Du73]{duistermaat} J. J. Duistermaat,  {\em Fourier Integral Operators}. Courant Institute of Mathematical Sciences, New York (1973)

\bibitem[DMP06]{DMP1} C. Dappiaggi, V. Moretti and N. Pinamonti, {\em Rigorous steps towards holography in asymptotically flat spacetimes},
Rev. Math. Phys. {\bf 18}, (2006), 349 [gr-qc/0506069]

\bibitem[DMP09]{DMP2} C. Dappiaggi, V. Moretti and N. Pinamonti, {\em Cosmological horizons and reconstruction of quantum field theories},
Commun. Math. Phys. {\bf 285} (2009), 1129
arXiv:0712.1770 [gr-qc] 

\bibitem[DMP09b]{DMP3} C. Dappiaggi, V. Moretti and N. Pinamonti, {\em Distinguished quantum states in a class of cosmological spacetimes and their Hadamard property},
J. Math. Phys. {\bf 50} (2009), 062304
arXiv:0812.4033 [gr-qc] 

\bibitem[DMP09c]{DMP4} C. Dappiaggi, V. Moretti and N. Pinamonti, {\em Rigorous construction and Hadamard property of the Unruh state in Schwarzschild spacetime},
submitted 	arXiv:0907.1034v1 [gr-qc]


\bibitem[DPP10]{DPP} C. Dappiaggi, N. Pinamonti and M. Porrmann, 
{\em Local causal structures, Hadamard states and the principle of local covariance in quantum field theory.}
In print in Commun. Math. Phys. [arXiv:1001.0858]

\bibitem[DMS97]{CFT} P. Di Francesco, P. Mathieu and David S\'en\'echal, {\em Conformal Field Theory}. Springer, Berlin (1997)
 

\bibitem[FG89]{FG} F. Figliolini and D. Guido, {\em The Tomita operator for the free scalar field,}  Ann. Inst. Henri Poinc\'are Phys. Theor. {\bf 51}, (1989) 419Ð435.

\bibitem[Fr75]{friedlander} F. G. Friedlander, {\em The Wave Equation on a Curved Space-Time},
Cambridge University Press, London (1975)


\bibitem[FJ75]{friedlander2} F. G. Friedlander and M. Joshi,  {\em Introduction to the Theory of Distributions}, 2nd edition
Cambridge University Press, London (2003)


\bibitem[GR95]{Grad}
 I. S. Gradshteyn, I. M. Ryzhik, \emph{Table of Integrals, Series, and Products}, Fifth Edition,
  Academic Press, (1995).


\bibitem[GS94]{GriSjoes} A. Grigis and J. 
Sj\"ostrand, {\em Microlocal Analysis for Differential Operators},  London Mathematical Society, Lecture Notes Series, 
vol. 196, Cambridge University Press, (1994).

\bibitem[Gu08]{Guido} D. Guido, {\em Modular theory for the von Neumann algebras of local quantum physics}, arXiv:0812.1511.


\bibitem[Ha92]{Haag} R.~Haag,
\emph{Local Quantum Physics: Fields, Particles, Algebras},
 Second Revised and Enlarged Edition, Springer (1992).

\bibitem[Ho00]{H00} S. Hollands, {\em Aspects of Quantum Field Theory in Curved Spacetime}. Ph.D.thesis, University of
York, 2000, unpublished.




\bibitem[H\"o94]{Hormander} L. H\"ormander,   \emph{The Analysis of Linear Partial Differential Operators III}
Pseudo-Differential Operators, Springer-Verlag, Berlin (1994)

\bibitem[H\"o96]{Hoe} L. H\"ormander,   \emph{Lectures on nonlinear hyperbolic differential equations},
Math\'ematiques \& Applications, Vol. 26, Springer-Verlag, Berlin (1996)


\bibitem[HL95]{HL} P. D. Hislop and R. Longo,   \emph{Modular Structure of the Local Algebras Associated
with the Free Massless Scalar Field Theory}, Commun. Math. Phys. 84, 71 (1982).


\bibitem[Hu72]{hug} 
N.M. Hugenholtz in {\em Mathematics of Contemporary Physics},
R.F. Streater Editor, Academic Press, London, (1972)

\bibitem[Ki96]{kich} S. Kichenassamy, {\em Nonlinear Wave Equations},
Marcel Dekker, Inc., New York (1996).


\bibitem[KW91]{KW} B.~S.~Kay and R.~M.~Wald,
 {\em Theorems on the uniqueness and thermal properties of stationary, nonsingular, quasifree states on space-times with a bifurcate Killing horizon},
Phys.\ Rept.\  {\bf 207} (1991), 49.


\bibitem[Lle09]{Lledo1} F. Lled\'o, {\em Modular theory by example}, [arXiv:0901.1004]


\bibitem[MLR10]{MLR} P. Martinetti, R. Longo and H. Rehren, {\em Geometric modular action for disjoint intervals and boundary conformal field theory}, Rev. Math. Phys. {\bf 22}, (2010), 331-354.

\bibitem[Mo06]{Mor1}  V.~Moretti, {\em Uniqueness theorem for BMS-invariant states of scalar QFT on the null boundary 
of asymptotically flat spacetimes and bulk-boundary observable algebra correspondence},
 Commun. Math. Phys. {\bf 268} (2006), 727  [ arXiv:gr-qc/0512049].

\bibitem[Mo08]{Mor2} V.~Moretti, {\em Quantum out-states states holographically induced by 
 asymptotic flatness: Invariance under spacetime symmetries, energy positivity and Hadamard property},
 Commun. Math. Phys. {\bf 279} (2008), 31  [ arXiv:gr-qc/0610143].

\bibitem[Pi10]{P} N. Pinamonti, {\em On the initial conditions and solutions of the semi-
classical Einstein equations in a cosmological scenario},
submitted  	arXiv:1001.0864v2 [gr-qc]

\bibitem[RS80]{RS1} M. Reed, B. Simon, {\em Methods of Modern Mathematical Physics. I Functional Analysis}. Revised and enlarged edition, Academic Press, San Diego, (1980)

\bibitem[RS80II]{RS2} M. Reed, B. Simon, {\em Methods of Modern Mathematical Physics. II Fourier Analysis and Self-Adjointness of operators}, Revised and enlarged edition, Academic Press, San Diego, (1980)

\bibitem[Sa06]{Saffary1} T. Saffary, {\em On the generator of massive modular groups}, Lett. Math. Phys. {\bf 77}, (2006) 235-248. 

\bibitem[Sa05]{Saffary2} T. Saffary, {\em Modular action on the massive algebras,} Ph.D. Thesis, Universit\"at Hamburg, (2005), 
http://unith.desy.de/research/aqft/doctoral theses/

\bibitem[Sch03]{Sch03} B. Schroer, {\em Lightfront holography and area density of entropy associated with localization on wedge horizons},
Int. J. Mod. Phys. A {\bf 18} (2003) 1671–1696.

\bibitem[Sch08]{Sch08}  B. Schroer, {\em Area density of localization entropy: I. The case of wedge localization}, Class. Quant. Grav. {\bf 23} (2006) 5227–5248, and Addendum ibid, {\bf 24} (2007) 4239–4249.

\bibitem[SW00]{SW} B. Schroer and H. W. Wiesbrock, {\em  Modular theory and geometry,}  Rev. Math. Phys. {\bf 12} (2000) 139-158

\bibitem[Se82]{Se} G. L. Sewell, {\em Quantum field on manifolds: PCT and gravitationally induced thermal states,} Ann. Phys. {\bf 141} (1982) 201.

\bibitem[Sc08]{schottenloher} M. Schottenloher,  \emph{A Mathematical Introduction to Conformal Field Theory}, second edition 
Springer, (2008).

\bibitem[So08]{sogge} C. D. Sogge,  \emph{Lectures on Non-Linear Wave Equations}, second Edition 
International Press, Boston, (2008).

\bibitem[Su05]{Summers} S. J. Summers, {\em Tomita-Takesaki modular theory}, [arXiv:math-ph/0511034].

\bibitem[Ta01]{TT} M. Takesaki, {\em Theory of Operator Algebras} II, Encyclopedia of Mathematical Sciences, Vol. 125, 
 Springer, New-York, (2001). 

\bibitem[Ta96]{Tay} M. E. Taylor, \emph{Partial Differential Equations}, Applied Mathematical Sciences, vol. 116, 
Springer, New York, (1996).

\bibitem[Vl79]{Vladimirov} 	V. S. Vladimirov,  \emph{Generalized Functions in Mathematical Physics},
 MIR, Moskow, (1979).
 
 \bibitem[Wa84]{Wald1} R.~M.~Wald, \emph{General Relativity}, Chicago University Press, (1984). 

\bibitem[Wa94]{Wald} R.~M.~Wald, \emph{Quantum Field Theory in Curved Space-Time and Black Hole Thermodynamics},
 The University of Chicago Press (1994).
 
\bibitem[Wi93]{Wie1} H. W. Wiesbrock, {\em Half-sided modular inclusions of von Neumann algebras,}  Commun. Math. Phys. {\bf 157}, (1993) 83Ð92 
 
\bibitem[Wi93b]{Wie2} H. W. Wiesbrock, {\em Half-sided modular inclusions of von Neumann alge-
bras,}  Erratum, Commun. Math. Phys. {\bf 184}, (1997) 683Ð685  
  
\bibitem[Wi97]{Wie3}  H. W. Wiesbrock, {\em Symmetries and modular intersections of von Neumann algebras,} Lett. Math. Phys. {\bf 39}, (1997) 203Ð212 


\end{thebibliography}
\end{document}